\documentclass[letterpaper]{article} % DO NOT CHANGE THIS

\usepackage{times}  % DO NOT CHANGE THIS
\usepackage{helvet}  % DO NOT CHANGE THIS
\usepackage{courier}  % DO NOT CHANGE THIS
\usepackage[hyphens]{url}  % DO NOT CHANGE THIS
\usepackage{graphicx} % DO NOT CHANGE THIS
\usepackage[authoryear]{natbib}  % DO NOT CHANGE THIS AND DO NOT ADD ANY OPTIONS TO IT
\usepackage{caption} % DO NOT CHANGE THIS AND DO NOT ADD ANY OPTIONS TO IT
\usepackage{algorithm}
\usepackage{algorithmic}
\usepackage{microtype}
\usepackage{newfloat}
\usepackage{listings}
\DeclareCaptionStyle{ruled}{labelfont=normalfont,labelsep=colon,strut=off} % DO NOT CHANGE THIS
\floatstyle{ruled}
\newfloat{listing}{tb}{lst}{}
\floatname{listing}{Listing}
\usepackage{amsmath}
\usepackage{amsthm}
\usepackage{booktabs}
\usepackage{algorithm}
\usepackage{algorithmic}
\usepackage{enumitem}
\usepackage{times}
\usepackage{xparse}
\usepackage{mathtools}
\usepackage{amssymb}

\usepackage{natbib}
\usepackage{mathptmx}
\usepackage{MnSymbol}

\usepackage{mathtools}

\usepackage{todonotes}

\usepackage[bibliography=common]{apxproof}

\usepackage{microtype}

\usepackage[capitalise,noabbrev]{cleveref}
\usepackage{natbib}

\newtheoremrep{theorem}{Theorem}
\newtheoremrep{example}[theorem]{Example} 
\newtheoremrep{definition}[theorem]{Definition}
\newtheoremrep{claim}[theorem]{Claim}
\newtheoremrep{corollary}[theorem]{Corollary}
\newtheoremrep{lemma}[theorem]{Lemma}
\newtheoremrep{proposition}[theorem]{Proposition}
\newtheoremrep{observation}[theorem]{Observation}

\newtheoremrep{hypothesis}{Hypothesis}

\crefname{example}{Ex.}{plural}
\crefname{section}{Sec.}{plural}
\crefname{theorem}{Thm.}{plural}
\nocite{*}

\usepackage{xspace}
\usepackage{tikz}
\usetikzlibrary{backgrounds}
\usetikzlibrary{calc}
\usetikzlibrary{positioning,fit}
\usepgfmodule{plot}

\def\rConstrHelper(#1,#2,#3,#4){{#1}\ \substack{#2 \\ #3}\ {#4}}

\newcommand{\sss}{{\mathrel{\kern.25em{\sqsubseteq}\kern-.5em \mbox{{\scriptsize *}}\kern.25em}}}
\newcommand{\smallsss}{{\mathrel{\kern.25em{\sqsubseteq}\kern-.4em \mbox{{\scriptsize *}}\kern.25em}}}

\newcommand{\EL}{\ensuremath{{\cal E\!L}}\xspace}

\newcommand{\ALC}{\ensuremath{{\cal ALC}}\xspace}

\newcommand{\NC}{\ensuremath{{\sf N_C}}\xspace}

\newcommand{\NR}{\ensuremath{{\sf N_R}}\xspace}

\newlength{\indxlength}
\newcommand{\Bmc}{\ensuremath{\mathcal{B}}\xspace}
\newcommand{\Cmc}{\ensuremath{\mathcal{C}}\xspace}

\newcommand{\Imc}{\ensuremath{\mathcal{I}}\xspace}
\newcommand{\Jmc}{\ensuremath{\mathcal{J}}\xspace}

\newcommand{\Lmc}{\ensuremath{\mathcal{L}}\xspace}

\newcommand{\Cn}{{\operatorname{Cn}}}

\newcommand{\llang}{\mathcal{L}}

\newcommand{\kernelsymb}{\mathbin{\perp\kern-5pt\perp}}
\newcommand{\finitepwset}{\operatorname{\powerset_{f}}}
\newcommand{\nepwset}{\operatorname{\powerset^{\ast}}}

\newcommand{\modelsofsymb}{\mathsf{mod}}
\newcommand{\modelsof}[1]{\modelsofsymb({#1})}
\newcommand{\modelsofx}[2]{\modelsofsymb_{{#2}}({#1})}

\newcommand{\FRsups}{{\operatorname{MinFRSups}}}
\newcommand{\FRsubs}{{\operatorname{MaxFRSubs}}}

\newcommand{\mSet}{\mathbb{M}}
\newcommand{\mUni}{\mathfrak{M}}

\newcommand{\baseb}{{\Bmc}}

\newcommand{\logsys}{{\Lambda}}
\newcommand{\Evc}{{\operatorname{{\sf evc}}}}
\newcommand{\Rcp}{{\operatorname{{\sf rcp}}}}
\newcommand{\mCon}{{\Evc}}

\newcommand{\mExp}{{\Rcp}}

\newcommand{\mconnm}{eviction}

\newcommand{\interp}{\mathcal{I}}

\newcommand{\psinter}{\rho}

\newcommand{\reduct}[1][\psinter]{\interp_{[\psinter]}}

\newcommand{\FRsetssymb}{{\operatorname{FR}}}
\NewDocumentCommand{\FRsets}{o}{%
	\IfNoValueTF{#1}
	{{\FRsetssymb}}%
	{{\FRsetssymb_{#1}}}%
}

\newcommand{\M}{\mathbb{M}}
\newcommand{\Mplus}{\M^+}
\newcommand{\Mminus}{\M^-}

\newcommand{\vacuousexpansion}{\textbf{vacuous-expansion}\xspace}
\newcommand{\vacuousremoval}{\textbf{vacuous-removal}\xspace}
\newcommand{\lethargy}{\textbf{lethargy}\xspace}
\newcommand{\prudence}{\textbf{circumspection}\xspace}

\newcommand{\revline}{\rev(\baseb,\M^+,\M^-)}

\newcommand{\sel}{\mathrm{sel}}

\newcommand{\naive}{na\"{i}ve\xspace}

\newcommand{\mmin}{\mathsf{min}}

\newcommand{\success}{\textbf{success}\xspace}
\newcommand{\uniformity}{\textbf{uniformity}\xspace}

\newcommand{\dualcon}{\chi_{\logsys}}

\newcommand{\rev}{\ensuremath{{\sf rev}}\xspace}

\newcommand{\decomp}{decomposable\xspace}

\date{}
\title{Model Change for Description Logic Concepts}
\author {
    Ana Ozaki\textsuperscript{\rm 1},
    Jandson S. Ribeiro \textsuperscript{\rm 2}
    \\
     \textsuperscript{\rm 1}University of Oslo, Norway \\
    \textsuperscript{\rm 2}Cardiff University, UK \\
    ana.ozaki@uib.no, ribeiroj@cardiff.ac.uk
}

\newcommand{\cang}{\mathsf{Kangaroo}}
\newcommand{\marsupial}{\mathsf{Marsupial}}
\newcommand{\placental}{\mathsf{Placental}}
\newcommand{\koala}{\mathsf{Koala}}
\newcommand{\platypus}{\mathsf{Platypus}}
\newcommand{\mammal}{\mathsf{Mammal}}

\begin{document}

\maketitle

\begin{abstract}
We consider the problem of modifying a description logic concept in light of models 
represented
%expressed
as pointed interpretations. 
We call this setting \emph{model change}, and distinguish three main kinds of changes: eviction, which consists of only removing models; 
reception, which  incorporates models; and revision, which combines removal with incorporation of models in a single operation. 
We introduce a formal notion of revision and argue that 
%it is more complex than it seems, and 
it does not reduce to a simple combination of eviction and reception,  contrary to intuition. We provide positive and negative results on the compatibility of eviction and reception for $\EL_\bot$ and \ALC description logic concepts and on the compatibility of revision for \ALC concepts. 
\end{abstract}

\section{Introduction}
Keeping beliefs updated is a central problem in knowledge representation, that has been investigated in the context of different logics and applications.
%where is the citation below?
%\cite{bc, ontorev}. 
%extensively studied under several 
%Keeping a knowledge base updated is a recurring problems that appears in knowledge representation and reasoning such .... 
%a central problem that has been investigated in several different content. 
In the main streams of belief change research,  the belief base is finite and the pieces of information expressing how to modify it are expressed as sets of formulae in the underlying logic \citep{Hansson1999,flux,Alchourron1985}. 
%\todo[inline]{elaborate more on finiteness and motivate it, see how we did on the previous AAAI paper, then motivate the change by models, and why revision is important and the challenge. Can we connect the motivation of models with learning from interpretations? Where one needs to learn the concept, but communication with oracles is via models. Also, partial models? Mismatching expressiveness between models and concepts? }

In many scenarios, however, using sets of models to specify the observed change is more suitable than using formulae. This is well studied in the context of \emph{learning from interpretations}~\citep{DeRaedt1997}, where the goal is to find a concise formula that is consistent with models labelled as positive or negative.  In the context of description logic (DL), the process of building an ontology usually goes through stages where the person creating it 
%called the \emph{ontology engineer},
studies possible models of world, discarding models when they are proven false and adding new models previously not considered.
%, based on observations.
The next examples illustrate model change operations for DL concepts.

%...
%\Cref{ex} illustrate a scenario for modifying a base in the light of interpretations. 
%for instance,  information observed 
%
%.. discuss finiteness....
%-  discuss motivation for change, .... 

\begin{example}\label{ex:platypus_evc}
	Araci is visiting a zoo in Australia 
	%for the first time and has very 
	and knows little %knowledge a
	about Australian animals. She 
	knows that a platypus is a mammal that lays eggs,   her  belief on platypus is %
	\begin{align*}
		\platypus \equiv \mammal \sqcap (\exists\mathsf{lays.Egg}).
	\end{align*}
	%
	%She decides to go to the platypuses aquarium in the zoo, where 
	She sees a platypus `$d$' 
	%However, as she 
	but knows nothing about their diet.
	%, she 
	%considers possible that they are herbivores %or not. 
	So, she entertains the two following possible worlds  
	$(\Imc_1,d)$ and $(\Imc_2,d)$ with $\{d,e\}\subseteq \Delta^{\Imc_1}=\Delta^{\Imc_2}$:
	\begin{align*}
		\Imc_1: \  \boxed{\begin{aligned}[t]
				d\in &\ \mathsf{Mammal}^{\Imc_1}  \\   d\not\in & \ \mathsf{Herbivore}^{\Imc_1}  \\
				(d,e)\in &\ \mathsf{lays}^{\Imc_1} \\
				e\in &\ \mathsf{Egg}^{\Imc_1}\\
		\end{aligned}} \hspace{2ex}
		\Imc_2 : \boxed{\begin{aligned}[t]
				d\in &\ \mathsf{Mammal}^{\Imc_2}  \\   d\in & \ \mathsf{Herbivore}^{\Imc_2}  \\
				(d,e)\in &\ \mathsf{lays}^{\Imc_2} \\
				e\in &\ \mathsf{Egg}^{\Imc_2} \\
		\end{aligned}}
	\end{align*}
	%After admiring the small animal, 
	She catches the platypus eating a small insect, which makes Araci retract the pointed interpretation $(\Imc_2,d)$. 
	So, she changes her conceptual beliefs about  platypuses to \begin{align*}
		\platypus \equiv \mammal \sqcap (\exists\mathsf{lays.Egg})\sqcap \neg \mathsf{Herbivore}.
	\end{align*}
\end{example}

\cref{ex:platypus_evc}  illustrates the process of removing a model from the concept, producing a concept closer to the real world. This change operation is called \emph{eviction}~\citep{ouraaai}, and should minimally modify the concept. 
The next example illustrates its dual: the process of adding a model, called \emph{reception}. 

\begin{example}\label{ex:devel_rcp}
	Araci 
	%is curious to know more about the other famous Australian animals: koalas, kangaroos and the infamous Tasmanian devils. She 
	knows that kangaroos and koalas are marsupials, but knows  little about Tasmanian devils. 
	%She believes that
	%they were endemic to the whole Australian continent and that 
	%they are carnivorous mammals that can only be found in Tasmania. 
	So, her  beliefs about marsupials %them 
	are
	%\begin{align*}
	$\marsupial \equiv \koala \sqcup \cang.$
	%\\
	%\mathsf{TasDevil} & \equiv \mammal \sqcap \neg \marsupial
	%\sqcap\mathsf{Carnivore}\sqcap \forall \mathsf{lives}.\mathsf{Tas.}
	%\end{align*}
	%She moves to 
	In the Tasmanian devils section, %curious to catch more information about them, 
	%and 
	she sees a   
	%She believes she will see at least one 
	devil `$d'$' %in the yard, 
	and reads a sign with the information that 
	\emph{tasmanian devils are marsupials}. So, she now admits a world $(\Imc_3,d')$ where
	\begin{align*}
		\Imc_3 : \boxed{\begin{aligned}[t]
				d'\in  \mathsf{TasDevil}^{\Imc_3}, \ 
				d'\in  \mathsf{Carnivore}^{\Imc_3}, \ 
				d'\in \marsupial^{\Imc_3}
		\end{aligned}}
	\end{align*}
	She changes her conceptual beliefs to comply with $(\Imc_3,d')$:
	\begin{align*}
		\marsupial &\equiv \koala \sqcup \cang\sqcup \mathsf{TasDevil}.
		%\\
		%\mathsf{TasDevil} & \equiv \mammal \sqcap \mathsf{Carnivore}\sqcap \forall \mathsf{lives}.\mathsf{Tas.}
	\end{align*}
	
\end{example}

In both examples, Araci minimally modified her concepts. For platypus, she modified only concepts related to the diet, while for the Tasmanian devil, she modified only concepts related to marsupials and devils; no further changes relate to kangaroos or koalas were carried out.  
Ensuring minimal change in this setting is a challenge, as it is not always possible to retract or incorporate the single input model \citep{ouraaai}. %For instance, in \cref{ex:devel_rcp}, all other interpretations that comply with .. are also absorbed in the process, while .... 
A third and more complex kind of operation is illustrated in \cref{ex:koala_rev}, where one must add and remove models in a single step.

\begin{example}\label{ex:koala_rev}
	Araci believes that 
	%marsupials are not placental. She also knows that koalas, just like kangaroos, are marsupials, so she concludes that 
	koalas are marsupial mammals that are not placental. %She also knows that koalas live on trees (arboreal animals). 
	So her concept on koalas is: 
	$$ \koala \equiv \mammal\sqcap\marsupial \sqcap   \neg \placental.$$
	She sees a koala   `$d''$' in the zoo, so she considers  $(\Imc_4,d'')$ possible, but not $(\Imc_5,d'')$.  
	\begin{align*}
		\Imc_4 : \boxed{\begin{aligned}[t]
				d'' \in & \ \mammal^{\Imc_4}  \\
				d'' \in & \  \marsupial^{\Imc_4} \\
				d'' \notin & \ \placental^{\Imc_4}
		\end{aligned}}\hspace*{2ex} 
		\Imc_5 : \boxed{\begin{aligned}[t]
				d'' \in & \ \mammal^{\Imc_5}  \\
				d'' \in & \  \marsupial^{\Imc_5} \\
				d'' \in & \ \placental^{\Imc_5} 
		\end{aligned}}
	\end{align*}
	She reads a sign informing that koalas are actually placental. So she is compelled to comply with $(\Imc_5,d'')$, whereas retracting $(\Imc_4,d'')$.
	Therefore, she changes her belief to 
	$$ \koala \equiv \mammal\sqcap\marsupial  \sqcap \placental .$$
	%    \todo{Ana: check}
\end{example}

In  \cref{ex:koala_rev}, the model $(\Imc_4,d'')$  had to be removed, while  $(\Imc_5,d'')$ is added. We call this more complex operation a \emph{revision}. After providing preliminaries (\cref{sec:prelidl}) and recalling eviction and reception (\cref{sec:model_change}), we present our main contributions, which are: 
\begin{itemize}
	\item a theoretical study on eviction and reception for DL concepts in $\ALC$ and $\EL_\bot$ (\cref{sec:dlconcept_evc_rcp}); and
	\item  the introduction of the notion of model revision (\cref{sec:modelrev}),   with some results for DL concepts (\cref{sec:revdl}).    
\end{itemize}
Contrary to intuition, revision \emph{does not} correspond to a serial combination of eviction and reception, as we show in this paper. % that is, it dos not collapses to adding .. an then removing. 
%In fact, in the second step, NY removing $I2$, one can eliminate ... which should be guaranteed. This issue is related to several urdles which we identify in this paper. 
In general, it may well be impossible to add and remove exactly and only the models of the input. 

\medskip
\noindent
\textbf{Related Work} Our contribution is closest to the work by \cite{ouraaai}, with the main differences being that they neither consider revision nor DL concepts. Although their results on eviction and reception of DL ontologies inspired some of our proofs, the settings are considerably different, with the DL concept case that we present here being arguably more natural to the process of modelling concepts, which is  fundamental  for building ontologies. In this line, there has been work on computing least common subsumers, which generalize a set of concepts with minimal change~\cite{DBLP:journals/japll/BaaderST07}.
Recent work considered learnability of DL concept expressions with pointed interpretations labelled as positive and negative~\citep{DBLP:conf/ijcai/CateKO24}. One can see the positive and negative pointed interpretations in their work as the sets of models to be added and removed in our setting. The difference from their work to ours is that here we focus on the existence of operators following minimal change rationality postulates, while they focus on the existence of sets of labelled models that characterize a DL concept, that is, that can be used to distinguish a particular concept from all the others in a DL language for concepts. Ontology learning from interpretations has been investigated by~\cite{DBLP:conf/ilp/KlarmanB15}. %\todo{cite more works}
Other works investigated learnability of DL concepts in a data retrieval setting~\citep{DBLP:conf/ijcai/FunkJL21,DBLP:conf/ijcai/FunkJLPW19}, using inductive logic programming~\citep{DBLP:journals/jmlr/Lehmann09,DBLP:conf/ilp/FanizzidE08,DBLP:conf/ilp/LehmannH09,DBLP:journals/ml/LehmannH10}, and with counterfactuals~\citep{DBLP:journals/apin/IannonePF07}. 
We also point out works relating learning, epistemic logic, and belief revision 
~\citep{DBLP:journals/sLogica/BaltagGS19,DBLP:journals/jlap/BaltagGOSS19,DBLP:conf/jelia/OzakiT19,DBLP:conf/ijcai/SchwindIKM25}. 
%\todo[inline]{put contributions and outlines}
% We introduce a novel class of revision operators based on symmetric difference between models, and we show that revision is characterised precisely by symmetric difference, that is, minimal change is characterised by symmetric difference.   
% conencting eviction an
% DLCs that are compl

%For instance, ..... Another main hurdle is related to the finiteness of the produced knowledge base, .... 

% ===========

% ... discuss compatibility....

% We investigate model change on Description Logic Concepts. 
% We identify the comtibulity fo such logcs with eviction and reception. 
% We introduce model revision, whose purpose is to ..... 
% Model revision is a complex operation, and we show hat cannot be defined by a simple combination of eviction and reception. 
% We introduce suitable minimal change rationality postulate for revision, and investigate the compatibility on DLC. We identify a tight connection between revision and the other two kinds of operators: eviction and reception. 

%\input{sections/2_preliminaries}
\section{Preliminaries}
\label{sec:prelidl}
The power set of a set $A$ is denoted by $\powerset(A)$,  while the set of all finite subsets of $A$ is denoted by $\finitepwset(A)$. 
Given a pre-order $\mathbin{\leqslant} \mathbin{\subseteq} \mathcal{D} \times \mathcal{D}$ on a domain $\mathcal{D}$, and a set $A \subseteq \mathcal{D},$ the set of all maximal and minimal elements of $A$ w.r.t. $\leqslant$ are respectively %\todo{essas definicoes precisam de correcao (!?)}
%\begin{align*}
$\mathsf{max}_{\leqslant}(A) = \{ x \in A \mid \mbox{ for all } y \in A \mbox{ if } x \leqslant y, \mbox{ then } y \leqslant x\}$, and  %    \\
$\mathsf{min}_{\leqslant}(A) = \{ x \in A \mid \mbox{ for all } y \in A \mbox{ if } y \leqslant x, \mbox{ then } x \leqslant y\}$.
%\end{align*}
%
We write \(\nepwset(A)\) to denote the non-empty subsets of \(A\).
Following \citet{Aiguier2018,DelgrandePW18}, and \citet{ouraaai}, we use satisfaction systems to define logics.
%
%\begin{definition}\cite{Aiguier2018}%
%\label{def:logsys}
A \emph{satisfaction system} is a triple \(\logsys=(\llang, \mUni, \models)\),  where 
\(\llang\) is a non-empty (possibly countably infinite)
%\todo{ countable set called }
language 
,  \(\mUni\) is a set of models, 
%the  set of all possible models,  
%also
%called interpretations,  
%that determine the meaning of the sentences in
%$\llang$,  
and \(\models \subseteq \mUni \times \powerset(\llang)\) is a relation, called the \emph{satisfaction} relation, which relates  
%\todo{discuss this, the pairs should be (model, base)} 
models to subsets of the language. We use the infix notation $M \models \baseb$ as a shorthand for $(M,\baseb) \in \ \models$ and say that $M$ \emph{satisfies} $\baseb$.
%\todo[color=green]{dei uma reescrita em algumas partes que nao estavam claras}
%all pairs 
%\((M, \baseb) \) means that a model \(M\) satisfies  $\baseb$. 
Every subset of $\llang$ is called a \emph{base} (which can be finite or infinite), denoted $\baseb$. 
%Note that bases can be infinite §. %We call a \emph{base} (that is, a subset of \(\llang\)),
% meaning that 
% of bases $\baseb$ (subsets of \(\llang\)) %\todo{bases were not defined yet, rewrite with a set of sentences... \\RG: tried solving this by adding the text in parenthesis} 
% and 
% models \(M \in \mUni\) 
%such that \(M\) satisfies \(\baseb\) (i.e., \(M \models\baseb\)).
We  denote by \(\modelsofx{\baseb}{\logsys}\) the set 
$\{M \in \mathfrak{M} \mid %\forall \varphi \in \baseb,
M \models
\baseb\}.$
We write  \(\modelsof{\baseb}\)  when the satisfaction system is clear from the context.
%    Let \(\llang\) be a language, \(\mUni\) set of models and \(\models
%    : \mUni \to 2^\llang\) a satisfaction relation.  We call the triple
%     a satisfaction system.
%\end{definition}
%
%Looking at a logic simply as a satisfaction system allows us to explore its
%properties without making assumptions about the language or putting constraints
%upon the logic's entailment relation.  %should satisfy.

Satisfaction systems 
%allow us to be more flexible and precise regarding the precise scope of the operations and constructions we define.
%This view also 
facilitate the generalisation of some results that do not depend on certain properties of the consequence relation of the logic. 
A  set of models \(\mSet \subseteq \mUni\) 
%\todo{dizer queue conjunto de models sera denotado por ??}
within $\logsys$ is finitely representable iff there is
$\baseb \in \finitepwset(\llang)$ such that $\modelsof{\baseb} = \mSet$. 
Let $\FRsets(\logsys)$ denote 
%the
%collection of 
all \emph{finitely representable sets of models in $\logsys$}, i.e.,
\[
\FRsets(\logsys) = \{\mSet \subseteq \mUni \mid \exists \baseb \in
\finitepwset(\llang) : \modelsof{\baseb} = \mSet\}.
\]%
Given a set $\mSet$ of models, the  greatest finitely-representable subsets of $\mSet$ and the least  finitely-representable supersets of $\mSet$ are given respectively by 
%from a set $\mSet$ are give by 
\begin{align*}
	\FRsubs(\mSet, \logsys) &= \mathsf{max}_{\subseteq}(\{ \mSet' \in \FRsets(\logsys) \mid
	\mSet' \subseteq \mSet \}),\\
	\FRsups(\mSet, \logsys) &= \mmin_{\subseteq} (\{ \mSet' \in \FRsets(\logsys) \mid
	\mSet \subseteq \mSet' \}).
\end{align*}
% On the converse, 
% the least finitely-representable supersets of  $\mSet$ are given by 
% \begin{align*}
	%     \FRsups(\mSet, \logsys) = \mmin_{\subseteq} (\{ \mSet' \in \FRsets(\logsys) \mid
	%             \mSet \subseteq \mSet' \}).
	%     \end{align*}

We say that a set of formulae $\baseb \subseteq \llang$ is finitely
representable iff there is   $\baseb' \in \finitepwset(\llang)$ with $\modelsof{\baseb} = \modelsof{\baseb'}$.  %Given a satisfaction system  \(\logsys = (\llang, \mUni, \models)\),
%We omit the subscript whenever \(\baseclass = \finitepwset(\llang)\). 
%Additionally, 
We write $\times$ for the
Cartesian product of two sets.
Also, we denote the logical closure of a base %\(\baseb\) 
in a satisfaction system \(\logsys\) by \(\Cn_\logsys\), omitting the subscript when clear from the context.

%\todo[inline]{notation}
% \todo[inline]{ sistema de sat}

%\subsection{Description Logic Concepts}\label{sec:prelidl}
\medskip
\noindent
\textbf{DL Concepts} Let \NC{} and \NR{} be countable and pairwise disjoint sets of concept %\todo{ana: try to shorten}
names and role names,   respectively. In this work,   \NC{} and \NR{} can be finite or infinite. We explicitly indicate when these sets, called \emph{signature}, are finite (otherwise, they are assumed to be infinite).
%\todo{define EL and EL bottom}
\EL
\emph{concepts} are built according to the rule:
%
%\[
$C,D ::= \top \mid A %\mid \neg C 
\mid (C \sqcap D) \mid (\exists r.C)$, 
%\]
%
where $A \in \NC$ and $r\in\NR$.  $\EL_\bot$ concepts extend \EL by allowing $\bot$ (interpreted as the empty set). $\ALC$
{concepts} extend \EL concepts with the rule %(and \EL$_\bot$) by allowing
$\neg C$ (recall that $C\sqcap \neg C$ is equivalent to $\bot$, so \ALC extends $\EL_\bot$).
\begin{toappendix}
	Here we introduce additional notions about the syntax and semantics of DLs, basic theoretical tools, and results needed in the proofs of subsequent sections.
	
	When writing \ALC concepts, we may use w.l.o.g. disjunctions ($C \sqcup D\equiv \neg(\neg C\sqcap \neg D)$) and universal quantification ($\forall r.C  \equiv \neg (\exists r.\neg C)$) as part of the \ALC syntax.
	%
	%\(\ALCfo\) formulae are    expressions $\varphi$ of the form 
	%\begin{align*}
	%\[\alpha ::= C(a) \mid r(a,b) \mid (C =\top) %\quad 
	%\]
	%\[
	%\varphi  ::= \alpha \mid \neg(\varphi) \mid (\varphi \wedge \varphi)
	%\]
	%
	%\end{align*}
	%where $C$ is an \ALC concept,  $a, b \in \NI$, and $r \in \NR$. %\footnote{
		We may omit parentheses if there is no risk of confusion. A concept $C$ is a subconcept of a concept $D$ if it occurs in $D$ (by this, every concept is a subconcept of itself). We define the (role) \emph{depth} of a concept inductively.
		The {depth} of a concept without any roles is $0$. 
		The depth of a concept of the form $(\exists r.C)$ or $(\forall r.C)$ is the depth of $C$ plus $1$. 
		The  depth of $(C\sqcap D)$ or $(C\sqcup D)$ is the maximum of the depths of $C,D$. 
	\end{toappendix}
	We may write $\exists r^n.\top$, with $n\in\mathbb{N}$, as a shorthand for the nesting of $n$ existential quantifiers (that is, $\exists r^{n+1}.\top=\exists r.(\exists r^{n}.\top)$) and $\exists r^0.\top=\top$.
	%The  depth of a concept is the maximum of the  depths of its subconcepts.
	%The usual  concept inclusions $C \sqsubseteq D$ can be expressed 
	%with  $\top \sqsubseteq \neg C\sqcup D$ and $\neg C\sqcup D \sqsubseteq \top$,
	%which is $(\neg C\sqcup D = \top)$. %\textcolor{black}{We also 
		%consider in this work standard abbreviations
		%for other operators such as $\bot$ (interpreted as the empty set or, equivalently,
		%as the negation of $\top$). }
	%\emph{Assertions} are expressions of the form 
	%$r(a,b)$ and $A(a)$, with $r\in\NR$, $a,b\in\NI$, and $A\in\NC$.
	%Whenever we speak of an $\EL_\bot$  finite  base
	%we mean a finite set of concept inclusions and assertions built from $\EL_\bot$ concepts. The same holds for \EL and \ALC.
	
	\medskip
	\noindent
	\textbf{Semantics} The semantics of \EL, $\EL_\bot$, and \ALC concepts is defined using pointed interpretations~\citep{Baader2017,Gabbay2003a}. An interpretation \Imc is a pair 
	$(\Delta^\Imc,\cdot^\Imc)$, where
	$\Delta^\Imc$ is a non-empty set, called the \emph{domain}, and $\cdot^\Imc$ is a function that maps every $A\in\NC$ to a subset of $\Delta^\Imc$ and every $r\in\NR$ to a subset of $\Delta^\Imc\times\Delta^\Imc$. It is \emph{finite} if $\Delta^\Imc$ is finite.
	A \emph{pointed interpretation} is a pair $(\Imc,d)$ where $\Imc=(\cdot^\Imc,\Delta^\Imc)$ is an interpretation and $d\in \Delta^\Imc$.
	\begin{toappendix}

		We say that an interpretation \Jmc with $\Delta^{\Jmc}\subseteq\Delta^\Imc$
		is a \emph{subinterpretation} of \Imc
		if
		$A^\Jmc=A^\Imc\cap\Delta^\Jmc$ for 
		all $A\in\NC$ and $r^\Jmc=r^\Imc\cap(\Delta^\Jmc\times\Delta^\Jmc)$ for 
		all $r\in\NR$. Given an interpretation $\Imc$ and a role name $r\in\NR$, we say that $e\in\Delta^\Imc$ is an \emph{$r$-successor} of $d\in\Delta^\Imc$ in   \Imc if
		$(d,e)\in r^\Imc$.
	\end{toappendix}
	A pointed interpretation $(\Imc,d)$ \emph{satisfies} a concept $C$ iff
	$d\in C^\Imc$. We define tree-shaped pointed interpretations using the notion of unfolding \cite{Dummett1959-DUMMLB-2} (see also~\cite{DBLP:conf/ijcai/KonevLWZ16}).
	%\todo{Ana: move to appendix unused def}
	We say that a concept $C$ \emph{entails} a concept $D$ if for all pointed interpretations $(\Imc,d)$ (over the signature), $d \in C^\Imc$
	implies $d \in D^\Imc$. Two concepts $C,D$ are \emph{equivalent}, written $C\equiv D$, iff $C$ entails $D$ and $D$ entails $C$. 
	%We use standard notions of canonical model, homomorphisms, isomorphisms, unfolding, and bisimulations, given in the appendix. 
	%inductively defined as expected. The role depth of a concept without any roles is $0$. Suppose the role depth of $C$ is $n$ then the 
	%\todo[inline]{introduce finite interpretation}
	%\todo[inline]{introduce def of depth of a concept, depth of a pointed tree-shaped finite interpretation}
	%\todo[inline]{introduce the notion of subconcept rooted at depth $k$}
	%\todo[inline]{entailment of concepts}
	%\todo[inline]{logical equivalence of concepts}
	\begin{toappendix}

		\paragraph{Canonical Model}
		%\todo{in the appendix it appears proofs but these are def}
		Given a satisfiable $\EL_\bot$ concept $D$,
		we inductively define the tree-shaped pointed interpretation $(\Imc_D,d_D)$, called the \emph{canonical model} of $D$, with the root denoted $d_D$, as follows. 
		%\todo{add the case $C=\top$}
		When $D$ is $\top$, we define $(\Imc_{\top},d_\top)$ as the pointed interpretation with $\Delta^{\Imc_{\top}}:=\{d_{\top}\}$  and all
		concept and role names interpreted as the empty set.
		For $D$ a concept name $A\in \NC$ we define $(\Imc_A,d_A)$ as the pointed interpretation with $\Delta^{\Imc_A}:=\{d_A\}$,
		$A^{\Imc_A}:=\{d_A\}$, and all
		other concept and role names interpreted as the empty set.
		For $D=\exists r.C$, we define
		$(\Imc_D,d_D)$ as the pointed interpretation with $\Delta^{\Imc_D}:=\{d_D\}\cup \Delta^{\Imc_C}$. 
		All concept and role name interpretations are as for
		$(\Imc_C,d_C)$ and we add $(d_D,d_C)$ to $r^{\Imc_D}$, 
		and assume $d_D$ is fresh (i.e., it is not in $\Delta^{\Imc_C}$).  Finally, for $D=D_1\sqcap D_2$
		we define $\Delta^{\Imc_D}:=\Delta^{\Imc_{D_1}}\cup (\Delta^{\Imc_{D_2}}\setminus\{d_{D_2}\})$,
		assuming $\Delta^{\Imc_{D_1}}$ and $\Delta^{\Imc_{D_2}}$
		are disjoint, and
		with
		all concept and role name interpretations as in
		$(\Imc_{D_1},d_{D_1})$ and $(\Imc_{D_2},d_{D_2})$, except that
		we connect $d_{D_1}$ with the elements in $\Delta^{\Imc_{D_2}}$
		in the same way as $d_{D_2}$ is connected to these elements in $(\Imc_{D_2},d_{D_2})$. In other words, we 
		{identify} $d_{D_1}$ with   $d_{D_2}$ in $\Delta^{\Imc_{D_2}}$. 
		\paragraph{Homomorphism and Isomorphism} 
		Let 
		$(\Imc,d_0)$ and $(\Jmc,e_0)$ be two pointed interpretations. 
		A \emph{homomorphism} from $(\Imc,d_0)$ to $(\Jmc,e_0)$
		is a function  $h:\Delta^{\Imc}\to\Delta^\Jmc$
		that   satisfies: (i) $h(d_0)=e_0$; % and the following:
		%\begin{itemize}
		%\item  ;
		%\item 
		(ii)   for all $A\in\NC$, if $d\in A^\Imc$  then $h(d)\in A^\Jmc$; and (iii)
		%   \item 
		for all $r\in\NR$, if $(d,d')\in r^\Imc$   then   $(h(d),h(d'))\in r^\Jmc$.
		An \emph{isomorphism} between $(\Imc,d_0)$ and $(\Jmc,e_0)$
		is a bijective function  $h:\Delta^{\Imc}\to\Delta^\Jmc$
		that   satisfies: (i) $h(d_0)=e_0$; % and the following:
		%\begin{itemize}
		%\item  ;
		%\item 
		(ii)   for all $A\in\NC$,  $d\in A^\Imc$  iff $h(d)\in A^\Jmc$; and (iii)
		%   \item 
		for all $r\in\NR$,  $(d,d')\in r^\Imc$   iff   $(h(d),h(d'))\in r^\Jmc$.
		
		%\end{itemize}
		
		%\todo{instanciar o framework para esse caso}

		\paragraph{Unfolding and Bisimulations} We use the classical notion of \emph{unfolding}~\cite{Dummett1959-DUMMLB-2} (see e.g.~\cite{DBLP:conf/ijcai/KonevLWZ16} for a definition in the context of description logic) in some technical lemmas.
		A \emph{path} $p$ in an interpretation \Imc is a sequence $d_0 \cdot r_0 \ldots r_{n-1}\cdot d_n$ such that $d_i \in \Delta^\Imc$, $r_i \in \NR$, and $(d_i,d_{i+1})\in r^\Imc_i$  for all $0\leq i < n$. The \emph{length} of a path $d_0 \cdot r_0 \ldots r_{n-1}\cdot d_n$ is $n$ (repetitions are counted).
		By $\mathsf{tail}(p)$ we denote the final element of $p$. 
		%If \Imc is a directed tree interpretation, then for every d ∈ ∆I there exists a unique path p starting from the root ρI of p such that tail(p) = d. 
		Given a pointed interpretation $(\Imc,d)$, the \emph{unfolding}
		of $(\Imc,d)$, denoted $(\Imc^u,d)$,
		is the pointed interpretation defined as follows:
		\begin{itemize}
			\item $\Delta^{\Imc^u}$ is the set of all paths that start at $d$;
			\item $p \in A^{\Imc^u}$ if $\mathsf{tail}(p)\in A^{\Imc}$;
			\item $r^{\Imc^u}:=\{(p,p\cdot r\cdot f)\mid  p,p\cdot r\cdot f\in \Delta^{\Imc^u} (\mathsf{tail}(p),f)\in r^{\Imc}\}$.
		\end{itemize}
		A pointed interpretation $(\Imc,d)$ is \emph{tree-shaped} if  $(\Imc,d)$
		and $(\Imc^u,d)$ are isomorphic. 
		Each finite pointed tree-shaped interpretation $(\Imc,d)$ is isomorphic to the canonical model of an \EL concept $C$. The depth of $(\Imc,d)$ is the depth of $C$. 
		
		\begin{definition}[Bisimulation]\label{def:bisimulaiton}
			Let \Imc and \Jmc be   interpretations over   a signature~$\Sigma$. A
			\emph{bisimulation} between \Imc and \Jmc is a non-empty relation $Z \subseteq\Delta^\Imc \times \Delta^\Jmc$ such that, for every
			$d\in \Delta^\Imc$ and $e\in \Delta^\Jmc$ with $(d,e) \in Z$, every concept name $A$ in $\Sigma$, and every role name $r$ in $\Sigma$: 
			\begin{itemize}
				\item (atom) $d \in A^\Imc$ iff $e \in A^\Jmc$; 
				\item (forth)
				if $(d,d') \in r^\Imc$ then there is $e' \in \Delta^\Jmc$ such that $(e,e') \in r^\Jmc$ and $(d',e') \in Z$;
				%and 
				\item (back) if $(e,e') \in r^\Jmc$ then there is $d' \in \Delta^\Imc$ such that $(d,d') \in r^\Imc$ and
				$(d',e') \in Z$. 
			\end{itemize}
		\end{definition}
		%Given  $(\Imc,d)$ and  $(\Jmc,e)$, 
		We say that pointed interpretations 
		$(\Imc,d)$ and $(\Jmc,e)$ are \emph{bisimilar} if there is a bisimulation between \Imc and \Jmc that contains  $(d,e)$. 
		By definition, there is a bisimulation
		between $\Imc$ and $\Imc^u$
		including the pair $(d,d)$ (see, e.g., page 15 in~\cite{DBLP:books/el/07/BBW2007}).

		To prove \cref{thm:alcnotrecpfinitetree}, we use  the following technical results.
		%\begin{toappendix}
		\begin{lemma}[Chapter~3, Lemma~9 in~\cite{DBLP:books/el/07/BBW2007}]\label{lem:bisimulation}
			Given a pointed interpretation $(\Imc,d)$, the \emph{unfolding}  $(\Imc^u,d)$
			of $(\Imc,d)$ is a tree-shaped pointed interpretation such that:
			\begin{itemize}
				\item   $(\Imc,d)$ and $(\Imc^u,d)$ are bisimilar;
				\item for all $\ALC$ ${\text{concepts}}$ $C$,
				we have that $d\in C^\Imc$ iff $d\in C^{\Imc^u}$. 
			\end{itemize}
		\end{lemma}

		\begin{definition}[$k$-bisimulations]
			Let $\Imc,\Jmc$ be two interpretations (over a  signature). Let $Z_0\subseteq\ldots\subseteq Z_k\subseteq\Delta^{\Imc}\times\Delta^{\Jmc}$ be binary relations on the domains of $\Imc$ and $\Jmc$. 
			We say that the tuple $(Z_0,\ldots,Z_k)$ of non-empty binary relations is a \emph{$k$-bisimulation}  between $\Imc$ and $\Jmc$ if, for every $0\leq i\leq k$ and $(d,e)\in Z_i$, the following  conditions  hold:
			\begin{itemize}
				\item  $d\in A^{\Imc}$ iff $e\in A^{\Jmc}$ for all $A\in\NC$;
				\item if $i<k$, then for all $r\in\NR$ and $(d,d')\in r^{\Imc}$, there is $(e,e') \in r^{\Jmc}$   with $(d',e')\in Z_{i+1}$;
				\item  symmetric to the (forth).
			\end{itemize}
			
		\end{definition}
		We denote by $\Imc^u_{|k}$ the subinterpretation of $\Imc^u$ induced by
		removing all paths of length larger than 
		$k$ from $\Delta^{\Imc^u}$.
		By definition, there is a $k$-bisimulation
		between $\Imc^u_{|k}$ and $\Imc^u$
		including the pair $(d,d)$ and,
		by \cref{lem:bisimulation}, there is a simulation between \Imc and $\Imc^u$. So the following lemma holds.
		%\cite{DBLP:books/cu/BlackburnRV01,DBLP:books/el/07/GorankoO07} (see also Proposition 2.6 in~\cite{DBLP:conf/aiml/BolanderB24}).
		
		\begin{lemma}[\cite{DBLP:books/cu/BlackburnRV01,DBLP:books/el/07/GorankoO07}(see also Proposition 2.6 in~\cite{DBLP:conf/aiml/BolanderB24})]\label{lem:kbisimulation}
			For all $\ALC$ ${\text{concepts}}$ $C$ of depth $k$, and all pointed interpretations
			$(\Imc,d)$, we have that 
			$d\in C^\Imc$ iff $d\in C^{\Imc^u_{|k}}$.
		\end{lemma}
		
		%\todo[inline]{relembrar q \Imc e isomorphico ao unfolding de \Imc}
		
		\begin{lemma}\label{lem:khbisimulation}[\cite{DBLP:books/cu/BlackburnRV01,DBLP:books/el/07/GorankoO07}]
			If $(Z_0,\ldots,Z_k)$ is a $k$-bisimulation between
			two interpretations $\Imc$ and $\Imc'$ 
			then, for all $0\leq k'\leq k$, $(Z_0,\ldots,Z_{k'})$ is a
			$k'$-bisimulation between
			$\Imc$ and $\Imc'$.   
		\end{lemma}
		\begin{lemma}[\cite{DBLP:books/cu/BlackburnRV01,DBLP:books/el/07/GorankoO07}]\label{lem:kinvariance}
			For all pointed interpretations
			$(\Imc,d)$ and $(\Jmc,d')$, the following is equivalent: (1)
			there is a $k$-bisimulation between $\Imc$ and $\Jmc$ including $(d,d')$ and (2)
			for all $\ALC$ ${\text{concepts}}$ $C$ of depth up to  $k$, $d\in C^\Imc$ iff
			$d'\in C^\Jmc$. 
		\end{lemma}
		
	\end{toappendix}
	%\paragraph{Chains} %In this paper, 
	%In some of  our proofs and examples, we  use interpretations representing finite and infinite chains of role names. We define these interpretations as follows. 
	Given a fixed but arbitrary %$a\in\NI$ and 
	$r\in\NR$, we define 
	%\begin{itemize}
	%	\item 
	\(M^n=(\mathbb{N},\cdot^{M^n})\)
	where 
	\(r^{M^n} = \{(i,i+1)\mid i\in \mathbb{N}, 0\leq i < n\}\)
	%and $a^{M_n}=0$, 
	and similarly
	%\item 
	\(M^{\infty}=(\mathbb{N},\cdot^{M^{\infty}})\)
	where \(r^{M^{\infty}} = \{(i,i+1)\mid i\in \mathbb{N}\}.\)
	%\todo{  $r$-successor}
	%Given a finite tree-shaped pointed interpretation $(\Imc,d)$, 
	%we denote by $C_e$ the \EL concept
	%`rooted' in $e\in\Delta^\Imc$. That is, let $(\Imc',e)$ be the subinterpretation of $(\Imc,d)$ induced by taking the subset of $\Delta^\Imc$ that consist of $e$ and all elements in $\Delta^\Imc$ reachable from $e$.

	\medskip
	\noindent
	\textbf{DL concepts in Satisfaction Systems}
	In a satisfaction system, we use the term `base' for a subset of formulas in a logic language. 
	%In the DL concepts setting, 
	We treat  a set of concepts and a concept formed by the  conjunction of the elements of the set interchangeably. So we may refer to a base as a concept or as a finite set of concepts (in the latter, we mean the concept formed by the conjunction). 
	The notion of a `model' in a satisfaction system corresponds to the notion of a pointed interpretation.

\section{Model Reception and Eviction}\label{sec:model_change}
%\todo[color=red]{put in 1.5 pages}
This section addresses the problem of modifying a finite base in light of a set of models, as illustrated in \cref{ex:platypus_evc} and \cref{ex:devel_rcp}.
%\todo{recheck if example remains after changing intro} 
%in the introduction. 
We call such kinds of operators \emph{model change operators}. 
\citet{ouraaai} distinguished two main primitive kinds of model change operators: 
\begin{description}
	\item[] \textit{eviction}: 
	remove a set $\mSet$ of models from a base  $\baseb$, %that occurs in a given set $\mathbb{M}$ of models, 
	that is, 
	turn $\baseb$ into a base $\baseb'$ whose models are not in $\mathbb{M}$;
	\item[] \textit{reception}: incorporate all models from the set $\mSet$ into a base $\baseb$, that is, turn $\baseb$ into a base $\baseb'$ such that all   models in $\mathbb{M}$ satisfy $\baseb'$.
\end{description}

In some scenarios, not all sets of models need to be taken into account. 
For instance, some DLs  have the finite model property
or %, while others have 
the tree-shaped property. 
%As in such logics, these models are deemed the most plausible ones;
So, it makes sense to also consider model change operators that only take into account such classes of models.
%should consider   such models. 
%We call such spaces of models a class of models. 
Formally, a class of models on a satisfaction system $\logsys = (\llang, \mUni, \models)$ is a set $\Cmc \subseteq \powerset(\mUni)$ with sets of interpretations.  
Model change operators are, therefore, defined on a given class of models. 
\begin{definition}    
	%Let $\Cmc$ be a class of models on a satisfaction system. 
	A model change operator in a class $\Cmc$ of models is a function  $\circ: \finitepwset(\llang) \times \Cmc \to \finitepwset(\llang)$, mapping each finite base $\baseb$ into a finite base $\baseb'$ in light of a set of models. 
\end{definition}

For conciseness, we may omit the reference to the class $\Cmc$ of models if $\Cmc$ is clear from the context. 
The main challenge of model change operators is to guarantee the finiteness of the new base, which presents two main hurdles:
%The hurdle has two main reasons:
\begin{enumerate}
	\item some sets of models cannot be uniquely added/removed to/from some bases, %some bwhen adding or removing a set of interpretations $\mSet$, these interpretations cannot, in general, be uniquely added/removed, 
	as \cref{ex:nonunique} below illustrates. 
	% In this case, extra models must be minimally added/removed so that a finite base is properly achieved. 
	%So, extra models must be added/removed in behalf of finiteness.
	% For example, in several DLs, like 
	% \EL and \ALC, the models are closed under bisimulation, and therefore, incorporating/removing a model implies removing/adding all other bisimilar models. 
	%This makes it hard to find a proper finite representation for the new base. 
	
	\item the simple addition/removal of a set of models might not be finitely representable. 
	This occurs because the language of logic is not expressive enough to distinguish all the models in a set from those which are not in the set. 
	%In this case, extra models must be added/removed to achieve a finite base. 
	This issue is illustrated at 
	\cref{ex:separate}. 
	%\todo[inline]{Moved to new example 7, which better shows this. Example 6 would be better to shown motivate issues with compliance}
\end{enumerate}

\begin{example}\label{ex:nonunique}
	Consider the $\EL$ concept 
	$\exists r^3. \top $, with $\NC = \emptyset$, $\NR = \{r\}$, and the pointed models $(\Imc_1,d_1)$ and $ (\Imc_2,d_1)$ with domains $\Delta^{\Imc_1} = \{d_1,d_2\}$ and $\Delta^{\Imc_2} = \{d_1,d_2,d_3\}$, where 
	\begin{align*}
		\Imc_1 : \boxed{\begin{aligned}[t]
				r^{\Imc_1}= &\ \{(d_1,d_2)\}\\
		\end{aligned}} \hspace{2ex}
		\Imc_2 : \boxed{\begin{aligned}[t]
				r^{\Imc_2}= &\ \{(d_1,d_2),(d_2,d_3)\}
		\end{aligned}}
	\end{align*}
	Neither $(\Imc_1,d_1)$ nor  $(\Imc_2,d_1)$ are models of  $\exists r^3. \top$. 
	Suppose that we want to add $(\Imc_1,d_1)$ to the models of $\exists r^3. \top $.
	%, that is, construct a concept that contains the models of $\exists r^3. \top$  and $(\Imc_1,d_1)$. 
	In $\EL_\bot$, however, every concept satisfied by $(\Imc_1,d_1)$ and $\modelsof{\exists r^3. \top}$ is also satisfied by $(\Imc_2,d_1)$. There is no concept separating  $(\Imc_2,d_1)$ from  $\modelsof{\exists r^3. \top} \cup \{(\Imc_1,d_1)\}$. 
\end{example}

% \begin{example}\label{ex:separate}
	% Consider the concept $\exists r.\top$ and assume we only want to remove the model $(M^\infty,0)$ (see def. of the infinite chain $M^\infty$ in \cref{sec:prelidl}). In \ALC, there is no
	% finite concept that can `best' represent this. 
	% To see this, consider concepts of the form
	% $C_n=(\forall r^{n+1}.\bot)\sqcap (\exists r^n.\top)$. The first conjunct expresses that there is no chain of size greater than $n$ while the second expresses that there is a chain of size $n$. Such concepts can have $n$ arbitrarily large, there is no finite union of these concepts that includes all  models $(M^n,0)$, with $n\in\mathbb{N}$,  while excluding $(M^\infty,0)$. 
	% \end{example}

\begin{example}\label{ex:separate}
	Let $\baseb = \{\exists r. \top\}$ be an \EL concept and let the signature be $\NC = \{A\}, \NR=\{r\}$. 
	We want to evict the pointed model $(\Imc, d)$ with $\Delta^{\Imc} = \{d\}$, $r^{\Imc} = \{(d,d)\}$, and $A^\Imc=\emptyset$. 
	%\begin{align*}
	% $ \Imc: \boxed{\begin{aligned}[t]
			%      %A^{\Imc} = \{d'\}, 
			%      r^{\Imc} = \{(d,d)\}.
			%  \end{aligned}}
	%  $
	%\end{align*}
	We have that $(\Imc,d)$ satisfies $\baseb$. 
	The %simple 
	removal of $(\Imc, d)$ from $\modelsof{\baseb}$ yields the base 
	$\baseb'=\{  \exists r. A, \exists r^2.A, \cdots, \exists r^n.A, \cdots   \}$,
	%$\modelsof{\baseb}$ does no yield a finite base. $\modelsof{\baseb}\setminus \{(\Imc,d)\}$ yields  $\baseb' = \{ C \in \EL_{concepts} \mid  \modelsof{\baseb}\setminus \{(\Imc,d)\} \models C\}$. Hence, $$ \baseb' = \{ \exists r. A, \exists r^2.A, \cdots, \exists r^n.A \cdots   \}.  $$
	%Precisely, $\baseb'$ corresponds to the set of all \EL concepts satisfied by $\modelsof{\baseb} \setminus \{(\Imc,d)\}$. 
	which is not finitely representable. %as in fact we need infinitely many concepts to represent this base. 
	% If it were finitely representable, then there would be a finite subset $\baseb ''$ of it that entails all concepts of $\baseb'$. 
	% Each concept $C_i = \exists r^i. A$ in $\baseb'$ is violated by $(\Imc,d)$, and $\baseb'$ entails $\baseb$. 
	% %and every other concept entailed from  $\baseb$violated by $\Imc$ is entailed from some $C_i$. 
	% There is no finite subset of $\baseb'$, however, that entails $\baseb'$, as no $C_i$ entails $C_{i+i}$. 
	% %So, there is no concept capable of separating $\modelsof{\baseb}\setminus \{(\Imc,d)\}$ from its complement. \todo{reve se esta tudo ok aqui}
	% The finite subset that conserves most of the models from $\baseb$ is the finite base $\baseb'' = \{ \exists r. A \}$, as $\exists r. A$ entails all concepts in $\baseb'$, while violated by $(\Imc,d)$. A proper eviction operator must identify such a finite base that conserves most of the original models of $\baseb$.
	% \todo{put proof in the appendix ... should we have the formal proof in the appendix, and an observation that the base B' from example is not finitely representable, even more explicitly, $Form(\ models {\baseb}\setminus\{\Imc\})$ is not fin.rep. }
\end{example}

On both cases 1 and 2, as illustrated respectively on \cref{ex:nonunique} and \cref{ex:separate}, extra models must be added/removed to achieve a finite base. %For the second case, 
Such addition/removal should be minimised, so only models that do contribute to reaching  finiteness are considered. %, and are added/removed as little as possible. 
In this case, a ``closest'' finite base is produced. 
Such minimality criteria are properly addressed in the form of rationality postulates. % and present some specific nuances that differ from eviction to reception. %We briefly review eviction and reception in the next subsection. 
%The principle of minimal change is formalised via rationality postulates, and 
The appropriate class of all eviction/reception operators abiding by such postulates are identified. 

We show that in several classes of models, the minimality principles cannot be guaranteed, which is due to the non-existence of a ``closest'' finite base, known as the compatibility problem \citep{ouraaai}. % of eviction and reception. 
We briefly review reception and eviction, respectively, on \cref{sec:reception} and \cref{sec:eviction}. 
%We discuss their rationality postulates and respective operators. 
%In \cref{sec:dlconcept_evc_rcp}, we identify under which conditions the compatibility of eviction and reception is guaranteed for Description Logic Concepts. 
%We briefly review their postulates and respective rational operators. 
%The main hurdle with compatibility, as explored by ...., is the existence of den
%the compatinility, as per the  eviction and reception operators 
%We also discuss the problem of a class of models being compatible with eviction and reception. 
%Due to the issues 

\subsection{Reception}\label{sec:reception}
%\todo[inline]{check reception compatibility how defined in previous work}

We  denote belief change operators related to reception as $\mExp$. \citet{ouraaai} proposed the following rationality postulates to govern reception %governing reception are listed below. 

\begin{description}[leftmargin=*]
	\item[(success)] \(\mSet \subseteq \modelsof{\mExp(\baseb, \mSet)}\)
	
	\item[(persistence)] \(\modelsof{\baseb} \subseteq
	\modelsof{\mExp(\baseb, \mSet)} \) 
	
	\item[(finite temperance)] \(\mSet'
	\not\in \FRsets(\logsys)\), if
	\(\modelsof{\baseb} \cup \mSet  \subseteq \mSet'\) and \(\mSet' \subset
	\modelsof{\mExp(\baseb, \mSet)}\) 
	
	\item[(uniformity)] $\modelsof{\mExp(\baseb, \mSet)} =\modelsof{ \mExp(\baseb',
		\mSet')}$, if $\FRsups(\modelsof{\baseb} \cup \mSet, \logsys)
	=\FRsups(\modelsof{\baseb'} \cup \mSet', \logsys)$ .
\end{description}

%In the remainder of this section, all the results are from \citep{ouraaai}.
\emph{Success} ensures that each model from $\mSet$ must be incorporated. 
The purpose of reception is to accommodate new models. 
\emph{Persistence} ensures that no model is removed in the process.  
%\emph{Uniformity} for reception is analogous to \emph{uniformity} for eviction: reception is neither syntax sensitive nor sensitive to model structure. 
\emph{Finite temperance} %is analogous to \emph{finite retainment}, and 
ensures that the addition of extra models is minimized, adding only models that contribute to reaching a finite base. 
\emph{Uniformity} ensures that reception is neither syntax sensitive nor sensitive to model structure. 
% Two finite bases $\baseb$ and $\baseb'$ might differ syntactically but yet be semantically equivalent. For example, in classical propositional logic, the bases $\{a, \ a\to b\}$ and $\{a, b\}$ are semantically equivalent, although syntactically different. 
% Uniformity guarantees that reception is not syntax sensitive. 
% A similar phenomenon occurs for interpretations.fbe
% In several logics, some interpretations, although different in structure, may satisfy exactly the same formulae, making them indistinguishable. 
% For example, in \ALC, models are closed under bisimilarity (see \cref{def:bisimulaiton}). 
% Intuitively, if two interpretations 
% are indistinguishable, then reception of  
% either interpretations 
% should yield the same result. 

For instance, in \cref{ex:nonunique}, the pointed interpretation $(\Imc_1,d_1)$ cannot be separated from $(\Imc_2,d_1)$ in the presence of the models of $\baseb$, so reception, on $\baseb$,  of the sets $\{(\Imc_1,d_1)\}$ and $\{(\Imc_1,d_1), (\Imc_2,d_1)\}$ must coincide. 

% This can be generalised for sets of interpretations. If two sets $\mSet_1$ and $\mSet_2$ satisfy precisely the same set of formulae, then evicting either $\mSet_1$ or $\mSet_2$ should yield the same result. 
% Trivially, if the trivial removal of $\mSet_1$ from $\baseb$ reaches a finitely representable base, then removing $\mSet_2$ will yield the same result. 
% However, if the trivial result does not reach a finite representable base, then a choice must be made upon the minimal removal, as \emph{finite retainment} demands. 
% In this case, the choices for evicting $\mSet_1$ must coincide with the choices for evicting $\mSet_2$. 
%
%\begin{definition}
A \emph{reception operator} on a class $\Cmc $ of models is a model change operator $\mExp: \finitepwset(\llang) \times \Cmc \to \finitepwset(\llang)$ that satisfies \emph{success}.
%\end{definition}
%\todo[inline]{check minfrsups and maxfrsubs for the second parameter}
A reception operator satisfying all rationality postulates is called \emph{rational}. 
% In the easiest case, reception of a set of models $\Mplus$ corresponds to simply augmenting the models of $\baseb$ with $\Mplus$, that is, 
When there is no finite base for 
%writing a base for the set 
$\modelsof{\baseb} \cup \mSet$,   
some further models must be added in favor of finiteness. 
Such extra removal must be minimized as finite-temperance demands. 
This corresponds to picking a least finitely representable superset of  
$\modelsof{\baseb} \cup \mSet$, that is, picking a set from 
$ \FRsups(\modelsof{\baseb}\cup \mSet, \logsys) $. 
However, for some classes of models, such a least finitely representable superset does not exists, that is, $\FRsups(\modelsof{\baseb}\cup \mSet, \logsys) $ is empty. 
In such cases, therefore, finite-temperance cannot be satisfied, which implies in the inexistence of rational reception operators. 
Classes of models in which such least finitely representable supersets exist are called reception-compatible. 

\begin{definition}
	A class $\Cmc$ of models %$\Cmc$ 
	on a satisfaction system $\logsys$,
	is reception-compatible iff for all $\mSet \in \Cmc$ and base $\baseb \in \finitepwset(\llang)$,
	$\FRsups(\modelsof{\baseb}\cup \mSet, \logsys) \neq \emptyset$.
\end{definition}

%\citet{ouraaai} have identified several satisfaction systems whose 
Several classes of all models are not reception-compatible, as we show in  \cref{sec:dlconcept_evc_rcp}. 
For instance, in the DL  $\EL_{\bot}$, the class of all its models is not reception compatible, as 
\Cref{ex:chains} illustrates. 
% In \cref{sec:dlconcept_evc_rcp}, we investigate the issue of reception-compatibility for classes of models on DL concepts. %\todo[color=red]{Na tabela da sec 4, apareced Elbottom, mas na verdade eh EL bot concept, a notacao la precisa ser mudada para nao causar confusao.}
%We show in \crefone of theses cases. \todo[color=red]{ver com Ana, EL bot nao eh, mas EL bot concept eh. Seria bom explicar aqui isso? Se sim, como fazer sem causar confusao.}
%We show on \cref{sec:dlconcepts}, that when the language is restricted on  % this issue. 
% when $\Cmc$ covers all sets of models of the satisfaction system, that is, when $\Cmc = \powerset(\mSet)$. 

% \todo[color=red,inline]{should we put one with finie model as well?}
\begin{example}\label{ex:chains} %\todo[color=green]{remove ex to gain space}
	Consider the concept $\bot$ in  $\EL_\bot$ %description logic 
	and a signature with  $\NR=\{r\}$. Suppose $\Mplus=\{(M^{\infty},0)\}$ (see def. %of $M^{\infty}$ 
	in \cref{sec:prelidl}). %That is, 
	We would like to include an interpretation pointed at an infinite chain but we are not asking the finite chains $M^n$ to be included. There is no $\EL_\bot$ concept that best represents adding this set to $\bot$. Given a concept $\exists r^n.\top$ 
	%of depth $n$ 
	one can always create another concept $\exists r^{n+1}.\top$ %of depth $n+1$ 
	that includes   $\Mplus$ but has strictly less models: 
	$  \modelsof{\exists r. \top} \supset  \modelsof{\exists r^2. \top} \supset \cdots \supset \modelsof{\exists r^n. \top} \supset \cdots $.
\end{example}

% On reception-compatible classes, 
% the closest finitely representable bases corresponds to the most plausible candidates for reception. When several plausible candidates exist, a choice must be made. % when several plausible candidates exist. 
% The choice is realised by a choice function (\cref{def:choicefun}).
% This strategy yields rational reception operators (\cref{def:reception_con}).  

% \begin{definition}\label{def:choicefun}\citep{ouraaai}
	%     A choice function  on a  class  $\Cmc$ of models is a map $\sel: \finitepwset(\finitepwset(\Cmc)) \to \finitepwset(\Cmc)$ s.t. 
	% $\sel(X) \in X$, if $X \neq \emptyset$. 
	% \end{definition}
% % A choice function must choose exactly one element, whenever the list of choices is not empty. 

% % \todo[inline, color=green]{maybe bring the example and definitin of compatibility here}

% \begin{definition}\label{def:reception_con}\citep{ouraaai}
	%     Let $\Cmc$ be a reception-compatible class of models, and $\sel$ a choice function on $\Cmc$. 
	%     A maxichoice reception operator $\mExp_{\sel}$, on $\sel$, is a function $\mExp_{\sel}: \finitepwset(\llang) \times \Cmc \to \finitepwset(\llang)$ such that 
	%     $\modelsof{\mExp_{\sel}(\baseb, \Mplus)} = \sel(\FRsups(\modelsof{\baseb} \cup \Mplus, \logsys))$. 
	% %    \todo{add parameter and check}
	% \end{definition}

\Cref{ex:recep} illustrates a reception operation. 

\begin{example}\label{ex:recep}(continued from \cref{ex:nonunique}). 
	%     Let $\baseb = \{\exists r. \top\}$ be a base on $\EL_{concept}$ on the signature $\NR = \{A,B\}, \NR=\{r\}$. 
	%     We want the reception of  the pointed model $(\Imc, d)$ on domain $\Delta^{\Imc} = \{d,d'\}$ such that
	%     \begin{align*}
		%     \Imc= \boxed{\begin{aligned}[t]
				%         A^{\Imc} = \{d\}, B^{\Imc} = \{d'\},  r^{\Imc} = \{(d,d), (d',d')\}
				%    \end{aligned}}
		% \end{align*}
	%     Note that $\Imc$ violates $\baseb$, 
	%     as $(\exists r.A)^{\Imc} \cap (\exists r.B)^{\Imc} = \emptyset$. 
	%     The set $\modelsof{\baseb} \cup \{\Imc\}$ does not yield a finite base, as 
	%     Simply adding $\Imc$ does not yield 
	%     fr$\modelsof{\baseb}$ does no yield a finite base. Indeed, $\modelsof{\baseb}\setminus \{\Imc\}$ yield the following infinite base 
	%     $$ \baseb' = \{ \exists r. A, \exists r^2.A, \cdots, \exists r^n.A \cdots   \}  $$
	% =====
	Recall we want the reception of $\exists r^3. \top$ with 
	$(\Imc_1, d_1)$. Both  $ \exists r^2.\top $ and $\exists r. \top$. 
	are more general than $\exists r^3. \top$. 
	From these, the closest that contains $(\Imc_1,d_1)$ is $\exists r. \top$. 
	So, $ \mExp(\{\exists r^3. \top\}, \{(\Imc_1,d_1)\}) = \{\exists r. \top\}$.
\end{example}

On reception-compatible classes of models, 
\citet{ouraaai} have proposed the family of maxichoice reception operators, which chooses one finitely representable base from $\FRsups$.
Maxichoice reception operators %are characterised by all rationality postulates of reception; that is, the family of maxichoice reception operators 
coincide with all rational reception operators.

% The rationality postulates of reception characterise the class of all maxichoice reception operators, that is, maxichoice operators are all and the only reception operators.

% \begin{theorem}\citep{ouraaai}
	%     A model change operator $\mExp$ on a class $\Cmc$ 
	%     is a rational reception operator iff $\mExp$ is a maxichoice reception operator. 
	% \end{theorem}

% As rational reception operators do not exists in all classes of models, % cannot always be defined on all satisfaction systems, 
% one should identify suitable classes of models, where reception can be properly defined. In \cref{sec:dlconcept_evc_rcp}, we investigate this issue for  DL concepts. 
% We call such classes reception-compatible. 

% \begin{definition}
	%     A class of models %$\Cmc$ 
	%     on a satisfaction system %$\logsys$,
	%     is reception-compatible iff there is a  
	%     \todo{we should have added rational}
	%    rational reception operator %$\mExp$
	%     on it. %$\Cmc$. 
	% \end{definition}

\subsection{Eviction}\label{sec:eviction}

%The principles of minimal change governing eviction are conceptualised via rationality postulates. 
%Let $\mCon$ denote a belief change operator. 
Eviction operators, denoted by $\Evc$, are governed by the following postulates: 

\begin{description}[leftmargin=*]
	\item[(success)] \(\mSet \cap \modelsof{\mCon(\baseb, \mSet)} = \emptyset\).
	\item[(inclusion)] \(\modelsof{\mCon(\baseb, \mSet)} \subseteq
	\modelsof{\baseb}\).
	%
	%\item[(vacuity)] If $\mSet \cap \modelsof{\baseb} = \emptyset$,  then \\
	%\hspace*{1cm} \hfill$\modelsof{\mCon(\baseb, \mSet)} = \modelsof{\baseb}$.
	\item[(finite retainment)] \(\mSet' \not\in \FRsets(\logsys)\), if
	\(\modelsof{\mCon(\baseb, \mSet)}
	\subset \mSet'\) and \( \mSet'\subseteq \modelsof{\baseb} \setminus \mSet\).
	
	\item[(uniformity)] 
	$\modelsof{\mCon(\baseb, \mSet)} = \modelsof{\mCon(\baseb',  \mSet')}$, if 
	$\FRsubs(\modelsof{\baseb} \setminus \mSet,
	\logsys) =\FRsubs(\modelsof{\baseb'} \setminus \mSet',
	\logsys)$ 
	
\end{description}

\emph{Success} ensures that each model from $\mSet$ must be relinquished. As the purpose of eviction is to remove models, \emph{inclusion} ensures that no models are added. \emph{Uniformity}, as for reception, guarantees that eviction is neither syntax sensitive nor sensitive to model structure. 
% \emph{Finite retainment} %and \emph{uniformity} 
% deserves a bit more explanation due to its complex nature. 
\emph{Finite retainment} ensures that the removal of extra models is minimized, retracting only models that contribute to reaching a finite base. 
An \emph{eviction operator} on a class $\Cmc $ of models is a model change operator $\mCon: \finitepwset(\llang) \times \Cmc \to \finitepwset(\llang)$ that satisfies \emph{success}.
%\end{definition}
%Evcition operators satisfying all the rationality postulates are called rational eviction operators.
An eviction operator satisfying all rationality postulates is called \emph{rational}. 
In the best scenario, to evict a set of models $\mSet$ from a finite base $\baseb$, one would simply identify a base for the set $\modelsof{\baseb} \setminus \mSet$. 
However, as not all sets of models are finitely representable, %and 
some further models must be minimally removed in favour of finiteness, as  % as illustrated at \cref{ex:separate}. % and 
%\cref{ex:evic}. 
%Such extra removal must be minimized as 
finite-retainment demands. 
This corresponds to picking a greatest finitely representable subset from 
$\modelsof{\baseb} \setminus \mSet$, that is, picking a set from 
$ \FRsubs(\modelsof{\baseb}\setminus \mSet, \logsys) $. 
This set, in general, is not a singleton, and a choice must be made among the most plausible candidates.  %(see sets $X_1$ and $X_2$ at \cref{ex:evic} below). 
The choice, as for reception, is realised by a choice function.
Similarly to reception, depending on the underlying  class $\Cmc$ of models, $\FRsubs(\modelsof{\baseb}\setminus \mSet, \logsys) $ might be empty. 
In such a case,  finite retainment cannot be satisfied, which means that in such classes rational eviction operators do not exist. 
Classes of models in which the greatest finitely representable subsets exist are called eviction-compatible.

\begin{definition}
	A class $\Cmc$ of models, 
	on a satisfaction system $\logsys$,
	is eviction-compatible iff for all $\mSet \in \Cmc$ and base $\baseb \in \finitepwset(\llang)$,
	$\FRsubs(\modelsof{\baseb}\setminus \mSet, \logsys) \neq \emptyset$.
\end{definition}

On eviction-compatible classes of models, 
\citet{ouraaai} have proposed the family of maxichoice eviction operators, which chooses one finitely representable base from $\FRsets$.
Maxihoicde reception operators are characterised by all rationality postulates of reception.

\section{Eviction and Reception: DL Concepts}\label{sec:dlconcept_evc_rcp}
%\todo[color=green, inline]{Poderiamos juntar sec 4 e 6 para economisar espaco} Ana: encurtei
%In this section, 
Here we investigate eviction and reception on DL concepts, focusing on the prototypical DLs \ALC and \EL. In the following, we denote by \(\logsys(\EL_{\bot\text{concepts}})\)
and \(\logsys(\ALC_{\text{concepts}})\)
%= (\llang_{\ALC_{\text{concepts}}}, \mUni_{\ALC_{\text{concepts}}}, \models_{\ALC_{\text{concepts}}})\)  
the
satisfaction systems for $\EL_\bot$ and $\ALC$ concepts with pointed interpretations as models. Table~\ref{results} summarises our results.  
%\todo{In the table, shouldn't be with concept in the subscript?}

\begin{table}[h]
	\centering
	\begin{tabular}{l |c| c }
		\toprule
		\bf{Sat. System} & \bf{Eviction} & \bf{Reception}  \\
		\midrule
		%\EL & no & yes? \\
		%Bottom con. & $\bot$ &  $\emptyset$\\
		$\logsys(\EL_{\bot con.})$  & yes (\cref{thm:positive}) & no (\cref{thm:notreception}) \\
		$\logsys(\EL_{\bot con.})^\dagger$   & yes (\cref{thm:positive}) & yes (\cref{thm:elreceptiontree}) \\
		\midrule
		%Existential restriction & $\exists R$ & $\{d\mid (d,e)\in R^\Imc\}$\\
		$\logsys(\ALC_{con.})$ & no (\cref{thm:alcconcepteviction}) & no (\cref{thm:notreception}) \\
		%\ALC (fin. tree) & ? & no (\cref{thm:alcnotrecpfinitetree}) \\
		%\ALC (fin. tree, fin. sig.) & ? & no (\cref{thm:alcnotrecpfinitetree}) \\
		%\ALC (fin. tree, fin. set) & ? & no (\cref{thm:alcrecfintreefinset}) \\
		$\logsys(\ALC_{con.})^\ddagger$ & yes (\cref{thm:alcposevictionreception}) & yes (\cref{thm:alcposevictionreception}) \\
		
		\bottomrule
	\end{tabular}
	%\captionsetup{justification=centering}
	\caption{Eviction and reception-compatibility for DL concepts. $\dagger$ is for the case pointed interpretations can only be tree-shaped and $\ddagger$ is for the case  %pointed interpretations 
		they can only be tree-shaped, over finite signatures, and sets of models can only be finite.}\label{results}
\end{table}

%\smallskip

%Our proof strategy for Theorem~\ref{thm:positive} is to invoke  Theorem~\ref{thm:invokefirst} (see~\cite{ouraaai}).
\begin{toappendix}

	\begin{theorem}[Theorem~16~\cite{ouraaai}]\label{thm:invokefirst}
		A satisfaction system $\logsys$ is eviction-compatible iff for every $\mSet \subseteq \mUni$ either (i) $\mSet 
		\in\FRsets(\logsys)$,
		(ii) $\mSet$ has an immediate predecessor in $(\FRsets(\logsys)\cup\{\mSet\},\subset)$, 
		or (iii) there is no $\mSet' \in \FRsets(\logsys)$ with $\mSet \subseteq \mSet'$.
	\end{theorem}
	%\todo{technically speaking no DL would satisfy (i) except for corner cases top and bot}
	We also use the following two technical lemmas.
	\begin{lemma}[Adapted~\cite{DBLP:conf/ijcai/BaaderKM99}]\label{clm:elsimulation} For all satisfiable $\EL_\bot$ concepts $C\text { and }D$, we have that 
		%the following are equivalent:
		%   \begin{enumerate}
			%\item
			$\models C \sqsubseteq D$ iff
			% \item           $d_C\in D^{\Imc_C}$;
			%   \item 
			there is a homomorphism from $(\Imc_D,d_D)$ to $(\Imc_C,d_C)$.
			%\end{enumerate}
		\end{lemma}
		
		\begin{lemmarep}\label{clm:tech} For all satisfiable $\EL_\bot$ concepts $C$, if $C$ is satisfiable then there is an $\EL_\bot$ concept $C'$
			such that $\models \bot\sqsubset C'\sqsubset C$.
		\end{lemmarep}
		
		\begin{proof}
			Let $n$ be the depth of the tree-shaped labelled structure representing $C$. If we take $C'$
			as the concept $C\sqcap \exists r^{n+1}.\top$
			then $\models \bot\sqsubset C'\sqsubset C$. Note that  if $C$ is satisfiable then $C\sqcap \exists r^{n+1}.\top$ is satisfiable (one can see this by the inductive construction of the canonical model). 
		\end{proof}
	\end{toappendix}
	
	\begin{theoremrep}\label{thm:positive}
		$\logsys(\EL_{\bot\text{concepts}})$ is  eviction-compatible. 
		%\todo{EL concepts nao e reception compatible}
	\end{theoremrep}
	\begin{proof}
		We show that every \emph{non-empty} $\mSet\subseteq\mUni$ has an immediate predecessor in
		$$(\FRsets(\logsys(\EL_{\bot\text{concepts}}))\cup \{\mSet\},\subset)$$ (the case $\mSet$ is empty is easy as $\EL_\bot$ has the concept $\bot$).
		Recall that an immediate predecessor of $\mSet$ in this case is a
		set of models $\mSet'\in \FRsets(\logsys(\EL_{\bot\text{concepts}}))$ such that $\mSet'\subset\mSet$ and there is no $\mSet''$ such that 
		$\mSet'\subset\mSet''\subset\mSet$.  
		For this, we show that given a finitely representable set of models $\mSet$, there is no infinite chain \[\mSet_1 \subset \mSet_2\subset \ldots\] of   sets of models in $(\FRsets(\logsys(\EL_{\bot\text{concepts}}))$ such that $\mSet_i\subset\mSet$ for all $i\in\mathbb{N}$.
		Indeed, given $\mSet,\mSet'\in(\FRsets(\logsys(\EL_{\bot\text{concepts}}))$,  with $\mSet'\subset\mSet$, let
		$C,D$ be $\EL_{\bot}$ concepts such that $\modelsof{C} = \mSet$ and $\modelsof{D} = \mSet'$. By the semantics of 
		$\EL_\bot$, we have that $\models D\sqsubset C$. We want to show that there are finitely many $\EL_{\bot}$ concepts $E$ such that \[\models D\sqsubset E\sqsubset  C.\] 
		%the following claim.
		%\begin{claim}
		%There are finitely many $\EL_{\bot}$ concepts $E$ such that $\models D\sqsubset E\sqsubset  C$. 
		%\end{claim}
		%\begin{proof}
		{Suppose} there is no $\mSet'\in \FRsubs(\logsys(\EL_{\bot \text{concepts}}))$
		such that $\mSet'\subset\mSet$   with $\mSet'\neq \emptyset$. That is,  
		$\bot$ is the only $\EL_\bot$ concept such that
		$\models \bot\sqsubseteq C$, or, in other words, there is no (satisfiable) $\EL_{\bot}$ concept $C'$ such that  \[\models \bot\sqsubset C'\sqsubset C.\] However, as $\mSet'\subset\mSet$, we have that $\mSet\neq\emptyset$, so $C$ is satisfiable. This cannot be the case for $\EL_{\bot}$ concepts due to \cref{clm:tech}. 
		By \cref{clm:tech}, if $C$ is satisfiable then there is $C'$
		such that \[\models \bot\sqsubset C'\sqsubset C.\]
		{Then, there is $ \mSet'\in(\FRsets(\logsys(\EL_{\bot\text{concepts}}))$  such that $\mSet'\subset\mSet$ and $\mSet'\neq \emptyset$, so $D$ is satisfiable.}
		%If  $\models  D\sqsubset E$ 
		%then $\models  D\sqsubseteq E$. 
		Let $E$ be an $\EL_{\bot}$ concept  such that \[\models D\sqsubset E\sqsubset  C\] (if there is no such concept then we are done).
		Since  
		%$D$ is satisfiable and 
		$\models  D\sqsubset E$  we have that $E$ is satisfiable. Then, by Lemma~\ref{clm:elsimulation},  there is a homomorphism from $(\Imc_E,d_E)$ to $(\Imc_D,d_D)$.
		As $D$ is finite, the number of concept and role names that occur in $D$ is finite. Denote with ${\sf sig}(D)$ the set of concept and role names occurring $D$.
		Also, the existence of  a homomorphism from $(\Imc_E,d_E)$ to $(\Imc_D,d_D)$ and the definition of  $\Imc_E$ and $\Imc_D$
		imply that the set of concept and role names occurring in $E$ is a subset of ${\sf sig}(D)$. 
		%\footnote{Technically this may not be the case if both $E,D$ are equivalent to $\bot$ since e.g. $\models \bot\sqsubseteq \exists r.\bot$. Though, the assumption that $\models D\sqsubset  E$ implies
			%$E$  is necessarily satisfiable.}. 
		That is, there are finitely many concept and role names occurring in $E$ and they are subset of ${\sf sig}(D)$.  Moreover, by the definition of 
		$(\Imc_E,d_E)$ and $(\Imc_D,d_D)$, these interpretations correspond to tree-shaped labelled structures (pointed at the root) where the depth of 
		(the tree corresponding to) $\Imc_E$ is less than or equal to the depth $n$ of $\Imc_D$. Since $E$ was an arbitrary 
		$\EL_\bot$ concept such that $\models  D\sqsubset E$ this holds for all such concepts.
		As  there are finitely many tree-shaped labelled structures with symbols in ${\sf sig}(D)$ and depth bounded by $n$, there are finitely many $\EL_\bot$ concepts $E$ (up to logical equivalence) such that
		$\models  D\sqsubset E$. Then there are finitely many $\EL_\bot$ concepts $E$ (up to logical equivalence) such that \[\models D\sqsubset E\sqsubseteq C.\] This means that, by \cref{thm:invokefirst}, $\logsys(\EL_{\bot\text{concepts}})$ is  eviction-compatible. 
	\end{proof}
	%\begin{proof}[Sketch]
	%   The main intuition here is because in this language one cannot express infinitely many ``more general'' concepts (which we can in e.g. \ALC using disjunctions). This means that one can always find a ``maximal'' concept that does not satisfy a given set of models.
	%\end{proof}
	
	Theorem~\ref{thm:positive} does not hold for $\EL$   (without $\bot$) as any language that cannot express inconsistencies is not eviction-compatible~\citep{ouraaai}. 
	We also do not have eviction-compatibility for \ALC.

	\begin{theoremrep}
		\label{thm:alcconcepteviction}
		$\logsys(\ALC_\text{concepts})$ is not \mconnm-compatible. 
	\end{theoremrep}
	\begin{proof}
		The proof of \cref{thm:alcconcepteviction}
		is based on \cite{ouraaai}.
		We show that $\logsys(\ALC_\text{concepts})$ is not \mconnm-compatible (the proof works if the signature is finite or infinite).
		%For conciseness, we  write \(\models\) instead of \(\models_{\ALC_{\text{concepts}}}\) within this proof.  
		Let \(C_\top :=\top\), that is, \(\modelsof{C_\top} = \mUni_{\ALC_{\text{concepts}}}\). 
		%and $a^{M^{\infty}}=0$.
		Let $\mSet$ be the set of all models $(M,0)$
		such that for some $n\in\mathbb{N}$
		we have that $0\in (\forall r^n.\bot)^M$.
		That is, there is no loop or infinite chain
		of elements connected via the role $r$ starting
		from $0$. 
		By definition of $\mSet$, we have that
		$(M^{\infty},0)\not\in \mSet$ since
		this model has an
		infinite chain
		of elements connected via the role $r$ starting
		from $0$, while $(M^{n},0)\in \mSet$ for all $n\in\mathbb{N}$.

		To prove that $\logsys(\ALC_\text{concepts})$ is not \mconnm-compatible, we need to prove that there is no
		\ALC concept $C$ 
		%\(C \in \finitepwset(\llang_{\ALC_\text{concepts}})\) 
		such that \[\modelsof{C} \in \FRsubs(\mSet, \logsys(\ALC_\text{concepts})),\] that is, \[\FRsubs(\mSet, \logsys(\ALC_\text{concepts})) = \emptyset.\]
		Intuitively, we want to show that 
		we cannot find a maximal $\ALC$ concept 
		that finitely represents the result of removing the models in $\mUni_{\ALC_{\text{concepts}}} \setminus \mSet$ 
		from $C_\top$.
		First, we  recall 
		the following claim.
		% \todo{dont know how to make cite work}
		The next result is by \citep{DBLP:conf/nmr/0001OR23}.
		\begin{claim}\label{cl:eviction-aux}
			For every \ALC{} concept $C$, 
			if there is $n\in\mathbb{N}$ 
			such that 
			$0\in C^{M^m}$
			for all $m\geq n$, with $m\in\mathbb{N}$, 
			then $0\in C^{M^{\infty}}$.
		\end{claim}
		
		We are now ready to show that 
		$\logsys(\ALC_\text{concepts})$ is not \mconnm-compatible.
		Suppose to the contrary that
		there is an \ALC concept $C$ 
		%\(C \in \finitepwset(\llang_{\ALC_\text{concepts}})\) 
		such that \[\modelsof{C} \in \FRsubs( \mSet, \logsys(\ALC_\text{concepts})).\] 
		If there is $n\in\mathbb{N}$ such that
		$0\not\in C^{M^n}$ then\footnote{Recall that
			$M^n$ has a chain of $n+1$ elements connected via the role $r$.}
		\[C':=C\sqcup (
		\exists r^{n+1}.\top\sqcap \neg \exists r^{n+2}.\top), 
		\]
		is such that $0\in C'^{M^n}$.
		Moreover, $\modelsof{C}\subset\modelsof{C'}$.
		By definition of $C'$ and $\mSet$, we also have that
		\(\modelsof{C'} \subseteq  \mSet\).
		This contradicts the assumption that
		\[\modelsof{C} \in \FRsubs( \mSet, \logsys(\ALC_\text{concepts})).\]
		So, for all $n\in\mathbb{N}$, we have  that
		$0\in C^{M^n}$.
		Then, by Claim~\ref{cl:eviction-aux}, it follows that  
		$0\in C^{M^{\infty}}$. 
		Since, as already mentioned, $M^{\infty}\not\in\mSet$,
		this contradicts the assumption that
		\[\modelsof{C} \in \FRsubs( \mSet, \logsys(\ALC_\text{concepts})).\]
		Thus, \[\FRsubs( \mSet, \logsys(\ALC_\text{concepts}))=\emptyset.\] 
	\end{proof}
	The next theorem establishes that reception-compatibility neither  holds for $\EL_\bot$ 
	%(with or without $\bot$) 
	nor for \ALC concepts. 
	\begin{theoremrep}
		\label{thm:notreception}
		$\logsys(\EL_{\bot \text{concepts}})$ and $\logsys(\ALC_{\text{concepts}})$  are not reception-compatible. This %result
		holds even if we restrict to the class of (possibly infinite) sets of (possibly infinite) tree-shaped pointed interpretations or if we restrict to the class of sets of pointed interpretations over a unique finite signature.
		% \todo{even if we restrict to tree-shaped pointed interpretations}
		%\todo{even if the signature is finite}
		%\todo{EL concepts nao e reception compatible}
	\end{theoremrep}
	\begin{proof}
		We prove this theorem for  $\logsys(\EL_{\bot \text{concepts}})$, the same argument can be used for 
		$\logsys(\ALC_{\text{concepts}})$.
		Let \(C_\bot :=\bot\), that is, \(\modelsof{C_\bot} = \emptyset\). 
		%and $a^{M^{\infty}}=0$.
		Let $\mSet$ be $\{(M^{\infty},0)\}$. To prove that $\logsys(\EL_{\bot \text{concepts}})$ is not reception-compatible, we need to prove that there is no \(C \in \finitepwset(\llang_{\EL_{\bot \text{concepts}}})\) such that \[\modelsof{C} \in \FRsups(\mSet, \logsys(\EL_{\bot \text{concepts}})).\] In other words, \[\FRsups(\mSet, \logsys(\EL_{\bot \text{concepts}})) = \emptyset.\]
		
		\begin{claim}\label{claim:elneg}
			For all satisfiable $\EL_\bot$ concepts $C$, there is $n\in\mathbb{N}$ such that  $\not\models C\sqsubseteq\exists r^n.\top$.
		\end{claim}
		
		This follows by the syntax and semantics of \EL, which only allows for finite concepts, with finite nesting of existential quantifiers. 
		%A simple structural induction argument can be given. 
		
		\medskip
		
		\noindent
		Suppose to the contrary that
		there is (a satisfiable) \(C \in \finitepwset(\llang_{\EL_{\bot \text{concepts}}})\) such that \[\modelsof{C} \in \FRsups( \mSet, \logsys(\EL_{\bot \text{concepts}})).\] 
		% Otherwise, 
		By Claim~\ref{claim:elneg}, there is $n\in\mathbb{N}$ such that  $\not\models C\sqsubseteq\exists r^n.\top$. Then $C':=C\sqcap\exists r^n.\top$
		is such that $\mod{(C')}\subset\mod{(C)}$. Note that, by assumption,
		$(M^{\infty},0)\in\mod{(C)}$ and the definition of $C'$ does not remove $(M^{\infty},0)$ so $(M^{\infty},0)\in\mod{(C')}$. This   contradicts  the assumption that $\modelsof{C} \in \FRsups( \mSet, \logsys(\EL_{\bot \text{concepts}}))$. 
		%We then have that $\models C\sqsubseteq\exists r^n.\top$ for all $n\in\mathbb{N}$.
		%However, this is impossible for an \EL concept $C$. 
		Thus, $\logsys(\EL_{\bot \text{concepts}})$ is not reception-compatible. 
	\end{proof}
	We now concentrate on finding a restricted class of concepts where reception-compatibility holds.  
	%that we have found a positive case for $\EL_\bot$, we   concentrate on finding a class of models that is also positive for \ALC, that is, one where $\logsys(\ALC_{\text{concepts}})$ is both eviction and reception-compatible.
	\begin{toappendix}
		\begin{theorem}[Theorem~16~\cite{ouraaai}]\label{thm:invoke}
			A satisfaction system $\logsys$ is reception-compatible iff for every $\mSet \subseteq \mUni$ either (i) $\mSet 
			\in\FRsets(\logsys)$,
			(ii) $\mSet$ has an immediate successor in $(\FRsets(\logsys)\cup\{\mSet\},\subset)$, 
			or (iii) there is no $\mSet' \in \FRsets(\logsys)$ with $\mSet' \subseteq \mSet$.
		\end{theorem}
	\end{toappendix}
	\begin{theoremrep}\label{thm:elreceptiontree}
		$\logsys(\EL_{\bot \text{concepts}})$    is   reception-compatible in the class of (possibly infinite) sets of finite tree-shaped pointed interpretations  (over a possibly infinite signature). 
		%\todo{EL concepts nao e reception compatible}
	\end{theoremrep}

	\begin{proof}
		Point (iii) of \cref{thm:invoke} never holds because we can express unsatisfiable concepts in  $\EL_\bot$ using $\bot$. We argue that
		(i) or (ii) always hold for this class of models.
		We first consider $\logsys(\EL_{\bot \text{concepts}})$.
		Let $\mUni$ be the class of all finite tree-shaped pointed interpretations.
		Given $\mSet \subseteq \mUni$, if $\mSet 
		\in\FRsets(\logsys)$ then we are done. Suppose this is not the case. We then need to show that $\mSet$ has an immediate successor in $(\FRsets(\logsys)\cup\{\mSet\},\subset)$. 
		
		Let $(\Imc,d)$ be  a tree-shaped pointed interpretation  in  $\mSet$ with minimal depth $k$. Such $(\Imc,d)$ and $k$ exist since pointed interpretations in $\mSet$ are finite and tree-shaped (we know that $\mSet\neq \emptyset$ since we assume case (i) does not hold and $\emptyset\in (\FRsets(\logsys)$). For any $\EL_{\bot}$ concept $C$, if $(\Imc,d)\in\modelsof{C}$ then the depth of $C$ is at most $k$. Moreover, 
		if $\mSet\subseteq \modelsof{C}$ then
		the symbols (that is, concept and role names) occurring in $C$ are a subset of the intersection of the symbols occurring in the elements of $\mSet$. Since each element of $\mSet$ is finite, such intersection must also be finite. We then have that, up to logical equivalence, there are finitely many $\EL_{\bot}$ concepts $C$ such that $\mSet\subseteq \modelsof{C}$. This means that $\mSet$ has an immediate successor in $(\FRsets(\logsys)\cup\{\mSet\},\subset)$. 
		% \todo{finite signature for %ALC}
	\end{proof}

	%\todo{ref def of  isomorphism}

	%\todo[color=red]{also Th12?} 
	\cref{thm:alcconcepteviction} and \cref{thm:notreception} compel us to investigate more restricted classes of models to obtain a positive result for \ALC. There are two natural ways of restricting this class: restricting to finite sets of finite tree-shaped interpretations or to finite tree-shaped interpretations with finite signature. The next  theorem establishes that these restrictions alone are not sufficient \ALC reception-compatibility.
	%\ALC \EL
	\begin{theorem}\label{thm:alcrecfintreefinsetboth}
		$\logsys(\ALC_{\text{concepts}})$ is neither reception-compatible  in the class of finite sets of finite tree-shaped pointed interpretations over a (possibly infinite) signature; nor  in the class of (possibly infinite) sets of finite tree-shaped
		pointed interpretations over a (unique) finite signature. 
	\end{theorem}
	\begin{toappendix}
		Theorem~\ref{thm:alcrecfintreefinsetboth} follows from Theorem~\ref{thm:alcrecfintreefinset} and Theorem~\ref{thm:alcnotrecpfinitetree}.
		
		\begin{theorem}\label{thm:alcrecfintreefinset}
			$\logsys(\ALC_{\text{concepts}})$ is not reception-compatible in the class of finite sets of finite tree-shaped pointed interpretations over a (possibly infinite) signature.  
		\end{theorem}
		\begin{proof}
			Suppose the signature $\Sigma$ has symbols of the form $A_i$ for all $i \in \mathbb{N}$.
			Assume the base concept is $\bot$ and we want to add a model $(\Imc,d)$ with one domain element $d$ in the extension of $A_0$ but not in the extension of any other $A_i$. Let $\mSet$ be a singleton set containing such model.
			Conditions (i) and (iii)  
			of \cref{thm:invoke} fail
			for $\mSet$. Indeed,
			Condition (i) fails because
			for any  concept $C$ satisfying  $(\Imc,d)$ there is $A_i$ that does not occur in it  (because $C$ is finite and $\Sigma$ is not)
			and if we add $d$ to the extension of $A_i$ then the resulting pointed interpretation $(\Imc',d)$ still satisfies $C$.
			Condition (iii) fails because we can take $\mSet'$ as the empty set, which is in $\FRsets(\logsys(\ALC_{\text{concepts}}))$ since $\bot$ is an \ALC concept. From the argument for Condition (i) we can already see that 
			Condition (ii) also fails. That is, $\mSet$ has no immediate successor because given $C$ with $\mSet\subseteq \modelsof{C}$ there is always $A_i$ that does not occur in $C$ and $C':=C\sqcap \neg A_i$ will be such that 
			$\mSet\subseteq \modelsof{C'}$ and $\modelsof{C'}\subset \modelsof{C}$.
			%\todo{...}
		\end{proof}
		
	\end{toappendix}
	\begin{toappendix}

		\begin{theorem}\label{thm:alcnotrecpfinitetree}
			$\logsys(\ALC_{\text{concepts}})$ is not reception-compatible in the class of (possibly infinite) sets of  finite tree-shaped pointed interpretations over a (unique) finite signature.  
		\end{theorem}
		\begin{proof}
			Consider a signature with two symbols: a role name $r$ and a concept name $A$. Assume
			$\mathbb{M}$ is the set of all pointed interpretations of the form $(M^n,0)$ with 
			%\todo{$M^n$ as def in the appendix}
			$n\in\mathbb{N}$.
			Suppose to the contrary that
			$\logsys(\ALC_{\text{concepts}})$ is   reception-compatible in the class of   finite tree-shaped pointed interpretations.
			That is, suppose
			there is
			an \ALC concept $C$ such that
			\[\modelsof{C} \in \FRsups( \mSet, \logsys(\ALC_{\text{concepts}})).\]
			Suppose $m$ is the depth of $C$. Take $k=m+1$ and let $C_k$ be as follows
			%   
			%    \todo{...}
			%
			\[\bigsqcup^k_{i=1}(\exists r^i.\forall r.\bot\sqcap \bigsqcap_{j<i}(\forall r^j.\exists r.\top) \sqcap (\bigsqcap^i_{l=1} \forall r^l.\neg A))\sqcup ((\exists r^{k+1}.\top)\sqcap \]\[(\bigsqcap^{k+1}_{l=1} \forall r^l.\neg A)) \]
			We claim that $\modelsof{C_k}\subset\modelsof{C}$ and
			$\mathbb{M}\subseteq\modelsof{C_k}$, which contradicts the assumption that $\modelsof{C} \in \FRsups( \mSet, \logsys(\ALC_{\text{concepts}}))$. 
			Indeed, we first show that $\mathbb{M}\subseteq\modelsof{C_k}$.
			For this we just note that $C_k$ has
			two parts:
			\[C^1_k=\bigsqcup^k_{i=1}(\exists r^i.\forall r.\bot\sqcap \bigsqcap_{j<i}(\forall r^j.\exists r.\top)  \sqcap (\bigsqcap^i_{l=1} \forall r^l.\neg A)) \text{ and }\] 
			\[C^2_k=((\exists r^{k+1}.\top)\sqcap \bigsqcap_{j<k+1}(\forall r^j.\exists r.\top) \sqcap (\bigsqcap^{k+1}_{l=1} \forall r^l.\neg A)). \]
			We have that
			$0\in (C^1_n)^{M^n}$ for all $0\leq n< k$, and,  %we have that
			$0\in (C^2_n)^{M^n}$, for all 
			$n\geq k$. By definition of $\mathbb{M}$, this means that  $\mathbb{M}\subseteq\modelsof{C_k}$.
			
			We now show that 
			$\modelsof{C_k}\subset\modelsof{C}$.
			We first show that $\modelsof{C_k}\subseteq\modelsof{C}$. 
			We argue that if a pointed interpretation
			$(\Imc,d)$  satisfies
			$C_k$ then there is a $k'$-bisimulation between
			$\Imc$ and $M^{k'}$ that includes $(d,0)$
			for some $k'\leq k$.
			Indeed, if $d \in (C^1_k)^\Imc$
			then, by definition of $C^1_k$, there is a chain
			of $r$-successors in \Imc
			starting from $d$  of length $k'\leq k$, all chains of $r$-successors are of such length $k'\leq k$, and none of  the elements of the chain(s) is in the extension of $A$.
			Since $0\in C^{M^{k'}}$, by~\cref{lem:kinvariance}, $d\in C^\Imc$. In other words, $(\Imc,d)$  satisfies
			$C$.
			Otherwise
			$d \in (C^2_k)^\Imc$.
			Then, by definition of $C^2_k$, there is a chain
			of $r$-successors in \Imc
			starting from $d$  of length at least $k+1$, all chains of $r$-successors are of   length at least $k+1$, and none of  the elements of the chain(s) is in the extension of $A$ up to the $k+1$ length.
			Then, there is a $k$-bisimulation between
			$\Imc$ and $M^{k}$ that includes $(d,0)$
			for some $k\leq k$. Since $0\in C^{M^{k}}$, by~\cref{lem:kinvariance}, $d\in C^\Imc$. In other words, we also have that $(\Imc,d)$  satisfies
			$C$.
			%above is a bit informal  

			It remains to show that there is a pointed interpretation that satisfies
			$C$ but not $C_k$.
			Let $M^{m+1}_A$ be the result of extending
			$M^{m+1}$ by adding the element $m+1\in\mathbb{N}$ to the extension of 
			$A$. %Since the depth of $C$ is $m$, 
			%if $0\in C^{M^{m+1}}$ then $0\in C^{M^{m+1}_A}$.  
			By assumption, 
			$0\in C^{M^n}$ for all $n\in\mathbb{N}$.
			Then,   $0\in C^{M^{m+1}}$ and, by \cref{lem:kbisimulation} and \cref{lem:khbisimulation}, $0\in C^{M^{m+1}_A}$ since there is an $m$-bisimulation between $M^{m+1}$ and $M^{m+1}_A$. As $k=m+1$, by definition of $C_k$, $0\not\in (C_k)^{M^{m+1}_{A}}$.
			So $\modelsof{C_k}\subset\modelsof{C}$. This contradicts our assumption that \[\modelsof{C} \in \FRsups( \mSet, \logsys(\ALC_{\text{concepts}})).\] 
			%For $n=k=m+1$, we have that 
			%For $0\leq n<k$ and for all pointed interpretations $(\Imc,d)$, if $d$
			%    One can define using disjunction and negation models so that there is always a more suitable concept.
		\end{proof}
	\end{toappendix}
	%\todo{alc concept}
	
	\begin{toappendix}
		For our positive results for \ALC concepts, we use the following notions.
		Given a finite tree-shaped pointed interpretation $(\Imc,d)$ over a finite signature $\Sigma$, let $C_{\Imc}$ be the $\EL_\bot$ concept
		such that $\Imc$ is isomorphic to its canonical model.
		We define an \ALC concept 
		$(C_\Imc)^\dagger$ such that
		\begin{itemize}
			\item  if $(\Jmc,e)$ satisfies 
			$(C_\Imc)^\dagger$ then there is a bisimulation between 
			$(\Jmc,e)$ and $(\Imc,d)$ including $(e,d)$ (over $\Sigma$). 
		\end{itemize}
		We construct $(C_\Imc)^\dagger$  inductively, as follows. 
		\begin{definition}\label{def:translation}
			Given a satisfiable $\EL_\bot$ concept $C$, denote by $(C)^{\dagger}$   the result of translating $C$ into an \ALC concept as follows. Consider \NC and \NR to be finite.
			%and
			%let $\NC(C)$ and $\NR(C)$ denote the sets of concept and role names that occur in $C$.
			%$\Sigma=\{A_1,\ldots,A_n,r_1,\ldots, r_m\}$. 
			Assume w.l.o.g. that an $\EL_\bot$ concept
			$C$ is of the form 
			\[\bigsqcap_{A \in S} A\sqcap\bigsqcap_{r\in T}\bigsqcap^{m_r}_{l=1} \exists r.C^l_r\]
			for some $S\subseteq\NC$, some $T\subseteq \NR$, and each $C^l_r$ being an $\EL_\bot$ concept.
			%with $0\leq n'\leq n$ and $0\leq m'\leq m$. %Consider the case of $m'=0$ as the one with role depth $0$ and the case of $n'=0$ as the one with no concept names as top-level conjuncts. 
			%For the concept $C_\Imc$,
			%W.l.o.g., assume that no  conjunct of $C_\Imc$ implies the other (because such implied conjunct would be redundant). \todo{expand
				%}
			Then, we define
			$(C)^{\dagger}$ as:
			\[\bigsqcap_{A\in S} A\sqcap \bigsqcap_{B\in \NC\setminus S} \neg B \sqcap\bigsqcap_{r\in T}(\bigsqcap^{m_r}_{l=1} \exists r.(C^l_r)^{\dagger}\sqcap \forall r.(\bigsqcup^{m_r}_{l=1}(C^l_r)^{\dagger}))
			\]\[
			\sqcap  \bigsqcap_{s\in\NR\setminus T} \forall s.\bot\]
		\end{definition}
		In \cref{def:translation}, consider the empty cases for conjunction and disjunction 
		%(e.g. $n'=0$ in the first conjunction) 
		as $\top$.

		\begin{example}
			Assume $\NC=\{A,B\}$ and $\NR=\{r,s\}$. Let $C=B\sqcap \exists r.(A\sqcap B)$.
			Then $C^\dagger$ is
			\[\neg A\sqcap  B \sqcap (  \exists r.(A\sqcap B\sqcap \forall r.\bot\sqcap \forall s.\bot) \sqcap 
			\]\[
			\forall r.(A\sqcap B\sqcap \forall r.\bot\sqcap \forall s.\bot)) \sqcap    \forall s.\bot\]
		\end{example}

		%\begin{lemma}\label{lem:elsimulation}
		%   \todo[inline]{bring EL lemma homomorphism from the literature}
		%\end{lemma}
		
		\begin{lemma}\label{lem:translationconcept}
			%Let $(\Jmc,e)$ be a finite tree-shaped pointed interpretation over $\Sigma$. 
			Given a finite tree-shaped pointed interpretation $(\Imc,d)$ over a finite signature $\Sigma$, let $C_{\Imc}$ be the $\EL_\bot$ concept
			with $(\Imc,d)$  isomorphic to its canonical model. Then
			\begin{enumerate}
				\item $(\Imc,d)$ satisfies 
				$(C_\Imc)^\dagger$;
				\item  if a pointed interpretation $(\Jmc,e)$ satisfies 
				$(C_\Imc)^\dagger$ then there is a bisimulation between 
				$(\Jmc,e)$ and $(\Imc,d)$ including $(e,d)$. 
			\end{enumerate}
		\end{lemma}
		\begin{proof}
			% \todo{prove the direction if there is a bisimulation between 
				% $(\Jmc,e)$ and $(\Imc,d)$ including $(e,d)$ then $(\Jmc,e)$ satisfies   %$(C_\Imc)^\dagger$ }
			By assumption, \Imc is a finite tree-shaped interpretation over a finite signature $\Sigma$ and $C_\Imc$ is the $\EL_\bot$ concept with \Imc isomorphic its  canonical model.
			%This means that

			We prove the two items of the lemma by induction on the role depth of $C_\Imc$ (that is, the number of nestings of existential quantifiers in $C_\Imc$).
			%
			%Let $\Sigma=\{A_1,\ldots,A_n,r_1,\ldots, r_m\}$. 
			At role depth $0$
			$C_\Imc$ is of the form
			\[\bigsqcap_{A\in S} A\]
			for some $S\subseteq \NC$. 
			%with $0\leq n'\leq n$. 
			Then, $(C_\Imc)^\dagger$ is of the form 
			\[\bigsqcap_{A\in S} A\sqcap \bigsqcap_{B\in \NC\setminus S} \neg B \sqcap \bigsqcap_{r\in \NR} \forall r.\bot\]
			We first argue about Item (1). It holds because $(\Imc,d)$ is a pointed interpretation with $r^\Imc=\emptyset$ for all $r\in\NR$ and
			$A^\Imc=\{d\}$ for all
			$A\in S$ and $B^\Imc=\emptyset$ for all
			$B\in \NC\setminus S$. So it satisfies $(C_\Imc)^\dagger$. 
			Regarding Item (2), suppose $(\Jmc,e)$ satisfies $(C_\Imc)^\dagger$.
			Let $Z$ be the relation
			\(\{(e,d)\}\). 
			The atom condition holds because we only add to $Z$ pairs of elements that satisfy the atom condition (recall that \NC is finite, so a finite \ALC concept can enforce the atom condition). The back and forth conditions hold vacuously.

			Now assume that the concept $C_\Imc$ is of the form 
			\[\bigsqcap_{A \in S} A\sqcap\bigsqcap_{r\in T}\bigsqcap^{m_r}_{l=1} \exists r.C^l_r\]
			%\[\bigsqcap^{n'}_{i=1} A_i\sqcap\bigsqcap^{m'}_{j=1}\bigsqcap^{m_j}_{l=1} \exists r_j.C^l_j.\]
			for some $S\subseteq \NC$ and $T\subseteq \NR$.
			Suppose that the two items of this lemma hold for $C^l_r$. We want to show that for $(C_\Imc)^\dagger=$
			\[\bigsqcap_{A\in S} A\sqcap \bigsqcap_{B\in \NC\setminus S} \neg B \sqcap\bigsqcap_{r\in T}(\bigsqcap^{m_r}_{l=1} \exists r.(C^l_r)^{\dagger}\sqcap \forall r.(\bigsqcup^{m_r}_{l=1}(C^l_r)^{\dagger}))
			\]\[
			\sqcap  \bigsqcap_{s\in\NR\setminus T} \forall s.\bot\]
			%\[\bigsqcap^{n'}_{i=1} A_i\sqcap \bigsqcap^n_{k=n'+1} \neg A_k \sqcap\bigsqcap^{m'}_{j=1}(\bigsqcap^{m_j}_{l=1} \exists r_j.(C^l_j)^{\dagger}\sqcap \forall r_j.(\bigsqcup^{m_j}_{l=1}(C^l_j)^{\dagger})) \] \[\sqcap \bigsqcap^m_{o=m'+1} \forall r_o.\bot\]
			we have that (1) \ $(\Imc,d)$ satisfies $(C_\Imc)^\dagger$
			and (2)
			there is a bisimulation between
			$(\Jmc,e)$ and $(\Imc,d)$ including $(e,d)$, where $(\Jmc,e)$ is a pointed interpretation that satisfies $(C_\Imc)^\dagger$. We start with Item (1).
			By assumption, for every subconcept $C^l_r$ of $C_\Imc$, we have that
			$(\Imc^l_r,d^l_r)$ satisfies $(C^l_r)^\dagger$, with $(\Imc^l_r,d^l_r)$ being isomorphic to the canonical model of $C^l_r$. Let
			$(\Imc,d)$ be the pointed interpretation with
			\begin{itemize}
				\item the domain being the union of the domains of all 
				$\Imc^l_r$ (assume w.l.o.g. that these domains are disjoint)
				and a fresh element $d$; and 
				\item the extension of concept and role names being as in each $\Imc^l_r$, except that
				we  add $(d,d^l_r)$ to 
				$r^\Imc$, for all $r\in T$, and we add $d$ to $A^\Imc$, for all $A\in S$.
			\end{itemize}
			By definition, $(\Imc,d)$ is 
			isomorphic to the canonical model of $C_\Imc$. It satisfies $(C_\Imc)^\dagger$ because we added $d$ exactly to the extension of the concept names that are not negated in $(C_\Imc)^\dagger$ and connected $d$ with elements $d^l_r$ where
			$d^l_r\in (C^l_r)^\Imc$ holds by the inductive hypothesis.  This finishes the proof of Item (1).
			
			We now proceed with Item (2).
			By assumption, for every subconcept $C^l_r$ of $C_\Imc$, there is a bisimulation
			$Z^l_r$ between a pointed interpretation $(\Imc_{C^l_r},d_{C^l_r})$, that is isomorphic to the canonical model of $C^l_r$, and every subinterpretation $(\Jmc_{C^l_r},e_{C^l_r})$ of $(\Jmc,e)$    that satisfies $(C^l_r)^\dagger$, with $(e,e_{C^l_r})\in r^\Jmc$.  
			We define $Z$ as \[\{(e,d)\}\cup\bigcup_{r\in T,l\in\{1,\ldots,m_r\}} Z^l_r.\] 
			By definition, the atom condition holds for $(e,d)$. The back, forth, and atom conditions hold for the successors of $e$ and $d$ because, by   assumption, each $Z^l_r$ is a bisimulation. 
		\end{proof}
		%$(C_\Imc)^\dagger$   
		%To prove this lemma, we first argue that the following technical claims hold.
		%
		%\begin{claim}\label{clm:translation}
		%   For all $\EL_\bot$ concepts $C$,
		%       if $(\Jmc,e)$ satisfies $(C)^\dagger$
		%        then it satisfies $C$.
		% \end{claim}
	%\begin{proof}
	%  
	%\end{proof}
	%   By \cref{lem:elsimulation},
	% there is a homomorphism $h$ from $(\Imc,d)$ to $(\Jmc,e)$. 
	%Let $\Jmc'$ be the subinterpretation of \Jmc that results from:  %performing the following operations:
	%\begin{itemize}
	%   \item removing all nodes not reachable via a (directed) path from 
	%  $e$.
	%\end{itemize}
	%
	%\begin{claim}
	%    $\Jmc'$ is a tree-shaped interpretation rooted in $e$.
	%\end{claim} 
	
	%
	%\begin{claim}\label{%clm:techtree}
	%     $(\Jmc,e)$ and $(\Jmc',e)$ are bisimilar.
	%\end{claim} 
\end{toappendix}
\begin{theorem}\label{thm:alcposevictionreception}
	$\logsys(\ALC_{\text{concepts}})$ is   reception-compatible and eviction-compatible in the class of finite sets of finite tree-shaped pointed interpretations over any finite signature.
\end{theorem} 

\begin{toappendix}

	Theorem~\ref{thm:alcposevictionreception} follows from the next two theorems.
	\begin{theorem}\label{thm:alcposreception}
		$\logsys(\ALC_{\text{concepts}})$ is   reception-compatible in the class of finite sets of finite tree-shaped pointed interpretations over any finite signature.
	\end{theorem} 
	\begin{proof}
		Let $\mSet$ be a set of tree-shaped pointed interpretations $(\Imc_1,d_1),\ldots,$ $(\Imc_n,d_n)$ over a finite signature $\Sigma$ and let $C$ be an \ALC concept over $\Sigma$. We need to show that
		\(\FRsups(\modelsof{C}\cup \mSet, \logsys(\ALC_{\text{concepts}}))\neq \emptyset\).
		Let $D_i$ be the $\EL_\top$ concept with $\Imc_{D_i}$ being isomorphic to $\Imc_i$.
		We argue that $\modelsof{C\sqcup \bigsqcup^n_{i=1}  (D_i)^\dagger}$ is in
		\[ \FRsups(\modelsof{C}\cup \mSet, \logsys(\ALC_{\text{concepts}})),\] where $(D_i)^\dagger$ is as in  \cref{def:translation}, for all $1\leq i\leq n$. 
		By \cref{lem:translationconcept}, %$\modelsof{(D_i)^\dagger} = \{(\Jmc,e)\mid (\Jmc,e) \text{ and }(\Imc_i,d_i)\text{ are bisimilar}\}$.
		%As every pointed interpretation is bisimilar with itself, 
		$\mSet\subseteq \bigcup^n_{i=1} \modelsof{(D_i)^\dagger}$.
		This means that 
		\[\modelsof{C}\cup \mSet\subseteq \modelsof{C\sqcup \bigsqcup^n_{i=1}  (D_i)^\dagger}.\]
		Moreover, by~\cref{lem:translationconcept}, if a pointed interpretation belongs to
		$\modelsof{C\sqcup \bigsqcup^n_{i=1}  (D_i)^\dagger}$ but not to $\modelsof{C}\cup \mSet$ then
		it is bisimilar to a model in $\mSet$. 
		Since \ALC is invariant under bisimulations, there is no $\mSet'$ such that
		$\mSet'\in\FRsets(\logsys(\ALC_{\text{concepts}}))$ and
		\[\mSet\subseteq\mSet'\subset \modelsof{C\sqcup \bigsqcup^n_{i=1}   (D_i)^\dagger}.\] 
		So \(\modelsof{C\sqcup \bigsqcup^n_{i=1} (D_i)^\dagger}\in \FRsups(\modelsof{C}\cup \mSet, \logsys(\ALC_{\text{concepts}})),\) as required. 
	\end{proof}
	The next theorem can be proved in a similar way.
	\begin{theorem}\label{thm:alcposeviction}
		$\logsys(\ALC_{\text{concepts}})$ is   eviction-compatible in the class of finite sets of finite tree-shaped pointed interpretations over a finite signature.
	\end{theorem}
	\begin{proof}
		Let $\mSet$ be a set of tree-shaped pointed interpretations $(\Imc_1,d_1),\ldots,$ $(\Imc_n,d_n)$ over a finite signature $\Sigma$ and let $C$ be an \ALC concept over $\Sigma$. We need to show that
		\(\FRsubs(\modelsof{C}\setminus \mSet, \logsys(\ALC_{\text{concepts}}))\neq \emptyset\).
		Let $D_i$ be the $\EL_\top$ concept with $\Imc_{D_i}$ being isomorphic to $\Imc_i$.
		We argue that $\modelsof{C\sqcap \bigsqcap^n_{i=1} \neg (D_i)^\dagger}$ is in
		\[ \FRsubs(\modelsof{C}\setminus \mSet, \logsys(\ALC_{\text{concepts}})),\] where $(D_i)^\dagger$ is as in  \cref{def:translation}, for all $1\leq i\leq n$. 
		By \cref{lem:translationconcept}, %$\modelsof{(D_i)^\dagger} = \{(\Jmc,e)\mid (\Jmc,e) \text{ and }(\Imc_i,d_i)\text{ are bisimilar}\}$.
		%As every pointed interpretation is bisimilar with itself, 
		$\mSet\subseteq \bigcup^n_{i=1} \modelsof{(D_i)^\dagger}$.
		This means that 
		\[\modelsof{C\sqcap \bigsqcap^n_{i=1} \neg (D_i)^\dagger}\subseteq \modelsof{C}\setminus \mSet.\]
		Moreover, by~\cref{lem:translationconcept}, if a pointed interpretation belongs to
		$\modelsof{C}\setminus \mSet$ but not to $\modelsof{C\sqcap \bigsqcap^n_{i=1} \neg (D_i)^\dagger}$ then
		it is bisimilar to a model in $\mSet$. 
		Since \ALC is invariant under bisimulations, there is no $\mSet'$ such that
		$\mSet'\in\FRsets(\logsys(\ALC_{\text{concepts}}))$ and
		\[\modelsof{C\sqcap \bigsqcap^n_{i=1} \neg (D_i)^\dagger}\subset\mSet'\subseteq\mSet.\] So \(\modelsof{C\sqcap \bigsqcap^n_{i=1} \neg (D_i)^\dagger}\in \FRsubs(\modelsof{C}\setminus \mSet, \logsys(\ALC_{\text{concepts}}))\).
	\end{proof}
\end{toappendix}

\section{Model Revision}\label{sec:modelrev}

In this section, we introduce a new kind of model change operation, which we call \emph{model revision}. 
%Whereas eviction and reception, respectively, only removes or incorporates interpretations, model revision assemble these two dimensions: it 
Model revision incorporates a set of models while also guaranteeing that another   set of models is removed.
%(equivalently, not incorporated). 
%
For instance, in \cref{ex:koala_rev},  
%\todo{fixed typos, wrong ref, rem repetition, shortened a bit}
$(\Imc_4,d'')$ had to be removed, while $(\Imc_5,d'')$ had to be added. 
Revision cannot be defined by assembling reception and eviction, as reception can add models required to be evicted and vice versa. 
For revision, we consider %classes of \emph{pairs of models}, that is, belief 
change operators defined on a %(binary) 
class $\Cmc \subseteq \powerset(\mUni) \times \powerset(\mUni)$ of \emph{pairs of models}. 
A class of pairs of models is called a \emph{binary class}.
%  (see \cref{ex:rev_noevicrecp}).  
%Such a na\"{i}ve strategy also violates minimality, as it potentially adds and removes more models than necessary. 

\begin{example}\label{ex:rev_noevicrecp}
	Let $\baseb = \{ \exists r. \top \}$ be an $\EL_{\bot}$ concept on the signature $\NC = \{A\}$ and $\NR = \{r\}$. %as described on \cref{ex:evic}, 
	Let %\todo[color=red]{rever o exemplo}
	$(\Imc_1, d_1)$ and $(\Imc_2, d_1)$ be pointed models with  $\Delta^{\Imc_1} = \{d_1, d_2\}, \Delta^{\Imc_2}= \{d_1, d_2, d_3\}$ and  
	\begin{align*}
		\Imc_1 & : \boxed{\begin{aligned}[t]
				A^{\Imc_1} = \{d_2\}, r^{\Imc_1} = \{(d_1,d_2)\}
		\end{aligned}}\\
		\Imc_2 & : \boxed{\begin{aligned}[t]
				A^{\Imc_2} = \{d_1\}, r^{\Imc_2} = \{(d_1,d_2), (d_2, d_3)
				\}
		\end{aligned}}
	\end{align*}
	%   \todo{fixed the order and typos}
	We want to revise $\baseb$ with $(\{(\Imc_1,d_1)\}, \{(\Imc_2,d_1)\})$, that is, receive $(\Imc_1,d_1)$ and evict $(\Imc_2,d_1)$. 
	%Eviction of $(\Imc_2, d_1)$ followed by the reception of $(\Imc_1,d_2)$ does not work, nor first reception of $(\Imc_1, d_1)$ followed by the eviction of $(\Imc_2, d_1)$.
	Combining rational eviction with rational reception, in any order, is not strong enough to achieve revision. 
	A rational eviction of $\baseb$ with $(\Imc_2, d_1)$ is $\baseb' = \{\exists r^3. \top \}$. %as $\FRsubs(\modelsof{\baseb}\setminus\{(\Imc_2, d_1)\}) = \{\modelsof{\baseb'}\}$. 
	%\todo{why? isnt there the possibility of using   $\exists r.A$ to remove?} 
	However, incorporating $(\Imc_1, d_1)$ to it yields the base $\{\exists r. \top\}$ again, which contains $(\Imc_2, d_1)$. 
	On the other hand,  reception of $\baseb$ with $(\Imc_1, d_1)$ does not change $\baseb$, as $(\Imc_1,d_1)$ is a model of $\exists r.\top$. 
	Eviction of $\baseb$ with $(\Imc_2, d_1)$ gives $\exists r^3. \top$, which does not contain $(\Imc_1, d_1)$. 
	%We want to revise $\baseb$ with $(\{(\Imc_2,d_1)\}, \{(\Imc_1,d_1)\})$, that is, evict $(\Imc_2,d_1)$ and receive $(\Imc_1,d_1)$. 
	%Eviction of $(\Imc_2, d_1)$ followed by the reception of $(\Imc_1,d_2)$ does not work, nor first reception of $(\Imc_1, d_1)$ followed by the eviction of $(\Imc_2, d_1)$.
	%Every eviction of $\baseb$ with $(\Imc_2, d_1)$ results in a  base  equivalent to $\baseb' = \{\exists r^3. \top \}$. 
	%However, incorporating $(\Imc_1, d_1)$ to it yields the base $\{\exists r. \top\}$ again, which contains $(\Imc_2, d_1)$. 
	%On the other hand,  reception of $\baseb$ with $(\Imc_1, d_1)$ does not change $\baseb$, as $(\Imc_1,d_1)$ is a model of $\exists r.\top$. 
	%Eviction of $\baseb$ with $(\Imc_2, d_1)$ gives $\exists r^3. \top$, which does not contain $(\Imc_1, d_1)$. 
	%So, combining eviction with reception, in any order, is not strong enough to achieve revision. 
	%\todo[inline]{we could use this same example to illustrate the following revision in this section, the results would be $\exists r. A$}
\end{example}
\begin{definition}
	A revision operator, on a binary class $\Cmc$, % \subseteq \powerset(\mUni) \times \powerset(\mUni)$, 
	is function $\rev: \finitepwset(\llang)\times \Cmc \to \finitepwset(\llang)$ which satisfies the postulate 
	\begin{description}
		\item[] \success 
		$\mathbb{M}^- \cap \modelsof{\rev(\Bmc,\mathbb{M}^+,\mathbb{M}^-)} = \emptyset$ and
		$\mathbb{M}^+ \subseteq \modelsof{\rev(\Bmc,\mathbb{M}^+, \mathbb{M}^-)}$. 
		%\todo[color=green]{add note abuse of notation}
	\end{description}
\end{definition}
%\todo{found comma before and the text was saying "first", fixed}
For clarity, we denote $\Mminus$ as the set of models to be removed, while $\Mplus$ denotes the set of models to be added. 
%The first rationality postulate of revision, 
\emph{Success}  guarantees that all models in $\Mminus$ are removed while all models in $\Mplus$ are added.  
Clearly, one cannot demand to add and remove the same model.  
%Therefore, 
Success cannot be satisfied 
%a revision operator cannot satisfy success
%\todo{changed cannot be defined with cannot have success}
for $(\baseb, \Mplus, \Mminus)$, if $\Mplus$ and $\Mminus$ are not disjoint. 
%These anomalous cases, therefore, must be ruled out from the class of models $\Cmc$. 
Unfortunately, 
%identifying such anomalous cases
the possibility of
success
goes beyond 
identifying whether or not $\Mplus$ and $\Mminus$ are disjoint. 
%share some models. 
%This occurs due to the 
It depends on the
logic's underlying satisfaction system.
%and their logics. 
%Different logics and satisfaction systems, present strong properties relating their interpretations. 
For example, for \ALC concepts, models are closed under bisimulation~\citep{DBLP:books/el/07/GorankoO07,DBLP:books/el/07/BBW2007}, which means that some distinct but bisimilar pointed interpretations $(\Imc_1,d_1)$ and $(\Imc_2,d_2)$ satisfy precisely the same formulae. 
Thus, a revision that demands incorporation of $(\Imc_1,d_1)$ and removal of $(\Imc_2,d_2)$ cannot occur. % not possible. 
% Moreover, as discussed in \cref{sec:model_change}, some models cannot be distinguished from others. For example, in \cref{ex:nonunique}, there is no concept satisfied by interpretation $\Imc_1$ that is violated by $\Imc_2$. So, reception of $\Imc_1$ enforces reception of $\Imc_2$. Consequently, a revision cannot be performed by incorporating $\Imc_1$ and eliminating $\Imc_2$. 
For success, revision must at least be defined on  classes of pairs of models in which incorporating $\Mplus$ does not conflict with eliminating $\Mminus$. 
We call such classes \emph{revision-realisable}. 

\begin{definition}
	A binary class of models $\Cmc$ is revision-realisable iff for all $(\Mplus, \Mminus) \in \Cmc$, there is a finitely representable set $\mSet$ of 
	models
	%    interpretations
	such that 
	$\Mplus \subseteq \mSet$ and $\Mminus \cap \mSet = \emptyset$.
\end{definition}
%\todo{...}
%Even when considering only revision-realisable classes, 

For conciseness, unless otherwise explicitly stated, we assume that all binary classes of models are revision-realisable. 
\emph{Success} alone is not enough to bring rationality to revision. 
We introduce other rationality postulates to capture the minimal change principle for revision. 
We start with %the following two,

\begin{description}
	%     \item[] \textbf{success:} 
	% $\mathbb{M}^- \cap \modelsof{\rev(\Bmc,\mathbb{M}^+,\mathbb{M}^-)} = \emptyset$ and
	% $\mathbb{M}^+ \subseteq \modelsof{\rev(\Bmc,\mathbb{M}^+, \mathbb{M}^-)}$, 

	\item[]  \vacuousexpansion: if $\mathbb{M}^+ \subseteq \mod(\Bmc)$ then $\modelsof{\rev(\Bmc,\mathbb{M}^+, \mathbb{M}^-)} \subseteq \mod(\Bmc)$,
	
	\item[] \vacuousremoval: if $\mathbb{M}^- \cap \mod(\Bmc) = \emptyset$ then $\mod(\Bmc) \subseteq \modelsof{\rev(\Bmc, \mathbb{M}^+, \mathbb{M}^-)}$.
\end{description}

If all models of $\Mplus$ are models of  $\baseb$, one should  only evict $\Mminus$ (\emph{vacuous expansion}). 
On the other hand, if none of the models in $\Mminus$ satisfy $\baseb$ \emph{(vacuous removal)}, all we need to do is to add models.
%there is nothing to be removed, and one should concentrate their efforts on incorporating $\Mplus$ 
%. 
%In the special case that 
If all models in $\Mplus$ satisfy $\baseb$, and none of  $\Mminus$ violate $\baseb$, the base $\baseb$ should be left untouched. 
We call this \emph{lethargy}. 

\begin{description}
	\item[] \lethargy: if $\Mplus \subseteq \modelsof{\baseb}$ and $\Mminus\cap \modelsof{\baseb
	} = \emptyset$, then $\rev(\baseb, \Mplus, \Mminus) = \modelsof{\baseb}$.
\end{description}

%The postulates \vacuousexpansion  and \vacuousremoval jointly imply \lethargy. 

\begin{propositionrep}
	If a revision operator %$\rev$ 
	satisfies   \vacuousexpansion and \vacuousremoval then it %$\rev$ 
	satisfies \lethargy.
\end{propositionrep}
\begin{proof}
	Let $\rev$ be a revision function, on a revision-realisable binary class of models $\Cmc$, satisfying both \vacuousexpansion and \vacuousremoval. 
	Let $\baseb$ be a finite base, $(\Mplus, \Mminus)\in \Cmc$ such that  
	$\Mplus \subseteq \modelsof{\baseb}$ and $\Mminus \cap \modelsof{\baseb} = \emptyset$. 
	From \vacuousexpansion, $\modelsof{\revline} \subseteq \modelsof{\baseb}$
	while from \vacuousremoval we have $\modelsof{\baseb} \subseteq \modelsof{\revline}$. 
	Therefore, $\modelsof{\revline} = \modelsof{\baseb}$. 
\end{proof}

These three postulates capture the most fundamental features of revision. 
Yet, such postulates are not enough, as they allow for drastic removal and addition of models.
For example, if some of the models in $\Mplus$ does not satisfy $\baseb$ and  some model in $\Mminus$ violates $\baseb$ then 
%the premisses of
\vacuousexpansion and \vacuousremoval allow the removal of all models of $\baseb$.
%, as long as some model in $\Mminus$ violates $\baseb$. 
%For example, if none of the models in $\Mplus$ satisfy $\baseb$, 
%both \vacuousexpansion and \vacuousremoval allow the removal of all models of $\baseb$, as long as some model in $\Mminus$ violates $\baseb$. 
However, ideally, changes should be minimised. 
% Instead of imposing finite retainment and finite temperance, we shall move to the more challenging cases: 
%simply upper-bound or lower-bound the revision changes the models of $\baseb$. However, they do not limit  do not capture the  notion
%The challenges emerge when one must both incorporate and remove interpretations. 

In the case that the trivial removal of $\Mminus$ and the trivial addition of $\Mplus$ reach a finitely representable set $\mSet$, the revision should correspond to $\mSet$ as finiteness is trivially obtained.  However, if finiteness is not reached in this way, then one must remove and add some extra interpretations in favour of finiteness. 
Such additions and removals must be minimised. 
Such a minimal change principle is conceptualised in the form of the \prudence postulate.  
%in the form of the following postulate: 

\begin{description}
	\item[] \prudence: if $X^+, X^- \in \Cmc$ are disjoint sets and 
	conditions (1) to (3) below are jointly satisfied, then condition (4) is satisfied: 
	\begin{enumerate}
		\item 
		$X^- \subseteq \big(\modelsof{\Bmc} \setminus \modelsof{\rev(\Bmc, \M^+, \mathbb{M}^-}\big)$, and 
		$\M^-\cap \modelsof{\baseb} \subseteq X^-$ 
		
		\item 
		$\M^+ \setminus \modelsof{\baseb} \subseteq X^+ \subseteq 
		\big(\modelsof{\rev(\Bmc, \M^+, \M^-} \big) \setminus \modelsof{\Bmc}$ and 
		
		\item  $\big((\modelsof{\Bmc} \setminus X^-) \cup X^+\big) \neq \modelsof{\rev(\Bmc, \M^+, \M^-}$
		
		\item $ \big((\modelsof{\Bmc} \setminus X^-) \cup X^+\big) \not \in \FRsets(\logsys)$.

	\end{enumerate} 
\end{description}

In \prudence above, the set $X^-$ (condition 1) denotes the extra interpretations to be removed, while $X^+$ (condition 2) denotes the extra interpretations to be added during revision. 
These extra additions and removals can only occur in favour of finiteness, and all of them must be necessary to achieve finiteness. 
Therefore, every smaller combination of removals or additions to form a revision candidate (condition 3) does not reach finiteness (condition 4). 
The postulate \prudence captures \lethargy. % also \prudence

\begin{propositionrep}
	If a revision operator 
	%$\rev$ 
	satisfies \prudence then it 
	%$\rev$ 
	satisfies \lethargy.
\end{propositionrep}
\begin{proof}
	Let us suppose for contradiction purposes that $\rev$ satisfies \prudence but it does not satisfy \lethargy. 
	Thus, for some finite base $\baseb$ and disjoint sets of models $\Mplus$ and $\Mminus$, $\modelsof{\baseb} \neq \modelsof{\revline}$. 
	Let $X^- = \emptyset$ and $X^+ = \emptyset$. It follows from \prudence that 
	$$ (\modelsof{\baseb}\setminus X^-) \cup X^+ \not \in \FRsets(\logsys).$$
	
	As $X^+ = X^- = \emptyset$, we have 
	\begin{align*}
		(\modelsof{\baseb}\setminus X^-) \cup X^+ = \modelsof{\baseb}.
	\end{align*}
	Thus, $\modelsof{\baseb} \not \in \FRsets(\logsys)$ which contradicts the hypothesis that $\baseb$ is a finite base. 
\end{proof}

% There is one more postulate that we require: \emph{\uniformity}. Uniformity, analogous to its form for eviction and reception, will ensure that revision operators are neither syntax sensitive nor model structure sensitive. 
% Its formalisation, however, involves some complex devices, which we will postpone towards the end of this section. 

Defining operators capable of satisfying principles of minimal change has been proved a challenge in the field of \emph{belief change}. 
In some logics, operators satisfying minimal change principles cannot even be defined \citep{flouris2006, RibeiroNW18, ouraaai}. 
This occurs due to the strong semantics and properties of the logics. 
\citet{ouraaai} have shown that in some logics satisfaction systems, eviction and reception do not exist. 
For revision, this would not be different. 
We shall direct the effort of defining revision operators to 
classes of models that are compatible with such rationality postulates. 

\begin{definition}\label{def:revisionmod}
	A binary class 
	%of binary models 
	$\Cmc$ is revision-compatible iff there is a revision operator 
	%$\rev$ 
	on $\Cmc$ satisfying \success, \vacuousexpansion, \vacuousremoval and \prudence.
\end{definition}

% It is worth highlighting that binary classes of models that are revision-compatible are also revision-realisable, since \success is necessary to the existence of revision operators. 

% The main challenge, however, lies in satisfying \prudence, as revision operators satisfying both \vacuousexpansion and \vacuousremoval exist in every revision-realisable binary class of models. 

% \begin{theorem}
	%     If a binary class of models $\Cmc$ is revision-realisable, then there exists one revision function that satisfies both \vacuousexpansion and \vacuousremoval.
	% \end{theorem}
% \begin{proof}
	
	% \end{proof}

% , satisfying such rationality postulates, to suitable 
% our efforsand before we can define revision operators satisfying all such rationality postulates, we shall first frame in which classes of models revision []

% Indeed, for revision, it will not be possible to define revision operators that satisfy all principles of minimal change. 
% In fact, even on revison-realisable classes of model
% Precisely, while all 

Given that we are working on a binary class of models that is revision-compatible, we would like to know how to construct a revision operator satisfying all the rationality postulates presented so far.
We will frame the precise class of operators that satisfy such postulates.  
For this, 
we need to define some auxiliary tools. % before introducing this novel class. 
%
%For sets of interpretations $\Mplus$ and $\Mminus$, 
Let 
$$ \dualcon(\M^+, \M^-) = \{ Y \in \FRsets(\logsys) \mid  \M^+ \subseteq Y \mbox{ and } \M^- \cap Y = \emptyset \},$$
be the set which contains exactly all finite bases satisfied by all models in a given set $\Mplus$ but violated by all models in $\Mminus$. 
One can regard $\dualcon(\Mplus, \Mminus)$ as the set of all potential candidates to revise a base with the pair $(\Mplus, \Mminus)$. 

%The postulates above capture principles of minimal change. However, it is unclear how one can construct revision operators that abide by such principles. 
Not all sets in $\dualcon(\Mplus, \Mminus)$, however, are suitable to revise a given base $\baseb$, as some of them might add or remove more than allowed by \prudence. 
%Intuitively, when revising a base $\baseb$ by a pair of sets $(\Mplus, \Mminus)$, one must minimise the change on $\modelsof{\baseb}$, so that $\Mplus$ is incorporated, $\Mminus$ is relinquished, and a finite base is reached. 
If we could ``measure'' the changes incurred on $\modelsof{\baseb}$ to achieve a finite representable set of models, then we can just select the finite representable sets with the minimal incurred changes. 
The symmetric difference between two sets $A$ and $B$ provide exactly the changes necessary to turn one set into another. 
To turn a set $A$ into a set $B$, we need only to add the elements of $B$ that are not in $A$, and remove the elements in $A$ that are not in $B$. 
We can, therefore, use the symmetric difference to ``measure'' changes between sets of interpretations, and choose those that minimise the changes.  
On each set of interpretations $\mSet$, 
we define the relation $\preceq_{\mSet} \subseteq \powerset(\mUni) \times \powerset(\mUni)$, % between sets of interpretations, 
such that 
$ \mSet_1 \preceq_{\mSet} \mSet_2 \mbox{ iff } (\mSet \oplus \mSet_1) \subseteq (\mSet \oplus \mSet_2).$

Intuitively, $\mSet_1 \leqslant_{\mSet} \mSet_2$ means that turning $\mSet$ into $\mSet_2$ incurs in at least as much change as turning $\mSet$ into $\mSet_1$. 
This means that turning $\mSet$ into $\mSet_1$ is equally cheap or cheaper than turning $\mSet$ to $\mSet_2$. 
We can use this closeness relation to revise a base $\baseb$ by a pair $(\Mplus, \Mminus)$. 
% The strategy consus
% Given a
% To revise a base $\baseb$ with a pair $(\Mplus, \Mminus)$, we cam simple choose 
%
%We define a revision operator upon such a relation. 
%We first give the general intuition of the construction before defining it. 
%The strategy works in the following way.  
%To revise a base  $\baseb$ with  $(\Mplus, \Mminus) \in \Cmc$ be a pair for revision. 

As the revision must minimise the changes, we choose, from $\dualcon(\Mplus, \Mminus)$, one of the closest options to $\modelsof{\baseb}$ that is,  
%we choose one option 
one from  $\min_{\preceq_{\modelsof{\baseb}}}(\dualcon(\Mplus, \Mminus)).$
%
%
% \begin{enumerate}
	%     \item we add $\Mplus$ to the models of $\baseb$ and remove $\Mminus$ from it, that is, we produce the set
	%     $\mSet' = (\modelsof{\baseb} \cup \Mplus) \setminus \Mminus$
	%     \item as the revision must minimise the changes, we choose from $\dualcon(\Mplus, \Mminus)$ the closest options to $\mSet'$ w.r.t., that is,  we select from 
	%     $min_{\preceq_{\mSet'}}(\dualcon(\Mplus, \Mminus))$
	% \end{enumerate}
% A selection function realises the choice. 
% \begin{definition}
	%     A selection function is a map $\sel: \powerset(\powerset(\mUni)) \to \powerset(\mUni)$ such that 
	%     if $X \neq \empty$ then $\sel(X) \neq \emptyset$.
	% \end{definition}
% From a non-empty list of options, a choice function chooses exactly one element of the list. 
To define a revision function using this strategy, we need the condition that for every base $\baseb$ and pair $(\Mplus, \Mminus)$, there is at least one choice on 
$\mmin_{\preceq}(\dualcon(\Mplus, \Mminus))$, that is,  the symmetric difference indeed minimises the distance from the base to the revision candidates. 
This condition follows from  %is satisfied whenever the class $\Cmc$ is 
revision-compatibility.

\begin{toappendix}
	
	\begin{lemma}\label{lem:finretcomp}
		Let $Z, Z'$ be sets of models and $\baseb$ a finite base such that 
		$\modelsof{\baseb} \oplus Z' \subset \modelsof{\baseb} \oplus Z$. All the following hold:
		\begin{enumerate}
			\item If $\Mminus \cap Z' = \emptyset$, then $\Mminus \cap \modelsof{\baseb} \subseteq \modelsof{\baseb}\setminus Z'$;
			
			\item $\modelsof{\baseb} \setminus Z' \subseteq \modelsof{\baseb}\setminus Z$;
			
			\item $Z' \setminus \modelsof{\baseb} \subseteq Z \setminus \modelsof{\baseb}$.
		\end{enumerate}
	\end{lemma}
	\begin{proof}
		Let $Z, Z'$ be sets of models and $\baseb$ a finite base such that 
		$\modelsof{\baseb} \oplus Z' \subset \modelsof{\baseb} \oplus Z$. 
		\begin{enumerate}
			\item Let us assume that $\Mminus \cap Z' = \emptyset$. Let $x \in \Mminus \cap \modelsof{\baseb}$. We will show that $x \in \modelsof{\baseb}\setminus Z'$. 
			As $x \in \Mminus \cap \modelsof{\baseb}$, it follows that $x \in \modelsof{\baseb}$ and $x \in \Mminus$. By hypothesis, $\Mminus \cap Z' = \emptyset$, which implies that $x \not \in Z'$. Therefore, $x \in \modelsof{\baseb}$ but $x \not \in Z'$ which implies that $x \in \modelsof{\baseb}\setminus Z'$. 
			
			\item let $x \in \modelsof{\baseb}\setminus Z'$. We will show that $x \in \modelsof{\baseb} \setminus Z$. 
			From the definition of symmetric difference, we have 
			$\modelsof{\baseb} \oplus Z' = (\modelsof{\baseb}\setminus Z') \cup (Z'\setminus \modelsof{\baseb})$ which implies 
			$$\modelsof{\baseb}\setminus Z' \subseteq \modelsof{\baseb} \oplus Z'.
			$$
			Thus, $x \in \modelsof{\baseb} \oplus Z'$. 
			By hypothesis,  $\modelsof{\baseb} \oplus Z' \subseteq \modelsof{\baseb} \oplus Z$ which implies that 
			$x \in \modelsof{\baseb} \oplus Z$. 
			By definition, 
			$\modelsof{\baseb} \oplus Z = (\modelsof{\baseb}\setminus Z) \cup (Z \setminus \modelsof{\baseb})$. 
			Thus, 
			(a) $x \in \modelsof{\baseb}\setminus Z$ or
			(b) $x \in Z \setminus \modelsof{\baseb}$.  Case (a) is trivial, as it is exactly the condition we want to prove. Case (b) leads to a contradiction, and we conclude the proof. 
			For case (b), we have that $x \in Z$ but $X \not \in \modelsof{\baseb}$ which contradicts the hypothesis that $x \in \modelsof{\baseb}\setminus Z'$. 
			
			\item let $x \in Z' \setminus \modelsof{\baseb}$. We will show $x \in  Z \setminus \modelsof{\baseb}$. 
			From the definition of symmetric difference, we have 
			$\modelsof{\baseb} \oplus Z' = (\modelsof{\baseb}\setminus Z') \cup (Z'\setminus \modelsof{\baseb})$ which implies 
			$$Z' \setminus \modelsof{\baseb} \subseteq \modelsof{\baseb} \oplus Z'.
			$$ 
			Thus, as $x \in \modelsof{\baseb} \oplus Z'$. 
			By hypothesis,  $\modelsof{\baseb} \oplus Z' \subseteq \modelsof{\baseb} \oplus Z$ which implies that 
			$x \in \modelsof{\baseb} \oplus Z$. 
			By definition, 
			$\modelsof{\baseb} \oplus Z = (\modelsof{\baseb}\setminus Z) \cup (Z \setminus \modelsof{\baseb})$. 
			Thus, (a) $x \in \modelsof{\baseb}\setminus Z$ or
			(b) $x \in Z \setminus \modelsof{\baseb}$.  
			Case (b) is trivial, as it is exactly the condition we want to prove. 
			Case (a) leads to a contradiction, and we conclude the proof. 
			For case (a), we have that $x \in \modelsof{\baseb}$, which contradicts the hypothesis that $x\in Z'\setminus \modelsof{\baseb}$. 
		\end{enumerate}
	\end{proof}
	
	\begin{lemma}\label{lem:circ_to_sym}
		If a revision operator $\rev$ satisfies \prudence then $$\modelsof{\revline} \in \mmin_{\preceq_{\modelsof{\baseb}}}(\dualcon(\Mplus, \Mminus)).$$
	\end{lemma}
	\begin{proof}
		
		For conciseness, let 
		$$Z = \modelsof{\revline}.$$ 
		Let us suppose, for contradiction purposes, that 
		$Z \not \in \mathrm{min}_{\preceq_{\modelsof{\baseb}}}(\chi(\logsys, \Mplus, \Mminus))$. 
		Thus, (i) $Z \not \in \chi(\logsys, \Mplus, \Mminus)$ or (ii) $Z\in \chi(\logsys, \Mplus, \Mminus)$ but it is not minimal modulo $\preceq_{\modelsof{\baseb}}$. The case (i) is ruled out, as by hypothesis $\rev$ satisfies success. So, case (ii) must hold. Thus, 
		there is some $Z' \in \chi(\logsys, \Mplus, \Mminus)$ such that 
		$$\modelsof{\baseb}\oplus Z' \subset \modelsof{\baseb} \oplus Z. $$
		Therefore,
		\begin{align}
			Z \neq Z'. \label{eq:czdifzprime}
		\end{align}
		From the definition of $\chi$, we have that 
		$ \Mplus \subseteq  Z'. $
		Therefore, 
		\begin{align}
			\Mplus \setminus \modelsof{\baseb} \subseteq Z'\setminus \modelsof{\baseb}, \label{eq:csubsmpluszede}   
		\end{align}
		and 
		\begin{align}
			Z' \in \FRsets(\logsys). \label{eq:cZprimefinite}
		\end{align}
		From \cref{lem:finretcomp} (item 3), we have 
		\begin{align}
			Z' \setminus \modelsof{\baseb} \subseteq Z \setminus \modelsof{\baseb}.\label{eq:csubszzedeprime}
		\end{align}
		
		Let $X^+ = Z' \setminus \modelsof{\baseb}$. Replacing $X^+ = Z' \setminus \modelsof{\baseb}$ in 
		\eqref{eq:csubsmpluszede} and \eqref{eq:csubszzedeprime} yields
		\begin{align}
			\Mplus \setminus \modelsof{\baseb} \subseteq X^+ \subseteq Z \setminus \modelsof{\baseb}.\label{eq:ccond2}
		\end{align}
		
		From the definition of $\chi$, 
		$$ \Mminus \cap Z' = \emptyset.$$
		which from \cref{lem:finretcomp} (item 1) implies that 
		\begin{align}
			\Mminus\cap \modelsof{\baseb} \subseteq \modelsof{\baseb} \setminus Z'. \label{eq:cpat1c1}
		\end{align}
		From \cref{lem:finretcomp} (item 2), we get 
		\begin{align}
			\modelsof{\baseb}\setminus Z' \subseteq \modelsof{\baseb} \setminus Z.\label{eq:cpat2c1}
		\end{align}
		Let $X^- = \modelsof{\baseb}\setminus Z'$. Replacing it at \eqref{eq:cpat1c1} and \eqref{eq:cpat2c1} yields 
		\begin{align}
			\Mminus\cap \modelsof{\baseb} \subseteq X^- \subseteq \modelsof{\baseb}\setminus Z. \label{eq:ccond1}
		\end{align}
		
		For every pair of sets $A$ and $B$, it holds that $A \setminus (A\setminus B) = A \cap B$, and $A = (A \cap B) \cup (A\setminus B)$. 
		
		Note that 
		\begin{align}
			(\modelsof{\baseb} \setminus X^-) \cup X^+ &= \modelsof{\baseb}\setminus (\modelsof{\baseb}\setminus Z') \cup (Z'\setminus \modelsof{\baseb}) \nonumber\\
			&= (\modelsof{\baseb} \cap Z') \cup (Z'\setminus \modelsof{\baseb}) \nonumber \\
			&= Z'. \label{eq:czprimeineq}
		\end{align}
		From \eqref{eq:czdifzprime} and \eqref{eq:czprimeineq}, we get 
		\begin{align}
			(\modelsof{\baseb} \setminus X^-) \cup X^+ \neq Z. \label{eq:ccond3}
		\end{align}
		
		Due to \eqref{eq:ccond1}, \eqref{eq:ccond2} and \eqref{eq:ccond3}, it follows from \prudence that 
		$(\modelsof{\baseb} \setminus X^-) \cup X^+ \not \in \FRsets(\logsys)$. 
		Thus,  from \eqref{eq:czprimeineq}, we have that  
		$ Z' \not \in \FRsets(\logsys) $, which contradicts \eqref{eq:cZprimefinite}.
		
	\end{proof}
	
\end{toappendix}

\begin{theoremrep}\label{th:recompatb_sym}
	If $\Cmc$ is revision-compatible, then for all finite base $\baseb$ and $(\Mplus, \Mminus) \in \Cmc$, $\mmin_{\preceq_{\modelsof{\baseb}}}(\dualcon(\Mplus, \Mminus)) \neq \emptyset.$
\end{theoremrep}
\begin{proof}
	Let $\Cmc$ be revision-compatible, 
	$\baseb$ a finite base, and $(\Mplus, \Mminus) \in \Cmc$. 
	As $\Cmc$ is revision-compatible, it follows that there is some revision operator on $\Cmc$ satisfying \prudence, which implies from \cref{lem:circ_to_sym} that 
	$$\modelsof{\revline} \in \mmin_{\preceq_{\modelsof{\baseb}}}(\dualcon(\Mplus, \Mminus)).$$ 
	Thus, $\mmin_{\preceq_{\modelsof{\baseb}}}(\dualcon(\Mplus, \Mminus)) \neq \emptyset$.
\end{proof}

%According to \cref{th:recompatb_sym}, a revision operator, on a class $\Cmc$, satisfying all the rationality postulates exists exactly when the symmetric difference minimises the changes. 
%So, we can apply the strategy using the 
We get revision operators from symmetric difference. % to construct revision operators. % in these classes. 
%Intuitively, this strategy should work, because all options within $\dualcon(\Mplus, \Mminus)$ contain $\Mplus$ and have no model from $\Mminus$, while minimality is guaranteed from the relation $\preceq_{\modelsof{\baseb}}$. 
%At first glance, it would seem that the strategy above would be enough to satisfy all the postulates. 
% \todo[inline]{move the followog explanation to after the deinition}
% However, the choices must be constrained a bit further. We first formalize this na\"{i}ve strategy and identify its weakness. 
% From this, we show how to strengthen the strategy so that all postulates are satisfied. 

\begin{definition}
	A \naive relational revision operator, on a binary class of models $\Cmc$, is a map $\rev^{\oplus}: \finitepwset(\llang) \times \Cmc \to \finitepwset(\llang)$ s.t. 
	$ \modelsof{\rev^{\oplus}_{\sel}(\Mplus, \Mminus)} \in \mathsf{min}_{\preceq_{\modelsof{\baseb}}}(\dualcon(\Mplus, \Mminus)).$
\end{definition}

The \naive operator minimises the distance between the models of a base $\baseb'$ to all bases satisfied by $\Mplus$ and violated by all models in $\Mminus$. 
%The operator chooses one of the closest bases. 
Indeed, the \naive revision operators is strongly connected with \prudence. 

\begin{toappendix}
	\begin{lemma}\label{lem:lethargy_special}
		If $Z \in \min_{\preceq_{\modelsof{\baseb}}}(\dualcon( \Mplus, \Mminus))$, and 
		\begin{enumerate}
			\item $\M^-\cap \modelsof{\baseb} \subseteq X^- \subseteq \big(\modelsof{\Bmc} \setminus Z\big)$ and 
			
			\item 
			$\M^+ \setminus \modelsof{\baseb} \subseteq X^+ \subseteq 
			Z \setminus \modelsof{\Bmc}$ and 
			
			\item  $\big((\modelsof{\Bmc} \setminus X^-) \cup X^+\big) \neq Z;$
			
		\end{enumerate}
		then 
		$((\modelsof{\baseb}\setminus X^-) \cup X^+) \not \in \FRsets(\logsys)$.
	\end{lemma}
	\begin{proof}
		Let us suppose for contradiction 
		$Z \in \min_{\preceq_{\modelsof{\baseb}}}(\dualcon( \Mplus, \Mminus))$,
		conditions (1) to (3) and the statement are satisfied but $((\modelsof{\baseb}\setminus X^-) \cup X^+)  \in \FRsets(\logsys)$. 
		For clarity, let 
		\begin{align*}
			Z^- &= \modelsof{\baseb} \setminus Z\\
			Z^+ &= Z\setminus \modelsof{\baseb}\\
			Y &= \big(\modelsof{\baseb} \setminus X^-\big) \cup X^+
		\end{align*}
		
		Note that $Y \in \FRsets(\logsys)$        
		\begin{align}
			\modelsof{\baseb} \setminus Y &= \modelsof{\baseb} \setminus \Big( ( \modelsof{\baseb} \setminus X^-) \cup X^+\Big) \nonumber\\
			&= \big(\modelsof{\baseb} \setminus ( \modelsof{\baseb} \setminus X^-) \big) \setminus X^+ \label{eq:doubleminus}
		\end{align}
		By hypothesis, 
		$ X^- \subseteq \modelsof{\baseb}$, which implies from \eqref{eq:doubleminus} above that 
		\begin{align*}
			\modelsof{\baseb} \setminus Y &= X^- \setminus X^+
		\end{align*}
		By hypothesis, $X^- \cap X^+ = \emptyset$. Therefore, 
		$$ \modelsof{\baseb} \setminus Y = X^-$$
		
		\begin{align*}
			Y \setminus \modelsof{\baseb} &= \Big( \big(\modelsof{\baseb} \setminus X^-\big) \cup X^+\Big) \setminus \modelsof{\baseb} \\
			&= \Big((\modelsof{\baseb} \setminus X^-) \setminus \modelsof{\baseb}\Big) \cup \big( X^+ \setminus \modelsof{\baseb}\big)\\
			&= \emptyset \cup (X^+ \setminus \modelsof{\baseb})
		\end{align*}
		Thus, as $X^+ \cap \modelsof{\baseb} = \emptyset$, we get 
		$$ Y \setminus \modelsof{\baseb} = X^+.$$
		Therefore, 
		$\modelsof{\baseb} \oplus Y = X^- \cup X^+$.
		
		Analogously, we show that 
		$Z \oplus \baseb = Z^+ \cup Z^-$. 
		From hypothesis, 
		$X^- \subseteq Z^-$ and $X^+ \subseteq Z^+$. 
		Moreover, by hypothesis, either $X^- \subset Z^- $ or $X^+ \subset Z^+$. In either case, we get that 
		$X^- \cup X^+ \subset Z^- \cup Z^ +$. 
		Thus, 
		$\modelsof{\baseb} \oplus Y \subset \modelsof{\baseb} \oplus Z$. 
		
		Therefore, $Z \not \in \min_{\preceq_{\modelsof{\baseb}}}(\dualcon( \Mplus, \Mminus))$ which contradicts the hypothesis that $Z  \in \min_{\preceq_{\modelsof{\baseb}}}(\dualcon( \Mplus, \Mminus))$
	\end{proof}
	
\end{toappendix}

\begin{theoremrep}\label{th:naiverep}
	A revision operator satisfies \prudence iff it is a \naive relational revision operator.
\end{theoremrep}
\begin{proof}
	Let $\Cmc$ be a revision-compatible binary class of models. 
	%be a congruous model-revision function on a satisfaction system 
	%\( \logsys = (\llang, \mUni, \models)\) be a satisfaction system. 
	\begin{description}
		\item[``$\Leftarrow$''] 
		Let $\rev$ be a \naive relational revision operator, and  $(\Mplus, \Mminus) \in \Cmc$. 
		\begin{itemize}
			\item \textbf{success}. By definition, for all $Y \in \dualcon(\M^+, \M^-)$, $M^+ \subseteq Y$ and $M^- \cap Y = \emptyset$. 
			Moreover, by definition, $\modelsof{\rev(\baseb, \M^+, \M^-)} \in \dualcon(\M^+, \M^-)$. 
			Therefore, $M^+ \subseteq Y$ and $M^- \cap Y = \emptyset$. 
			\item \prudence. It follows directly from \cref{lem:lethargy_special}. 
		\end{itemize}
		\item[``$\Rightarrow$''] Let us suppose that $\rev$ satisfies both \textbf{success} and \prudence. 
		From \cref{lem:circ_to_sym}, $\modelsof{\revline}\in   \mathrm{min}_{\preceq_{\modelsof{\Bmc}}}(\dualcon(\M^+, \M^-) )  $.
	\end{description}
\end{proof}

Although \naive revision operators capture \prudence, they are too weak to satisfy \emph{vacuous-removal} and \emph{vacuous-expansion}, as 
\cref{ex:noncircums} illustrates. % a \naive operator that violates \emph{vacuous-removal}. 

\begin{example}\label{ex:noncircums}
	%\todo{precisa colocar assinatura?} Nao precisa aqui porque esse ex vale p qq assinatura 
	Let $\baseb =  (B \sqcap C)$ be an $\EL$ concept   and consider the     interpretations $\Imc_i =(\Delta^{\Imc_i},\cdot^{\Imc_i})$, with $i\in \{1,2,3,4\}$,   where $\Delta^{\Imc_i}=\{d\}$, and each $\cdot^{\Imc_i}$ is as follows.
	\begin{align*}
		\begin{aligned}[b]
			A^{\Imc_1} &= \{d\} \\
			B^{\Imc_1} &= \emptyset \\
			C^{\Imc_1} &= \{d\}
		\end{aligned}\hspace*{2ex}\vrule\hspace*{2ex}
		\begin{aligned}[b]
			A^{\Imc_2} = \emptyset \\
			B^{\Imc_2} = \emptyset \\ 
			C^{\Imc_2} = \emptyset 
		\end{aligned}\hspace*{2ex}\vrule\hspace*{2ex}
		\begin{aligned}[b]
			A^{\Imc_3} &= \emptyset \\
			B^{\Imc_3} &= \emptyset \\
			C^{\Imc_3} &= \{d\}
		\end{aligned}\hspace*{2ex}\vrule\hspace*{2ex}
		\begin{aligned}[b]
			A^{\Imc_4} &= \emptyset \\
			B^{\Imc_4} &= \{d\}\\
			C^{\Imc_4} &= \{d\}
		\end{aligned}
	\end{align*}
	By definition, $(\Imc_4,d)$ is a model of $\baseb=(B\sqcap C)$. 
	Assume we want to revise $\baseb$ with $(\{(\Imc_1,d)\}, \{(\Imc_2,d)\})$, that is, accommodate $(\Imc_1,d)$ while relinquishing $(\Imc_2,d)$. 
	Since $d\not\in (B\sqcap C)^{\Imc_2}$,  according to \emph{vacuous-removal}, we should only add models to $\baseb$.  
	Adding only $(\Imc_1,d)$ is not possible, because $(\Imc_4,d)$ is a model of $\baseb$ and every \EL concept satisfied by both $(\Imc_4,d)$ and $(\Imc_1,d)$ is also satisfied by $(\Imc_3,d)$. 
	So, every revision 
	%(to an \EL concept)
	%\todo{added}
	satisfying \emph{vacuous-removal} contains the set $\{(\Imc_1,d), (\Imc_3,d), (\Imc_4,d)\}$. This corresponds to revising $\Bmc=(B\sqcap C)$ to $C$.  
	%On the other hand, 
	If we want to avoid adding $(\Imc_3,d)$, we could remove $(\Imc_4,d)$, violating \emph{vacuous-removal}, and staying only with $(\Imc_1,d)$. 
	%  The interpretation $(\Imc_1,d)$ is finitely representable, as it yields the finite base 
	%$\baseb' = (A\sqcap C)$. 
	%
	%In fact, $\modelsof{\baseb'} \in \min_{\preceq_{\modelsof{\baseb}}}(\dualcon(\Imc_1, \Imc_2))$. 
	%To see this, note that 
	% The symmetric difference between $\modelsof{\baseb}$ and $\{(\Imc_1,d)\}$ is the set 
	%$\mSet' = \{(\Imc_4,d), (\Imc_1,d)\}$. 
	%The model $(\Imc_1,d)$ is in the symmetric difference between $\modelsof{\baseb}$ and each set that contains $(\Imc_1,d)$. 
	%This happens because $(\Imc_1,d) \not \in \modelsof{\baseb}$. 
	%In fact, $\mSet'$ is minimal w.r.t. $\preceq_{\modelsof{\baseb}}$. 
	%As $(\Imc_1,d)$ must be incorporated, it would be in the symmetric difference between $\modelsof{\baseb}$ and each potential revision candidate. 
	%So, the only candidate that would violate minimality of $\mSet'$  would be a finitely representable set $X$ whose difference symmetric with $\modelsof{\baseb}$ yields $\{(\Imc_1,d)\}$. 
	%This means that $\Imc_4 \not \in X$. 
	%So, both $\Imc_1$ and $\Imc_4$ are in $X$ which, as argued above, would imply in $\Imc_3 \in X$. This, in turn, would imply that $\Imc_3$ is in the symmetric difference of $X$ with $\modelsof{\baseb}$, which contradicts the hypothesis that the symmetric difference is $\{\Imc_1\}$. 
	So, a na\"{i}ve revision operator can output $(A\sqcap C)$ as a solution for revising $(B\sqcap C)$.  %\todo{removed repetition}
	%However, it violates \emph{vacuous removal}.
\end{example}

In \cref{ex:noncircums}, \emph{vacuous-removal} is violated, as the operator removes further models from $\baseb$ when $\Mminus$ does not violate $\baseb$. 
Also, when $\Mplus$
%$\Mplus$
%\todo{check}
satisfies $\baseb$, the na\"{i}ve operator allows adding further models. 
%The main reason is that the na\"{i}ve operator simply minimises the distance between the base and the possible solutions that contain $\Mplus$ and do not contain $\Mminus$. 
%Even if it means removing interpretations from $\baseb$, when $\Mminus$ does not violate the base, or conversely, when $\Mplus$ satisfies $\baseb$. 
%To control these cases, 
We strengthen the 
na\"{i}ve operator. % with some conditions. 

\begin{definition}\label{def:symdefrev} 
We define a symmetric-differential revision function on a binary class $\Cmc$ of models %founded on a selection function $\sel$, 
as a function 
$\rev: \finitepwset(\llang) \times \Cmc \to \finitepwset(\llang)$, such that for all $(\Mplus, \Mminus) \in \Cmc$, 
$\modelsof{\rev(\Bmc,\mathbb{M}^+, \mathbb{M}^-)} = \mSet$ where,
\begin{enumerate}[label=(\roman*),leftmargin=*]
\item if $\Mplus \subseteq \modelsof{\baseb}$, then 
$$\mSet \in  \mathrm{min}_{\preceq_{\modelsof{\Bmc}}}(\dualcon(\M^+ , \M^- \cup (\mUni\setminus \modelsof{\baseb})) ) $$

\item if $\Mplus \not \subseteq \modelsof{\baseb}$, but $\Mminus \cap \modelsof{\baseb} = \emptyset$, then 
$$ \mSet \in  \mmin_{\preceq_{\modelsof{\Bmc}}}(\dualcon(\M^+ \cup \modelsof{\baseb}, \M^-)) $$

\item otherwise, 
$\mSet \in  \mmin_{\preceq_{\modelsof{\Bmc}}}(\dualcon(\M^+, \M^-) ) $.
\end{enumerate}

\end{definition}

% We explain conditions (i) to (iii). 
% %To capture \emph{vacuos-removal} and \emph{vacuous-expansion}, a symmetric-differential operator  breaks into three cases: 
% \begin{enumerate}[label=(\roman*),leftmargin=*]
%     \item 
%     when $\Mplus$ satisfies the base,  according to \emph{vacuous-expansion}, no interpretation can be incorporated. 
%     This corresponds to enforcing all counter models of $\baseb$ to be removed jointly with $\Mminus$. 
%     %Therefore, this case reduces to a revision by     $(\Mplus, \Mminus \cup (\mUni \setminus \modelsof{\baseb}))$.
%     %Keeping $\Mplus$ is necessary to guarantee that the revision will not eliminate models from $\Mplus$ in the process. 

%     \item analogous to case (i), whenever all models of $\Mminus$ violate $\baseb$, as per  \emph{vacuous-removal}, no model of $\baseb$ should be removed. 
%     This corresponds to enforcing all models of $\modelsof{\baseb}$ to be incorporated jointly with $\Mplus$. 
%     %Hence, this case reduces to a revision with $(\Mplus \cup \modelsof{\baseb}, \Mminus)$. Keeping $\Mminus$ is necessary to avoid models of $\Mminus$ to be incorporated. 

%     \item cases (i) and (ii) gives enough protection to \emph{vacuous-removal}  and \emph{vacuous-expansion}. Hence, in any other case, we can just perform a \naive revision. 

% \end{enumerate}

For case (i), 
when $\Mplus$ satisfies the base,  according to \emph{vacuous-expansion}, no interpretation can be incorporated. 
This corresponds to enforcing all counter models of $\baseb$ to be removed jointly with $\Mminus$.  

For case (ii), analogous to case (i), whenever all models of $\Mminus$ violate $\baseb$, as per  \emph{vacuous-removal}, no model of $\baseb$ should be removed. 
This corresponds to enforcing all models of $\modelsof{\baseb}$ to be incorporated jointly with $\Mplus$.
As for case (iii), the cases (i) and (ii) give enough protection to \emph{vacuous-removal}  and \emph{vacuous-expansion}. Hence, in any other case, we can just perform a \naive revision. 
The symmetric-differential revision operators are characterised by the rationality postulates of revision. 

\begin{toappendix}

\begin{proposition}\label{prop:symrational}
Every symmetric-differential revision operator satisfies 
\emph{\success, \vacuousexpansion, \vacuousremoval} and \emph{\prudence}.
\end{proposition}
\begin{proof}
Let $\rev$ be a symmetric-differential revision function on a binary class $\Cmc$ of models, and    $(\Mplus, \Mminus) \in \Cmc$. % be disjoint sets of models in $\mUni$. 
\begin{itemize}
	\item \textbf{success}. By definition, for all $Y \in \chi(\M^+, \M^-)$, $M^+ \subseteq Y$ and $M^- \cap Y = \emptyset$. 
	Moreover, by definition, $\modelsof{\rev(\baseb, \M^+, \M^-)} \in \chi(\M^+, \M^-)$. 
	Therefore, $M^+ \subseteq Y$ and $M^- \cap Y = \emptyset$. 
	
	\item \textbf{vacuous-expansion}. Let $\M^+ \subseteq \modelsof{\baseb}$, and $\mSet \in \modelsof{\revline}$. 
	We will show that $\mSet \in \modelsof{\baseb}$.  
	So, from definition of $\rev$, we have 
	$\modelsof{\revline} \in \mathrm{min}_{\preceq_{\modelsof{\Bmc}}}(\chi(\logsys,\M^+ , \M^- \cup (\mUni\setminus \modelsof{\baseb}) ))$. 
	So, none of the models in $$\mathrm{min}_{\preceq_{\modelsof{\Bmc}}}(\chi(\logsys,\M^+ , \M^- \cup (\mUni\setminus \modelsof{\baseb}) ))$$ are in $\Mminus \cup (\mUni\setminus \modelsof{\baseb})$. Therefore, $\mSet  \in \modelsof{\baseb}$. 
	
	\item \textbf{vacuous removal:} Let $\mathbb{M}^- \cap \mod(\Bmc) = \emptyset$. 
	Either 
	(a) $\Mplus \not \subseteq \modelsof{\baseb}$ or  
	(b)$\Mplus \subseteq \modelsof{\baseb}$ 
	
	\begin{itemize}
		\item (a) $\Mplus \not \subseteq \modelsof{\baseb}$. 
		Thus, from the definition of $\rev$, we have that 
		$$\modelsof{\revline} = \mathrm{sel}\big( \mathrm{min}_{\preceq_{\modelsof{\Bmc}}}(\chi(\logsys,\M^+ \cup \modelsof{\baseb}, \M^-) ) \big),  $$
		which means that
		$$ \modelsof{\revline} \in \chi(\logsys,\M^+ \cup \modelsof{\baseb}, \M^-).$$
		
		By the definition, of $\chi$, it follows that for every $Y \in \chi(\logsys, X, X')$,  $X \subseteq Y$. 
		Therefore, as $\modelsof{\revline} \in \chi(\logsys,\M^+ \cup \modelsof{\baseb}, \M^-)$, we get 
		$$ \Mplus \cup \modelsof{\baseb} \subseteq \modelsof{\revline},$$
		
		which implies that 
		$\modelsof{\baseb} \subseteq \revline$.

		\item (b) $\Mplus \subseteq \modelsof{\baseb}$. 
		Thus, from the definition of $\rev$, we have that 
		$$ \modelsof{\revline} = \mathrm{sel}\big( \mathrm{min}_{\preceq_{\modelsof{\Bmc}}}(\chi(\logsys,\M^+ , \M^- \cup (\mUni\setminus \modelsof{\baseb})) ) \big). $$
		By hypothesis, $\Mminus \cap \modelsof{\baseb}= \emptyset$. Thus, as $\mUni\setminus \modelsof{\baseb} = \emptyset$ we have that 
		$\Mminus \cup (\mUni\setminus \modelsof{\baseb}) = \emptyset$. This jointly with the hypothesis that $\Mplus \subseteq \modelsof{\baseb}$ implies that 
		$\modelsof{\baseb} \in \chi(\logsys, \Mplus, \Mminus \cup (\mUni\setminus \modelsof{\baseb}))$. 
		Therefore, as $\modelsof{\baseb} \oplus \modelsof{\baseb} = \emptyset$, we get that $$\mathrm{min}_{\leqslant_{\modelsof{\baseb}}}(\chi(\logsys, \Mplus, \Mminus \cup (\mUni\setminus \modelsof{\baseb})) = \{\modelsof{\baseb}\}.$$
		Therefore, $\modelsof{\revline} = \modelsof{\baseb}$ which implies that $\modelsof{\baseb} \subseteq \modelsof{\revline}$. 
		
	\end{itemize}
	
	\item \prudence. 
	Let us suppose, for contradiction, that there is some \naive revision operation $\rev$ that does not satisfy \prudence. 
	Thus, there is a pair of sets $(\Mplus, \Mminus)$ and 
	$X^-$ and $X^+$ satisfying Points~1-3 of \prudence but violating point 4, that is,  
	\begin{align}
		\big((\modelsof{\Bmc} \setminus X^-) \cup X^+\big) \in \FRsets(\logsys)\label{eq:FrRestwo} 
	\end{align}

	As $\rev$ is a \naive revision function, we have that 
	\begin{equation}        
		\modelsof{\revline} \in \mmin_{\preceq_{\modelsof{\baseb}}}(\dualcon(\Mplus, \Mminus)).\label{eq:revisRelationaltwo}    
	\end{equation}     
	For clarity, let 
	\begin{align*}
		Z^- &= \modelsof{\baseb} \setminus \modelsof{\revline}\\
		Z^+ &= \modelsof{\revline}\setminus \modelsof{\baseb}\\
		Y &= \big(\modelsof{\baseb} \setminus X^-\big) \cup X^+
	\end{align*}

	Thus, from \eqref{eq:FrRestwo}, $Y \in \FRsets(\logsys)$. 
	\begin{align}
		\modelsof{\baseb} \setminus Y &= \modelsof{\baseb} \setminus \Big( ( \modelsof{\baseb} \setminus X^-) \cup X^+\Big) \nonumber\\
		&= \big(\modelsof{\baseb} \setminus ( \modelsof{\baseb} \setminus X^-) \big) \setminus X^+ \label{eq:doubleminustwo}
	\end{align}
	By hypothesis, 
	$ X^- \subseteq \modelsof{\baseb}$, which implies from \eqref{eq:doubleminustwo} above that 
	\begin{align*}
		\modelsof{\baseb} \setminus Y &= X^- \setminus X^+
	\end{align*}
	By hypothesis, $X^- \cap X^+ = \emptyset$. Therefore, 
	$$ \modelsof{\baseb} \setminus Y = X^-$$
	
	\begin{align*}
		Y \setminus \modelsof{\baseb} &= \Big( \big(\modelsof{\baseb} \setminus X^-\big) \cup X^+\Big) \setminus \modelsof{\baseb} \\
		&= \Big((\modelsof{\baseb} \setminus X^-) \setminus \modelsof{\baseb}\Big) \cup \big( X^+ \setminus \modelsof{\baseb}\big)\\
		&= \emptyset \cup (X^+ \setminus \modelsof{\baseb})
	\end{align*}
	Thus, as $X^+ \cap \modelsof{\baseb} = \emptyset$, we get 
	$$ Y \setminus \modelsof{\baseb} = X^+.$$
	Therefore, 
	$\modelsof{\baseb} \oplus Y = X^- \cup X^+$.
	
	Let $Z = \modelsof{\revline}$. Analogously, we show that 
	$Z \oplus \baseb = Z^+ \cup Z^-$. 
	From hypothesis, 
	$X^- \subseteq Z^-$ and $X^+ \subseteq Z^+$. 
	Moreover, by hypothesis, either $X^- \subset Z^- $ or $X^+ \subset Z^+$. In either case, we get that 
	$X^- \cup X^+ \subset Z^- \cup Z^ +$. 
	%which implies that $X^- \cup X^+ \subseteq Z^- \cup Z^+$. 
	Thus, 
	$\modelsof{\baseb} \oplus Y \subset \modelsof{\baseb} \oplus \modelsof{\revline}$. 
	
	Therefore, 
	$$\modelsof{\revline} \not \in \mmin_{\preceq_{\modelsof{\baseb}}}(\dualcon(M^+,M^-)),$$ 
	which contradicts  \eqref{eq:revisRelationaltwo}.
	
\end{itemize}
\end{proof}

\begin{proposition}\label{th:rational_symmetric}
If a revision operator $\rev$ satisfies \emph{\success, \vacuousexpansion, \vacuousremoval} and \emph{\prudence} then 
$\rev$ is a symmetric-differential operator.
\end{proposition}
\begin{proof}
Let $\Cmc$ be a revision-compatible class of models, and 
$\rev$ revision operator on $\Cmc$ satisfying  \emph{\success, \vacuousexpansion, \vacuousremoval} and \emph{\prudence}. 

Let $\baseb$ be a finite base, and $(\Mplus, \Mminus) \in \Cmc$.
It suffices to show that 
\begin{itemize}
	\item (a) $\modelsof{\revline}\in \mathrm{min}_{\preceq_{\modelsof{\Bmc}}}(\dualcon(\M^+ , \M^- \cup (\mUni\setminus \modelsof{\baseb})) ) $, if $\Mplus \subseteq \modelsof{\baseb}$;
	\item (b) $\modelsof{\revline}\in     \mathrm{min}_{\preceq_{\modelsof{\Bmc}}}(\dualcon(\M^+ \cup \modelsof{\baseb}, \M^-))$, if  $\Mplus \not \subseteq \modelsof{\baseb}$, but  $\Mminus \cap \modelsof{\baseb} = \emptyset$;

	\item (c) 
	$\modelsof{\revline}\in   \mathrm{min}_{\preceq_{\modelsof{\Bmc}}}(\dualcon(\M^+, \M^-) )  $, if $\Mplus\not \subseteq \modelsof{\baseb}$ and $\Mminus \cap \modelsof{\baseb} \neq \emptyset$. 
\end{itemize}

\begin{itemize}
	\item (a) let us suppose that $\Mplus \subseteq \modelsof{\baseb}$. 
	Thus, from vacuous expansion, $$\modelsof{\revline} \subseteq \modelsof{\baseb}.$$ 
	For conciseness, let $Z = \modelsof{\revline}$. 
	So, 
	\begin{equation}
		Z \subseteq \modelsof{\baseb}. \label{eq:zedbase}
	\end{equation}
	Let us suppose, for contradiction purposes, that 
	$$Z \not \in \mathrm{min}_{\preceq{\modelsof{\baseb}}}(\dualcon(\Mplus, \Mminus \cup (\mUni \setminus \modelsof{\baseb}))).$$
	Thus, there is some $Z' \in \dualcon(\Mplus, \Mminus \cup (\mUni \setminus \modelsof{\baseb}))$ such that 
	\begin{equation}
		Z' \oplus \modelsof{\baseb} \subset Z \oplus \modelsof{\baseb}. \label{eq:zedsubsbase}   
	\end{equation}
	It follows from the definition of $\chi$ that 
	\begin{align}
		Z' \in \FRsets(\logsys) \label{eq:zedprimefr}
	\end{align}
	$Z' \cap  \Mminus \cup (\mUni \setminus \modelsof{\baseb}) = \emptyset$. Therefore, $Z' \cap  (\mUni \setminus \modelsof{\baseb}) = \emptyset$ which implies 
	\begin{align}
		Z' \subseteq \modelsof{\baseb}\label{eq:zprimemodb}
	\end{align}
	
	From \eqref{eq:zedbase}, $Z \subseteq \modelsof{\baseb}$, which implies that 
	\begin{equation}
		Z \oplus \modelsof{\baseb} = \modelsof{\baseb} \setminus Z.  \label{eq:zedb1}  
	\end{equation}
	It follows from \eqref{eq:zprimemodb}
	that  $Z' \oplus \modelsof{\baseb} \subseteq \modelsof{\baseb}$, which implies that 
	\begin{equation}
		Z' \oplus \modelsof{\baseb} = \modelsof{\baseb} \setminus Z'. \label{eq:zedb2} 
	\end{equation}
	Replacing \eqref{eq:zedb1} and \eqref{eq:zedb2} in \eqref{eq:zedsubsbase}, we get 
	\begin{align}
		\modelsof{\baseb}\setminus Z' \subset \modelsof{\baseb}\setminus Z
	\end{align}
	
	By hypothesis, $Z' \in \dualcon(\Mplus, \Mminus \cup (\mUni \setminus \modelsof{\baseb}))$. Therefore, 
	$$Z' \cap \Mminus = \emptyset,$$
	which implies that 
	$ \Mminus \cap \modelsof{\baseb} \subseteq \modelsof{\baseb} \setminus Z'$, which jointly with \eqref{eq:zedb2} implies 
	\begin{align}
		\Mminus \cap \modelsof{\baseb} \subseteq \modelsof{\baseb} \setminus Z' \subset \modelsof{\baseb}\setminus Z. \label{eq:zedecond1}
	\end{align}
	First, recall that for all sets $A$ and $B$, $A \setminus (A \setminus B) = A \cap B$. 
	Therefore, 
	$\modelsof{\baseb} \setminus (\modelsof{\baseb}\setminus Z') = \modelsof{\baseb}\cap Z'$. 
	Let 
	$$X^- = \modelsof{\baseb}\setminus Z'. $$ 
	So, $\modelsof{\baseb}\setminus X^- = \modelsof{\baseb}\cap Z'$. 
	From \eqref{eq:zedb2}, it follows that $\modelsof{\baseb}\cap Z' \neq Z$. Therefore,  
	\begin{equation}
		(\modelsof{\baseb} \setminus X^-) \cup X^+ \neq Z. \label{eq:zedcond3} 
	\end{equation}
	Replacing $X^-= \modelsof{\baseb}\setminus Z'$ in \eqref{eq:zedecond1} yields
	\begin{align}
		\Mminus \cap \modelsof{\baseb} \subseteq X^- \subset \modelsof{\baseb}\setminus Z. \label{eq:zedecond1p1}
	\end{align}
	By hypothesis, $\Mplus \subseteq \modelsof{\baseb}$, which implies that $\Mplus\setminus\modelsof{\baseb} = \emptyset$. Thus, 
	\begin{equation}
		\Mplus\setminus\modelsof{\baseb} \subseteq X^+ \subseteq \modelsof{\revline}\setminus\modelsof{\baseb}\label{eq:zedconidtion2}
	\end{equation}
	Due to \eqref{eq:zedecond1p1}, \eqref{eq:zedconidtion2} and \eqref{eq:zedcond3}, it follows from finite retainment-temperance that 
	\begin{align}
		\modelsof{\baseb}\setminus X^- \not \in \FRsets(\logsys). \label{eq:nonfinrep}
	\end{align}
	As, from \eqref{eq:zprimemodb}, $Z' \subseteq \modelsof{\baseb}$, we get that 
	$$ \modelsof{\baseb}\setminus(\modelsof{\baseb}\setminus Z') = Z'.$$
	Therefore, $\modelsof{\baseb}\setminus X^- =Z'$ which implies from \eqref{eq:nonfinrep} that 
	$$Z' \not \in \FRsets(\logsys), $$
	which contradicts \eqref{eq:zedprimefr}.

	\item (b) Let us suppose that $\Mplus \not \subseteq \modelsof{\baseb}$, but  $\Mminus \cap \modelsof{\baseb}$. 
	For conciseness, let 
	$$Z = \modelsof{\revline}.$$ 
	As $\Mminus \cap \modelsof{\baseb} = \emptyset$, it follows from vacuous removal that 
	$$\modelsof{\baseb}\subseteq Z.$$
	Therefore, 
	\begin{align}
		\modelsof{\baseb} \oplus Z = Z\setminus \modelsof{\baseb}. \label{eq:bzede}
	\end{align}
	Let us suppose, for contradiction purposes, that 
	$$Z \not \in \mmin_{\preceq_{\modelsof{\baseb}}}(\dualcon(\Mplus \cup \modelsof{\baseb}, \Mminus)).$$
	Thus, there is a 
	$Z' \in \dualcon(\Mplus \cup \modelsof{\baseb}, \Mminus),$
	such that  
	\begin{align}
		Z' \oplus \modelsof{\baseb} \subset Z\oplus \modelsof{\baseb}. \label{eq:bzedesubs}  
	\end{align}
	From de definition of $\chi$, we get 
	\begin{align}
		Z' \in \FRsets(\logsys) \label{eq:frsetsbzede}, \mbox{ and }\\
		\Mplus \cup \modelsof{\baseb} \subseteq Z'. \label{eq:bmplusbasesubs}   
	\end{align}
	Therefore, 
	$$\modelsof{\baseb} \subseteq Z',$$ 
	which implies that 
	\begin{align}
		\modelsof{\baseb} \oplus Z' = Z' \setminus \modelsof{\baseb}. \label{eq:bzedeprime}    
	\end{align}
	Replacing \eqref{eq:bzede} and \eqref{eq:bzedeprime} in \eqref{eq:bzedesubs} renders
	\begin{align}
		Z'\setminus \modelsof{\baseb} \subset Z \setminus \modelsof{\baseb}.\label{eq:bzzprimesubs}
	\end{align}
	Therefore, 
	$$ Z \neq Z'$$
	From \eqref{eq:bmplusbasesubs} we get 
	\begin{align*}
		\Mplus &\subseteq Z'\\
		\Mplus\setminus \modelsof{\baseb} &\subseteq Z'\setminus \modelsof{\baseb}.
	\end{align*}
	This jointly with  \eqref{eq:bzzprimesubs} implies
	\begin{align*}
		\Mplus \setminus \modelsof{\baseb} \subseteq Z'\setminus \modelsof{\baseb} \subseteq Z \setminus \modelsof{\baseb}.
	\end{align*}
	Let $X^+ = Z'\setminus \modelsof{\baseb}$. 
	Thus,
	\begin{align}
		\Mplus \setminus \modelsof{\baseb} \subseteq X^+ \subseteq Z \setminus \modelsof{\baseb}. \label{eq:bcond1}
	\end{align}
	Let $X^- = \emptyset$. So, 
	\begin{align}
		\Mminus \cap \modelsof{\baseb} \subseteq X^- \subseteq \modelsof{\baseb}\setminus Z. \label{eq:bcond2}
	\end{align}
	As $\modelsof{\baseb} \subseteq Z'$, we have that
	\begin{align*}
		\modelsof{\baseb} \cup (Z'\setminus \modelsof{\baseb})  &= Z' \\
		\modelsof{\baseb} \cup X^+ &= Z'  \tag{from $X^+ = Z'\setminus \modelsof{\baseb}$} \\ 
		(\modelsof{\baseb}\setminus X^-) \cup X^+ &= Z' \tag{from $X^- = \emptyset$}
	\end{align*}
	Therefore, as $Z \neq Z'$, we get  
	\begin{align}
		(\modelsof{\baseb}\setminus X^-) \cup X^+ \neq Z.\label{eq:bcond3}    
	\end{align}
	
	Due to \eqref{eq:bcond1}, \eqref{eq:bcond2} and \eqref{eq:bcond3}, it follows from finite temperance-retainment that $(\modelsof{\baseb} \setminus X^-) \cup X^+ \not \in \FRsets(\logsys)$. 
	Therefore, 
	$Z' \not \in \FRsets(\logsys)$ which contradicts \eqref{eq:frsetsbzede}.

	\item (c) Let us suppose that $\Mplus\not \subseteq \modelsof{\baseb}$ and $\Mminus \cap \modelsof{\baseb} \neq \emptyset$. This case follows directly from \cref{lem:circ_to_sym}.
\end{itemize}
\end{proof}
\end{toappendix}

\begin{theoremrep}\label{th:symrep}
A revision operator $\rev$ satisfies
\emph{\success, \vacuousexpansion, \vacuousremoval} and \emph{\prudence}
iff it is a symmetric differential revision operator.    

\end{theoremrep}
\begin{proof}
It follows from \cref{prop:symrational} and \cref{th:rational_symmetric}.
\end{proof}

% In the special case that every additional of models is trivially finitely representable, and each removal of interpretation is also finitely representable, vacuous \vacuousexpansion and \vacuousremoval are trivially satisfied by the \naive operators, as there is no room to add or remove extra interpretations. 
% In this case, the classes of \naive relational operators and symmetric-differential operator coincide. 

\begin{toappendix}

\begin{proposition}\label{prop:spcialcase}
Let $\Cmc$ be a binary class of models, such that for every finite base $\baseb$ and  $(\Mplus, \Mminus) \in \Cmc$, 
\begin{description}
	\item[] (i) $\modelsof{\baseb} \cup \Mplus$ is finitely representable, and 
	\item[] (ii) $\modelsof{\baseb} \setminus \Mminus$ is finitely representable. 
\end{description}

A revision operator on $\Cmc$ is a \naive relational operator iff it is a symmetric-differential operator. 
\end{proposition}
\begin{proof}
Let $\Cmc$ be a binary class of models satisfying conditions (i) and (ii) from the statement. 
Let $\rev$ be a revision operator on $\Cmc$. 
\begin{description}
	\item[] ``$\Leftarrow$'' Let us suppose that $\rev$ is a symmetric-differential operator. 
	Thus, from \cref{th:symrep}, $\rev$ satisfies \prudence, which from \cref{th:naiverep} implies that $\rev$ is a \naive relational operator. 
	\item[] ``$\Rightarrow$''. Let us suppose that $\rev$ is a \naive relational operator. 
	We will show that 
	$\rev$ satisfies \vacuousremoval and \vacuousexpansion. 
	Let
	$$ Z = (\modelsof{\baseb} \cup \Mplus)\setminus \Mminus $$
	So, 
	\begin{align*}
		Z &= (\modelsof{\baseb} \setminus \Mminus) \cup (\Mplus \setminus \Mminus)
	\end{align*}
	As $\rev$ is defined on $\Cmc$, we have that $\Cmc$ is revision-realisable, which implies that 
	$\Mplus \cap \Mminus = \emptyset$. 
	So, $\Mplus \setminus \Mminus = \Mplus$, which implies 
	\begin{align}
		Z &= (\modelsof{\baseb} \setminus \Mminus) \cup \Mplus.
		\label{eq:zeddist}
	\end{align}
	By hypothesis, item (i) of the statement, 
	$\modelsof{\baseb} \cup \Mplus$ is finitely-representable. 
	Let $\baseb'$ be a finite base for this. 
	So, from item (ii) of the statement, we have that 
	$\modelsof{\baseb'}\setminus \Mminus $ is also finitely representable, which implies that 
	$Z$ is finitely representable. 
	
	We will show that 
	$$\mmin_{\preceq_{\modelsof{\baseb}}}(\dualcon(\Mplus, \Mminus)) = \{Z\}$$
	
	By construction, $\Mplus \subseteq \modelsof{Z}$ and $\Mminus \cap \modelsof{Z} = \emptyset$. Also, $Z$ is finitely representable. 
	So, $Z \in \dualcon(\Mplus, \Mminus)$. 
	Let us suppose for contradiction that $$Z \not \in \mmin_{\preceq_{\modelsof{\baseb}}}(\dualcon(\Mplus, \Mminus)).$$ 
	Then, there is some $Z' \in \dualcon(\Mplus, \Mminus)$ such that 
	$$ Z' \oplus \modelsof{\baseb} \subset Z \oplus \modelsof{\baseb} $$
	Thus, there is some $M \in Z \oplus \modelsof{\baseb}$ such that 
	\begin{align}
		M \not \in Z' \oplus \modelsof{\baseb}. \label{eq:zpcontra}
	\end{align}
	By construction, we have that 
	$$ Z \oplus \modelsof{\baseb} = (\modelsof{\baseb} \cap \Mminus) \cup (\Mplus \setminus \modelsof{\baseb}) $$
	So, 
	(a) $M \in\modelsof{\baseb} \cap \Mminus $ or 
	(b) $M \in \Mplus \setminus \modelsof{\baseb}$. 
	\begin{itemize}
		\item[] (a) $M \in\modelsof{\baseb} \cap \Mminus$. 
		So, $M \in \modelsof{\baseb}$
		As $Z' \in \dualcon(\Mplus, \Mminus)$, we get $Z' \cap \Mminus = \emptyset$.
		Thus, as $M \in \Mminus$,  we get that  $M \not \in Z'$. 
		Hence, as $M \in \modelsof{\baseb}$, we get $M \in \modelsof{\baseb} \setminus Z'$
		which implies $M \in Z' \oplus \modelsof{\baseb}$, which contradicts \eqref{eq:zpcontra}.
		
		\item (b) $M \in \Mplus \setminus \modelsof{\baseb}$. So $M \in \Mplus$ and $M \not \in \baseb$ .
		As $Z' \in \dualcon(\Mplus, \Mminus)$, we get that $\Mplus \subseteq Z'$ which implies $M \in Z'$. 
		Thus, as $M \not \in \baseb$, 
		$M \in Z' \setminus \modelsof{\baseb}$ which implies $M \in Z' \oplus \modelsof{\baseb}$ which contradicts \eqref{eq:zpcontra}.
	\end{itemize}
	
	In either case, (a) or (b), we reach a contradiction. Therefore, $$Z \in \mmin_{\preceq_{\modelsof{\baseb}}}(\dualcon(\Mplus, \Mminus)).$$ 
	
	For the uniqueness of $Z$. Let us suppose for contradiction that there is some distinct $Z' \in \mmin_{\preceq_{\modelsof{\baseb}}}(\dualcon(\Mplus, \Mminus))$. 
	
	So, as both $Z$ and $Z'$ are minimal, we get that 
	$Z \oplus \modelsof{\baseb} \not \subseteq Z' \oplus \modelsof{\baseb}$. 
	Thus,  there is some $M \in Z \oplus \modelsof{\baseb}$ such that     
	$M \not \in Z' \oplus \modelsof{\baseb}$. 
	This is the same argument to prove the minimality of $Z$. 
	So, using the same proof, we reach a contradiction. Therefore, we conclude that $Z$ is unique in $\in \mmin_{\preceq_{\modelsof{\baseb}}}(\dualcon(\Mplus, \Mminus))$, that is, 
	\begin{equation}
		\mmin_{\preceq_{\modelsof{\baseb}}}(\dualcon(\Mplus, \Mminus)) = \{Z\} \label{eq:unique}
	\end{equation}
	By definition, 
	$$\modelsof{\rev(\baseb, \Mplus, \Mminus)} \in \mmin_{\preceq_{\modelsof{\baseb}}}(\dualcon(\Mplus, \Mminus)),$$ which from \eqref{eq:unique} implies that  
	\begin{align} 
		\modelsof{\rev(\baseb, \Mplus, \Mminus)} &= Z. \label{eq:revsubbase}
	\end{align}
	
	\begin{description}
		\item \vacuousexpansion .
		Let $\Mplus \subseteq \modelsof{\baseb}$. 
		
		As $Z = \Mplus \subseteq \modelsof{\baseb}$,
		we get that 
		$Z =  \modelsof{\baseb} \setminus \Mminus$, which implies $Z \subseteq \modelsof{\baseb}$ which from \eqref{eq:revsubbase} renders 
		$\modelsof{\rev(\baseb, \Mplus, \Mminus)} \subseteq \modelsof{\baseb}$.
		
		\item \vacuousremoval . 
		Let $\Mminus \cap \modelsof{\baseb} = \emptyset$. 
		Thus, $\modelsof{\baseb} \setminus \Mminus = \modelsof{\baseb}$, which implies from \eqref{eq:zeddist}, that
		$$ Z = \modelsof{\baseb} \cup \Mplus. $$
		
		Thus, $\modelsof{\baseb} \subseteq Z$ which from \eqref{eq:revsubbase} implies   
		$\modelsof{\baseb} \subseteq \modelsof{\revline}$.
		
	\end{description}
	
	As $\rev$ is a revision operator, it satisfies success. Thus, $\rev$ satisfies success, \vacuousexpansion, \vacuousremoval and \prudence, which implies from \cref{th:symrep} that $\rev$ is a symmetric-differential operator.
	
\end{description}
\end{proof}

\begin{corollary}\label{cor:alcrev}
Let $\logsys = (\Lmc,\mUni,\models)$ be a satisfaction system, and $\Cmc \subseteq \powerset(\mUni) \times \powerset(\mUni)$ a binary class of models on $\logsys$. 
%\todo[inline]{$\Cmc \subseteq \powerset(\mUni) \times \powerset(\mUni)$}
If  for every finite base $\baseb$ and  $(\Mplus, \Mminus) \in \Cmc$, 
\begin{description}
	\item[] (i) $\modelsof{\baseb} \cup \Mplus$ is finitely representable, % by a base $\baseb'$, and 
	\item[] (ii) $\modelsof{\baseb} \setminus \Mminus$ is finitely representable, and 
	\item[] (iii) $\Mplus$ and $\Mminus$ are disjoint,
	%$(\modelsof{\baseb} \cup \Mplus)\setminus \Mminus$
	%    \item[] (iii) there are formulae $\alpha$ and $\beta$ such that 
	% $\Mplus \subseteq \modelsof{\alpha}$ and $\Mminus\cap \modelsof{\alpha} = \emptyset$.
	%$\modelsof{\baseb'} \cap \Mminus = \emptyset$
\end{description}

then $\Cmc$ is revision-compatible.
\end{corollary}
\begin{proof}
Let $\logsys = (\Lmc,\mUni,\models)$ be a satisfaction system, and $\Cmc \subseteq \powerset(\mUni) \times \powerset(\mUni)$ a binary class of models on $\logsys$ satisfying conditions (i) to (iii) from the statement. 
%Let $\baseb$ be a finite base, and $(\Mplus, \Mminus) \in \Cmc$. 
% From condition (iii), we have that $\dualcon(\Mplus, \Mminus) \not \emptyset$
We define the following revision operator on $\Cmc$,
$\rev(\baseb, \Mplus, \Mminus) = (\modelsof{\baseb} \cup \Mplus) \setminus \Mminus$.
As $\Mplus$ and $\Mminus$ are disjoint, $(\modelsof{\baseb} \cup \Mplus) \setminus \Mminus$ does not remove any model from $\Mplus$, and therefore, success is satisfied. So $\rev$ is a revision function. 
As $(\modelsof{\baseb} \cup \Mplus)$ is finitely representable by a base $\baseb'$, it follows from item (ii) that $\modelsof{\baseb'}\setminus \Mminus$ is finitely representable. 
So, $(\modelsof{\baseb} \cup \Mplus)\setminus \Mminus$ is finitely  representable by some finite base $\baseb''$. So, adding $\Mplus$ and removing $\Mminus$ are the minimal changes necessary to obtain a finite representation, which implies that 
$\min_{\preceq_{\modelsof{\baseb}}}(\dualcon(\Mplus, \Mminus) = \{\modelsof{\baseb''}\}$. 
Thus,  $\rev$ is a \naive relational revision operator. 
Hence, from \cref{prop:spcialcase}, $\rev$ is a symmetric-differential operator, and therefore, from \cref{th:symrep}, it satisfies all revision postulates. 
Thus, $\Cmc$ is revision-compatible.
\end{proof}

\end{toappendix}

Revision-compatibility is tightly connected to both reception-compatibility and eviction-compatibility.
In the case that $\Mminus$ is empty, revising a base with $(\Mplus,\Mminus)$ intuitively corresponds to performing a reception, as \vacuousremoval would forbid removal of interpretations. 
Analogously, if $\Mplus$ is empty, then revising with $(\Mplus, \Mminus)$, with a non-empty $\Mminus$, would correspond to evicting $\Mminus$, as \vacuousexpansion would forbid adding interpretations. 
From this perspective, we can trace an important connection between revision-compatibility with eviction-compatibility and reception-compatibility. 
To establish this connection, the underlying class of models must cover the cases that we can perform revision of the kind $(\Mplus, \emptyset)$ and $(\emptyset, \Mminus)$, that is, cover the possibility of solely adding or removing interpretations. 
We call such classes \decomp classes of models. 
%\todo{chamar non-trivial de decomposable, usar uma macro.}

\begin{definition}
A binary class of models $\Cmc$ is \decomp iff for all $(\Mplus, \Mminus) \in \Cmc$, %it holds that 
$(\Mplus, \emptyset) \in \Cmc$ and $(\emptyset, \Mminus) \in \Cmc$.
\end{definition}

% Whereas revision concerns pairs of sets of models, reception and eviction involve unary classes of models, as they are concerned with either solely adding interpretation or removing interpretation. 
Given a binary class of models $\Cmc$, let 
%\begin{align*}
$\Cmc^+ = \{ \Mplus \in \mUni \mid (\Mplus, \Mminus)\in \Cmc\}, \mbox{ and }$, %$\\
$\Cmc^- = \{ \Mminus \in \mUni \mid (\Mplus, \Mminus)\in \Cmc\}$. 
%\end{align*}
The set $\Cmc^+$ is the greatest subclass of $\Cmc$ with reception candidates, whereas $\Cmc^-$ is the largest subclass with eviction candidates. 
Every rational revision operator induces an eviction and a reception operator. 
Consequently, compatibility of revision implies compatibility with both reception and eviction. 

\begin{toappendix}

\begin{lemma}\label{lem:revtorcp}
Let $\Cmc$ be a \decomp revision-compatible class. 
If $\rev$ is a revision operator on $\Cmc$, then 
$ \Rcp(\baseb, \Mplus) := \rev(\baseb, \Mplus, \emptyset)$ on $\Cmc^+$ satisfies \textit{success, persistence, vacuity} and \textit{finite temperance}. 
\end{lemma}
\begin{proof}
Let us define the operation 
$\Rcp(\baseb, \Mplus) := \rev(\baseb, \Mplus, \emptyset)$. We will show that it is a reception operation, that is, that it satisfies all postulates for reception: 
\begin{itemize}
	\item \textbf{success}: follows trivially from the success of $\rev$
	
	\item \textbf{persistence}: from vacuous-removal, we have that $\modelsof{\baseb} \subseteq \modelsof{\rev(\baseb,\Mplus, \emptyset)}$. 
	Thus, $\modelsof{\baseb} \subseteq \modelsof{\Rcp(\baseb,\Mplus)}$. 
	
	\item \textbf{vacuity}: let us suppose that $\Mplus \subseteq \modelsof{\baseb}$. From vacuous-expansion, we have $\modelsof{\rev(\baseb, \Mplus, \Mminus)} \subseteq \modelsof{\baseb}$. 
	Thus, $\rev(\baseb, \Mplus), \emptyset \subseteq \modelsof{\baseb}$ which implies that 
	$\Rcp(\baseb, \Mplus) \subseteq \modelsof{\baseb}$. 
	
	\item \textbf{finite temperance}: 
	Let us suppose that finite temperance is not satisfied for contradiction purposes. 
	Thus, there are some $\M$ and $\M'$ such that 
	$$ \modelsof{\baseb} \cup \M \subseteq \M' \subset \modelsof{\Rcp(\baseb, \M)},$$
	but $\M' \in \FRsets(\logsys)$. 
	Let $X^{-} = \emptyset$ and 
	$X^{+} = \M' \setminus \modelsof{\baseb}$. 
	Condition 1 of finite-temperance, for revision, is trivially satisfied. 
	For condition 2, note that by hypothesis $\M' \subseteq \modelsof{\Rcp(\baseb, \M)} = \modelsof{\rev(\baseb, \M, \emptyset)}$ which implies that 
	$$\M' \setminus \modelsof{\baseb} \subseteq \modelsof{\rev(\baseb, \M, \emptyset)}\setminus \modelsof{\baseb}.$$ 
	Thus, as $X^{+} = \M' \setminus \modelsof{\baseb}$, we get that 
	$$X^{+} \subseteq \modelsof{\rev(\baseb, \M, \emptyset)} \setminus \modelsof{\baseb}.$$ 
	Hence, condition 2 is satisfied. 
	For condition 3, 
	as $X^- = \emptyset$, we get
	$$ (\modelsof{\baseb}\setminus X^{-}) \cup X^+ = \modelsof{\baseb} \cup X^+.$$
	
	By hypothesis, $\modelsof{\baseb} \cup \M  \subseteq \M'$ which means that 
	$\modelsof{\baseb} \subseteq \M'$. 
	Therefore, 
	$$ \M' = (\M' \setminus \modelsof{\baseb}) \cup \modelsof{\baseb} $$
	By definition, $X^+ = (\M' \setminus \modelsof{\baseb})$. 
	Thus, 
	$$ \M' = X^+ \cup \modelsof{\baseb}.$$
	
	By hypothesis, 
	$\M' \subset \modelsof{\rev(\baseb, \M, \emptyset)}$. 
	Thus, 
	$$X^{+} \cup \modelsof{\baseb} \neq \modelsof{\rev(\baseb, \M, \emptyset)}.$$
	Hence, condition 3 is also satisfied. 
	Thus, according to finite-temperance of revision, $\modelsof{\baseb} \cup X^{+} \in \FRsets(\logsys)$. 
	Thus, as $\modelsof{\baseb} \cup X^{+} = \M'$, we get that $\M' \in \FRsets(\logsys)$ which contradicts the hypothesis.

\end{itemize}
\end{proof}

\begin{lemma}\label{lem:minSymMinSups}
For every finite base $\baseb$ and set of models $\Mplus$, 
$$\FRsups(\modelsof{\baseb} \cup \Mplus) = \mathrm{min}_{\preceq_{\modelsof{\baseb}}}(\dualcon(\Mplus \cup \modelsof{\baseb}, \emptyset)).$$
\end{lemma}
\begin{proof}
We   prove that 
(a) $\FRsups(\modelsof{\baseb} \cup \Mplus) \subseteq \mathrm{min}_{\preceq_{\modelsof{\baseb}}}(\dualcon(\Mplus \cup \modelsof{\baseb}, \emptyset))$; and 
(b)$ \mathrm{min}_{\preceq_{\modelsof{\baseb}}}(\dualcon(\Mplus \cup \modelsof{\baseb}, \emptyset)) \subseteq \FRsups(\modelsof{\baseb} \cup \Mplus) $. 
\begin{description}
	\item[] (a) $\FRsups(\modelsof{\baseb} \cup \Mplus) \subseteq \mathrm{min}_{\preceq_{\modelsof{\baseb}}}(\chi(\logsys, \Mplus \modelsof{\baseb}, \emptyset)). $ 
	Let $$X \in \FRsups(\modelsof{\baseb} \cup \Mplus).$$ 
	So, $ \Mplus \subseteq X$ which implies $X \in \chi(\logsys, \modelsof{\baseb} \cup \Mplus, \emptyset)$. 
	Let us suppose for contradiction purposes that $$X \not \in \mathrm{min}_{\preceq_{\modelsof{\baseb}}}(\chi(\logsys, \Mplus \cup \modelsof{\baseb}, \emptyset))$$. 
	So, there is some $Y \in \chi(\logsys, \Mplus, \emptyset)$ such that 
	$$ \modelsof{\baseb
	} \oplus Y \subset \modelsof{\baseb} \oplus X. $$
	As $\modelsof{\baseb} \subseteq X$ and  $\modelsof{\baseb} \subseteq Y$, we get 
	$\modelsof{\baseb} \oplus Y = Y \setminus \modelsof{\baseb}$ and 
	$\modelsof{\baseb} \oplus X = X \setminus \modelsof{\baseb}$. 
	Thus, 
	\begin{align}
		Y \setminus \modelsof{\baseb} \subset Y \setminus \modelsof{\baseb}. \label{eq:stepsubsxy}  
	\end{align}
	As $\modelsof{\baseb}$ si a subset of both $X$ and $Y$, we have that 
	\begin{align*}
		Y &= (Y\setminus \modelsof{\baseb}) \cup \modelsof{\baseb}, \mbox{ and }
		X &= (X\setminus \modelsof{\baseb}) \cup \modelsof{\baseb}.
	\end{align*}
	This jointly with \eqref{eq:stepsubsxy} implies that 
	$$Y \subset X$$
	
	As $Y \in \chi(\logsys, \Mplus \cup \modelsof{\baseb}, \emptyset)$, we have that $Y \in \FRsets(\logsys)$ and $\Mplus \cup \modelsof{\baseb}$ which implies that $Y \in \mathrm{FRSups}(\modelsof{\baseb} \cup \Mplus)$. However, as $Y \subseteq X$, this contradicts the hypothesis that $X \in \FRsups(\modelsof{\baseb} \cup \Mplus)$. 
	
	\item (b) $ \mathrm{min}_{\preceq_{\modelsof{\baseb}}}(\dualcon(\modelsof{\baseb} \cup \Mplus, \emptyset)) \subseteq \FRsups(\modelsof{\baseb} \cup \Mplus) $.   
	Let us suppose, for contradiction purposes, that it does not hold. 
	So, there is  $\mSet \in \mmin_{\preceq_{\modelsof{\baseb}}}(\dualcon(\Mplus \cup \modelsof{\baseb}, \emptyset))$, such that 
	$\mSet \not \in  \FRsups(\modelsof{\baseb} \cup \Mplus)$.
	By definition,
	$\modelsof{\baseb} \cup \Mplus \subseteq \mSet$ and $\mSet$ is finitely representable. 
	Thus, there is some finitely represetnable 
	$\mSet' \subset \mSet$ such that 
	$ \modelsof{\baseb} \cup \Mplus \subseteq  \mSet' $. 
	As $\modelsof{\baseb} \cup \Mplus$ is contained on both $\mSet$ and $\mSet'$, and $\mSet' \subset \mSet$, it follows that 
	$$  \mSet' \oplus (\modelsof{\baseb} \cup \Mplus) \subset  \mSet \oplus (\modelsof{\baseb} \cup \Mplus),$$
	
	which implies that $\mSet \not \in \mmin_{\preceq_{\modelsof{\baseb}}}(\dualcon(\modelsof{\baseb} \cup \Mplus, \emptyset))$ which contradicts the hypothesis.

\end{description}
\end{proof}

\begin{theoremrep}\label{th:revToEvcExp}
Let $\Cmc$ be a \decomp binary class of models. 
If $\Cmc$ is revision-compatible then 
$\Cmc^+$ is reception-compatible and 
$\Cmc^-$ is eviction-compatible.
%$\Cmc^+$ is reception-compatible, and $\Cmc^-$ is eviction-compatible, for $ C^+ = \{ \Mplus \in \mUni \mid (\Mplus, \Mminus)\in \Cmc\}$ and     $C^- = \{ \Mminus \in \mUni \mid (\Mplus, \Mminus)\in \Cmc\}$
% \begin{itemize}[leftmargin=*]
	%     \item $\Cmc^+ = \{ \Mplus \in \mUni \mid (\Mplus, \Mminus)\in \Cmc\}$ is reception-compatible,
	%     \item $\Cmc^- = \{ \Mminus \in \mUni \mid (\Mplus, \Mminus)\in \Cmc\}$ is eviction-compatible.
	% \end{itemize}
	\end{theoremrep}
	\begin{proof}
Let 
$\Cmc$ be revision-compatible. 
Thus, 
there is a revision operator on $\Cmc$ satisfying \emph{\prudence, \vacuousexpansion} and \emph{\vacuousremoval}. 
To show eviction and reception-compatibility, it suffices to show that there exists 
\begin{description}
	\item[] (a) a reception operator on $\Cmc^+$ and 
	\item (b) an eviction operator $\mCon$  on $\Cmc^-$. 
\end{description}

We will construct such eviction and reception operators from a variation of $\rev$. 
To guarantee the postulate of \uniformity, let  
$\sel: \powerset(\powerset(\mUni)) \to \powerset(\mUni)$ be a function that, given a set $X$ of finitely representable sets, chooses exactly one of them,
that is 
$\sel(X) \in X$, if $X \neq \emptyset$. 
We define a second revision operator $\rev'$ from it, where 
$\modelsof{\rev'(\baseb, \Mplus, \Mminus)} = \mSet$ such that 
\begin{enumerate}[label=(\roman*)]
	\item if $\Mplus \subseteq \modelsof{\baseb}$, then 
	$$\mSet = \sel\big( \mathrm{min}_{\preceq_{\modelsof{\Bmc}}}(\chi(\M^+ , \M^- \cup (\mUni\setminus \modelsof{\baseb})) ) \big) $$
	
	\item if $\Mplus \not \subseteq \modelsof{\baseb}$, but $\Mminus \cap \modelsof{\baseb} = \emptyset$, then 
	$$ \mSet =\sel\big(  \mathrm{min}_{\preceq_{\modelsof{\Bmc}}}(\chi(\logsys,\M^+ \cup \modelsof{\baseb}, \M^-)) \big)$$
	
	\item otherwise, 
	$\mSet = \modelsof{\revline}$
\end{enumerate}

Note that $\rev'$ satisfies all three conditions from the symmetric-differential definition, so $\rev'$ is a symmetric-differential function.
Therefore, it satisfies all rationality postulates of revision. 
Let 
\begin{align*}
	\mExp(\baseb, \Mplus) &= \rev'(\baseb, \Mplus, \emptyset)\\
	\mCon(\baseb, \Mplus) &= \rev'(\baseb, \emptyset, \Mminus)        
\end{align*}
We will show that they are respectively reception and eviction operators. 
\begin{enumerate}
	\item (a) $\mExp$ is a reception operator. We show it satisfies all postulates from reception. 
	From \cref{lem:revtorcp}, $\mExp$ satisfies \success, persistence, and finite-temperance. 
	We proceed to prove  \uniformity. 
	Let us suppose that 
	\begin{align}
		\FRsups(\modelsof{\baseb} \cup \mSet,  \logsys)
		=\FRsups(\modelsof{\baseb'} \cup \mSet',  \logsys). \label{eq:hypBBiprime}    
	\end{align}
	
	By definition, $\mExp(\baseb, \mSet) = \rev'(\baseb, \mSet, \emptyset)$. 
	Either 
	(I) $\mSet \subseteq \modelsof{\baseb}$ or 
	(II) $\mSet \not \subseteq \modelsof{\baseb}$. 
	\begin{description}
		\item[] Case (I), $\modelsof{\baseb} \cup \mSet = \modelsof{\baseb}$ which is finitely representable. 
		Therefore, 
		\begin{align}
			\FRsups(\modelsof{\baseb} \cup \mSet,  \logsys) = \{ \modelsof{\baseb} \cup \mSet \}. \label{eq:bprimeeqFrB}    
		\end{align}

		So, 
		$$\mmin_{\preceq_{\modelsof{\baseb}}}(\dualcon(\mSet, \mUni \setminus \modelsof{\baseb}) = \{\modelsof{\baseb}\cup\mSet\}.$$ 
		Thus, from  the construction of $\rev'$, condition (i), 
		$$\modelsof{\mExp(\baseb, \mSet)} = \modelsof{\baseb}\cup\mSet$$
		From \eqref{eq:hypBBiprime} and \eqref{eq:bprimeeqFrB},
		$$ 
		\FRsups(\modelsof{\baseb'} \cup \mSet',  \logsys) = \{ \modelsof{\baseb} \cup \mSet \}
		$$
		From persistence, success and finite-temperance, 
		$$\modelsof{\mExp(\baseb', \mSet')} \in \FRsups(\modelsof{\baseb'} \cup \mSet',  \logsys).$$ 
		So, 
		$\modelsof{\mExp(\baseb', \mSet')} = \modelsof{\baseb} \cup \mSet$, 
		which implies
		$$\modelsof{\mExp(\baseb, \mSet)} = \modelsof{\mExp(\baseb', \mSet')}$$

		\item Case (II), $\mSet \not \subseteq \modelsof{\baseb}$. 
		Let 
		\begin{align}
			X = \modelsof{\mExp(\baseb, \mSet)}. \label{eq:xRec}   
		\end{align}
		
		So, from the construction of $\rev'$, item (ii),
		\begin{align}
			X = \sel\big(\mmin_{\preceq_{\modelsof{\baseb}}}(\dualcon(\mSet \cup \modelsof{\baseb}, \emptyset)) \big) \label{eq:xSel}        
		\end{align}

		From \cref{lem:minSymMinSups},
		
		$$\FRsups(\modelsof{\baseb} \cup \mSet) = \mathrm{min}_{\preceq_{\modelsof{\baseb}}}(\dualcon(\mSet \cup \modelsof{\baseb}, \emptyset)).$$
		
		This implies from \cref{eq:hypBBiprime} that 
		\begin{align}
			\FRsups(\modelsof{\baseb'} \cup \mSet') = \mathrm{min}_{\preceq_{\modelsof{\baseb}}}(\dualcon(\mSet \cup \modelsof{\baseb}, \emptyset)). 
			\label{eq:bprimeMprime}
		\end{align}
		
		Either $\mSet' \subseteq \modelsof{\baseb'}$ or not. 
		In the case that $\mSet' \subseteq 
		\modelsof{\baseb'}$, we get that 
		$\FRsups(\modelsof{\baseb'} \cup \mSet')$ is a singleton set $\{Y\}$. 
		From success, persistence and finite-retainment, it follows that 
		$\modelsof{\mExp(\baseb', \mSet')} = Y$. 
		This jointly with   \eqref{eq:bprimeMprime}, \eqref{eq:xSel} and \eqref{eq:xRec}, implies that $X = Y$. 
		Therefore, 
		$\modelsof{\mCon(\baseb, \mSet)} = \modelsof{(\mCon(\baseb, \mSet'))}$.
		
		In the case that $\mSet' \not \subseteq 
		\modelsof{\baseb'}$, we get from definition of $\rev'$, item (ii), that 
		$\modelsof{\mExp(\baseb', \mSet')} = \sel\big(  \mathrm{min}_{\preceq_{\modelsof{\Bmc}}}(\chi(\logsys,\M^+ \cup \modelsof{\baseb}, \M^-)) \big)$. 
		From \eqref{eq:xSel} and \eqref{eq:xRec} we have that 
		$\modelsof{\mExp(\baseb', \mSet')}$ which implies
		$\modelsof{\mCon(\baseb, \mSet)} = \modelsof{(\mCon(\baseb, \mSet'))}$.
		
	\end{description}

	\item (b) success follows from success from $\rev'$. 
	For inclusion, by definition $\mCon(\baseb, \Mminus) = \rev'(\baseb, \emptyset, \Mminus)$ which implies from the construction of $\rev'$, item (i), that 
	$\mCon(\baseb, \Mminus) \in \dualcon(\emptyset, \Mminus \cup( \mUni \setminus \modelsof{\baseb})$. 
	Which implies that no model can be gained. So, inclusion is satisfied. 
	\begin{description}
		\item[](\textbf{finite-retainment}) Let $\modelsof{\mCon(\baseb,\Mminus)} \subset \mSet \subseteq \modelsof{\baseb} \setminus \Mminus$. We will show that $\mSet$ is not finitely representable. 
		Let $$ X^- = \modelsof{\baseb} \setminus \modelsof{\mSet}.$$
		
		It follows from the hypothesis that  
		\begin{align*}
			&\mSet \subseteq \modelsof{\baseb}, \mbox{ and } \\
			&(\Mminus \cap \modelsof{\baseb}) \cap \mSet = \emptyset.
		\end{align*}
		Therefore, 
		\begin{align}
			\modelsof{\baseb} \setminus X^- &= \mSet \label{eq:dbcompl} \\
			\Mminus \cap \modelsof{\baseb} &\subseteq \modelsof{\baseb} \setminus \modelsof{\mSet} \nonumber\\
			\Mminus \cap \modelsof{\baseb} &\subseteq X^-\label{eq:p1_cd}
		\end{align}
		It follows from the hypothesis, 
		$$\modelsof{\mCon(\baseb,\Mminus)} \subset \mSet \subseteq \modelsof{\baseb}, $$
		which implies 
		\begin{align}
			\modelsof{\baseb} \setminus \modelsof{\mSet} &\subseteq \modelsof{\baseb} \setminus \modelsof{\mCon(\baseb, \Mminus)} \nonumber\\
			X^- &\subseteq \modelsof{\baseb} \setminus \modelsof{\mCon(\baseb, \Mminus)} \label{eq:p2_cd}
		\end{align}
		
		Note that 
		$$ \ $$
		
		By the definition of $\mCon$, 
		$\mCon(\baseb, \Mminus) = \rev'(\baseb, \emptyset, \Mminus)$. 
		So, $\Mplus = \emptyset$. 
		Let  $X^+ = \emptyset$. 
		So, 
		$$(\modelsof{\baseb} \setminus X^-) \cup X^+ = (\modelsof{\baseb} \setminus X^-), $$
		which from \eqref{eq:dbcompl} implies
		$$(\modelsof{\baseb} \setminus X^-) \cup X^+ = \mSet.$$
		Therefore, as $\modelsof{\mCon(\baseb,\Mminus)} \subset \mSet$, we get that 
		$\mSet \neq \modelsof{\mCon(\baseb,\Mminus)}$.
		So, 
		\begin{align}
			(\modelsof{\baseb} \setminus X^-) \cup X^+ \neq \modelsof{\mCon(\baseb,\Mminus)}\label{eq:cond3_pru}
		\end{align}
		From \eqref{eq:cpat1c1} and \eqref{eq:p2_cd}, condition (1) of \prudence is satisfied. 
		As $\Mplus = X^+ = \emptyset$, condition (2) of \prudence is also satisfied. 
		While condition (3) of \prudence follows from \eqref{eq:cond3_pru}.
		Therefore, it follows from \prudence that 
		$\mSet$ is not finitely representable.
		
		\item \uniformity . The proof is analogous to the proof of \uniformity for case (a). 
		The eviction reduces to case (i) of $\rev'$. The selection function guarantees that the eviction result coincides whenever the closest subsets not containing $\Mminus$ also coincide. So uniformity is satisfied. 
		%$\dualcon(\baseb, \M)$$$
		% Whenever a result were to be in $\FRsups$, in the proof for eviction it will be in $\FRsubs$  and the results will coincide.     
	\end{description}
	
\end{enumerate}

\end{proof}

\end{toappendix}

% As \cref{ex} illustrates, 

% In the special case that every set of interpretations is finitely representation, then trivially symmetric-differential revisions functions and then 

% The symmetric-differential operators capture precisely the class of all and only revision operators satisfying all rationality postulates. 
%not only satisfy all rationality postulates of revision, but also 
%opero actually, operations satisfies one further not only satisfies the rationaity po

\begin{toappendix}

\begin{proposition}\label{lem:from_evc_to_rev}
Let $\Cmc$ be a \decomp revision-compatible class of models.  
If $\Evc$ is an eviction operator on $\Cmc^-$, then  there is a revision operation $\rev$ on $\Cmc$ satisfying all postulates such that 
$\Evc(\baseb, \M) = \rev(\baseb,\emptyset, \Mminus)$.
\end{proposition}
\begin{proof}
As $\Cmc$ is revision compatible, there is a revision operation  $\rev'$. Let us fix such an operation.  
We define the following operation
\begin{align*}
	\rev(\baseb, \Mplus, \Mminus) &=\left\{ 
	\begin{array}{cc}
		\Evc(\baseb, \Mminus) &  \mbox{ if } \Mplus = \emptyset \\
		\rev'(\baseb, \Mplus, \Mminus)  & \mbox{otherwise}
	\end{array}
	\right.
\end{align*}

By definition, $\Evc(\baseb, \Mminus) = \rev(\baseb, \emptyset, \Mminus)$. 
To conclude the proof, we have to show that $\rev$ is a revision function. 
For $\Mplus \neq \emptyset$, all four postulates are trivially satisfied, as in such a case $\rev(\baseb, \Mplus, \Mminus) = \rev'(\baseb, \Mplus, \Mminus)$, which by hypothesis satisfies all four postulates. 
So we proceed with the case that $M^+ = \emptyset$. 
\begin{itemize}
	\item \textbf{success}: follows from the success of $\rev'$ and $\Evc$.
	
	\item \vacuousexpansion.
	Let $\Mplus \subseteq \modelsof{\baseb}$. If $\Mplus \neq \emptyset$, then \vacuousexpansion follows from \vacuousexpansion of $\rev'$. 
	If $\Mplus = \emptyset$, then 
	$\revline = \Evc(\baseb, \Mminus)$, which from inclusion implies 
	$\modelsof{\Evc(\baseb,\Mminus)} \subseteq \modelsof{\baseb}$. Therefore, 
	$\modelsof{\revline} \subseteq \modelsof{\baseb}$. 
	
	\item \vacuousremoval: 
	Let $\Mminus \cap \modelsof{\baseb} = \emptyset$. 
	If $\Mplus \neq \emptyset$, then \vacuousremoval follows from \vacuousremoval of $\rev'$. 
	If $\Mplus = \emptyset$, 
	then $\revline = \Evc(\baseb, \Mminus)$. Thus, as $\Mminus \cap \modelsof{\baseb} = \emptyset$, we have that 
	$\Evc(\baseb, \Mminus) = \modelsof{\baseb}$, 
	therefore $\modelsof{\revline} = \modelsof{\baseb}$. 
	So, $\modelsof{\baseb} \subseteq \modelsof{\revline}$.
	
	\item \prudence:  
	Let $X^-$ and $X^+$ be sets that satisfy conditions (1) to (3) of finite temperance. 
	We will show that 
	$(\modelsof{\baseb}\setminus X^-) \cup X^+ \not \in \FRsets_{\logsys}$. 
	From inclusion, 
	\begin{align}
		\modelsof{\Evc(\baseb, \Mminus)} \subseteq \modelsof{\baseb}, \label{eq:mode_evc_min_baseb}
	\end{align}        
	and from condition (1) $X^- \subseteq \modelsof{\baseb}$ and $\modelsof{\Evc(\baseb, \Mminus)} \cap X^- = \emptyset$. 
	Therefore, 
	\begin{align}
		\modelsof{\Evc(\baseb, \Mminus)}  \subseteq \modelsof{\baseb} \setminus X^- \label{eq:evc_minus_emp}
	\end{align}
	
	By definition, $\rev(\baseb, \emptyset, \Mminus) = \Evc(\baseb, \Mminus)$, which 
	from \eqref{eq:evc_minus_emp} that 
	\begin{align}
		\modelsof{\rev(\baseb, \Mminus)}  \subseteq \modelsof{\baseb} \setminus X^- \label{eq:evc_minus_emp_rev}
	\end{align}
	and from        
	\eqref{eq:mode_evc_min_baseb},  that 
	\begin{align*}
		\modelsof{\rev(\baseb, \emptyset, \Mminus)} \subseteq \modelsof{\baseb}. %\label{eq:rev_subs_bases_b}
	\end{align*}
	%which jointly with \eqref{eq:evc_minus_emp}  implies that 
	Therefore, 
	\begin{align*}
		\modelsof{\rev(\baseb, \Mminus)} \setminus \modelsof{\baseb} = \emptyset
	\end{align*}        
	
	% \begin{align}
		%     \modelsof{\Evc(\baseb, \Mminus)} & \subseteq \modelsof{\baseb}\\
		%     \modelsof{\rev(\baseb, \emptyset, \Mminus) }\setminus \modelsof{\baseb} = \emptyset
		% \end{align}
	which from condition (2) of finite temperance-retainment implies that 
	$$ X^+ = \emptyset . $$ 
	Thus, from  condition (3) of finite temperance-retainment, it follows that
	$ \modelsof{\baseb} \setminus X^- \neq \modelsof{\rev(\baseb, \emptyset, \Mminus)},$
	which from \eqref{eq:evc_minus_emp_rev} implies that 
	\begin{align}
		\modelsof{\rev(\baseb, \emptyset, \Mminus)} \subset \modelsof{\baseb} \setminus X^-. \label{eq:cond_finrep_retain_1}    
	\end{align}

	From condition (1) of finite temperance-retainment, $\modelsof{\baseb}\cap\Mminus \subseteq X^- \subseteq \modelsof{\baseb}$. 
	Therefore, 
	$$ \modelsof{\baseb} \setminus X^- \subseteq \modelsof{\baseb} \setminus (\modelsof{\baseb} \cap \Mminus)$$
	Thus, as $\modelsof{\baseb} \setminus (\modelsof{\baseb} \cap \Mminus) = \modelsof{\baseb} \setminus \Mminus$, we get
	\begin{align}
		\modelsof{\baseb} \setminus X^- \subseteq \modelsof{\baseb} \setminus \Mminus. \label{eq:cond_finrep_retain_2}    
	\end{align}
	
	Thus, from \eqref{eq:cond_finrep_retain_1} and \eqref{eq:cond_finrep_retain_2}
	$$ \modelsof{\rev(\baseb, \emptyset, \Mminus)} \subset \modelsof{\baseb} \setminus X^- \subseteq \modelsof{\baseb} \setminus \Mminus $$
	By definition, $\rev(\baseb, \emptyset, \Mminus) = \Evc(\baseb, \Mminus)$, which implies that 
	$$ \modelsof{\Evc(\baseb, \Mminus)} \subset \modelsof{\baseb} \setminus X^-1 \subseteq \modelsof{\baseb} \setminus \Mminus $$
	Thus, from finite-retainment, we get 
	$\modelsof{\baseb} \setminus X^- \not\in \FRsets_{\logsys}$. 
	
	Thus,as $X^+ =\emptyset$, we get 
	$(\modelsof{\baseb} \setminus X^-) \cup X^+ \not\in \FRsets$. 
	%So, fin ... is satisfied. 
	
\end{itemize}

\end{proof}

\begin{proposition}\label{prop:rcp_to_rev}
Let $\Cmc$ be a \decomp revision-compatible class of models. For If $\Rcp$ is a reception operator on $\Cmc^+$, then  there is a revision operation $\rev$ on $\Cmc$ satisfying all postulates such that 
$\Rcp(\baseb, \M) = \rev(\baseb, \M, \emptyset)$. 
\end{proposition}
\begin{proof}
As $\Cmc$ is revision-compatible, there is a revision operation  $\rev'$. Let us fix such an operation.  
We define the following operation
\begin{align*}
	\rev(\baseb, \Mplus, \Mminus) &=\left\{ 
	\begin{array}{cc}
		\Rcp(\baseb, \Mplus) &  \mbox{ if } \Mminus = \emptyset \\
		\rev'(\baseb, \Mplus, \Mminus)  & \mbox{otherwise}
	\end{array}
	\right.
\end{align*}

By definition, $\Rcp(\baseb, \Mplus) = \rev(\baseb, \Mplus, \emptyset)$. 
To conclude the proof, we have to show that $\rev$ is a revision function. 
For $\Mminus \neq \emptyset$, all four postulates are trivially satisfied, as in such a case $\rev(\baseb, \Mplus, \Mminus) = \rev'(\baseb, \Mplus, \Mminus)$, which by hypothesis satisfies all four postulates. 
So we proceed with the case that $M^- = \emptyset$. 
\begin{itemize}
	\item \textbf{success}: trivially $\Mminus \cap \modelsof{\rev(\baseb,\Mplus, \Mminus)} = \emptyset$. From the success of reception, we have that $\Mplus \subset \modelsof \Rcp(\baseb, \Mplus)$ which implies that 
	$\Mplus \subset \modelsof \rev(\baseb, \Mplus, \emptyset)$. 
	
	\item \vacuousexpansion: let $\Mplus \subseteq \modelsof{\baseb}$. Thus, from vacuity of reception, \\
	$\modelsof{\Rcp(\baseb, \Mplus)} = \modelsof{\baseb}$. Therefore, $\rev(\baseb, \Mplus, \emptyset) = \modelsof{\baseb}$, which implies $\rev(\baseb, \Mplus, \emptyset) \subseteq \modelsof{\baseb}$ 
	
	\item \vacuousremoval: From persistence, $\modelsof{\baseb} \subset \modelsof{\Rcp(\baseb, \Mplus)}$. 
	Therefore, $\modelsof{\baseb} \subseteq \rev(\baseb, \Mplus, \emptyset)$. 
	
	\item \prudence: 
	Let us suppose that finite temperance is not satisfied for contradiction purposes. 
	Thus, there are $X^+$ and $X^-$ satisfying conditions 1 to 3 of finite-temperance, but 
	\begin{align}
		(\modelsof{\baseb} \setminus X^-) \cup X^+ \in \FRsets(\logsys). \label{eq:xmin_emp_Fr}
	\end{align}

	From persistence,  
	$\modelsof{\baseb} \subseteq \Rcp(\baseb, \Mplus)$ which implies that 
	\begin{align}
		\modelsof{\baseb} \subseteq \rev(\baseb, \Mplus, \emptyset)    \label{eq:pers_model_subs_rev}
	\end{align}
	From condition 1, we have that $X^- \subseteq \modelsof{\baseb} \setminus\modelsof{\rev(\baseb, \Mplus, \emptyset)}$. 
	Therefore, $X^- = \emptyset$. 
	Thus, 
	$$ (\modelsof{\baseb} \setminus X^-) \cup X^+ = \modelsof{\baseb} \cup X^+$$
	Moreover, as $X^- = \emptyset$, it follows from \eqref{eq:xmin_emp_Fr},
	\begin{align}
		\modelsof{\baseb} \cup X^+ \in \FRsets(\logsys) \label{eq:fr_in_xmin}
	\end{align}

	From condition 2, 
	$X^+ \subseteq \modelsof{\rev(\baseb, \Mplus, \emptyset)}$. 
	Thus, 
	\begin{align}
		X^+ \cup \modelsof{\baseb} & \subseteq \modelsof{\rev(\baseb, \Mplus, \emptyset)} \cup \modelsof{\baseb, \Mplus, \emptyset} \nonumber \\
		& \subseteq \modelsof{\rev(\baseb, \Mplus, \emptyset)}\tag{from \eqref{eq:pers_model_subs_rev}}
	\end{align}
	
	From condition 3, $X^+ \cup \modelsof{\baseb} \neq \modelsof{\rev(\baseb, \Mplus, \emptyset)}$, which from above implies 
	$X^+ \cup \modelsof{\baseb} \subset \modelsof{\rev(\baseb, \Mplus, \emptyset)}$. 
	Thus, 
	$$ X^+ \cup \modelsof{\baseb} \subset \modelsof{\Rcp(\baseb, \Mplus}$$
	From condition 2, $\Mplus \setminus \modelsof{\baseb} \subseteq X^+$. 
	Thus, 
	$$\modelsof{\baseb} \cup \Mplus \subseteq X^+ \cup \modelsof{\baseb} \subset \modelsof{\Rcp(\baseb, \Mplus)}.$$
	Therefore, from finite temperance of reception, we get 
	$X^+ \cup \modelsof{\baseb} \not \in \FRsets(\logsys)$ which contradicts \eqref{eq:fr_in_xmin}.
	
\end{itemize}
\end{proof}

\end{toappendix}

\begin{toappendix}

\begin{theorem}\label{thm:posrev_posrev_two}
Let $\Cmc$ be a \decomp revision-compatible class of models,  

\begin{enumerate}
	\item if  $\mCon$ is a eviction operator on $\Cmc^-$, then there is a revision operator $\rev$ on $\Cmc$ satisfying all rationality postulates such that 
	$ \mCon(\baseb, \Mminus) =  \rev(\baseb, \emptyset, \Mminus) $,
	\item if  $\mExp$ is a reception operator on $\Cmc^+$, then there is a revision operator $\rev$ on $\Cmc$ satisfying all rationality postulates such that 
	$ \mExp(\baseb, \Mplus) =  \rev(\baseb, \Mplus, \emptyset) $. 
\end{enumerate}

% Let $\rev$ be a revision operator on a binary class of models $\Cmc$. 
% If $\rev$ satisfies \emph{\vacuousexpansion, \vacuousremoval} and \prudence, then 
% there are eviction and recepto
\end{theorem}
\begin{proof}
Item 1 follows from \cref{lem:from_evc_to_rev}, while item 2 follows from \cref{prop:rcp_to_rev}
\end{proof}

\end{toappendix}

\begin{theoremrep}\label{thm:posrev}
Let $\Cmc$ be a \decomp revision-compatible binary class of models.
%, then 
The following hold for \Cmc.   
\begin{enumerate}
\item 
$\Cmc^+$ is reception-compatible and 
$\Cmc^-$ is eviction-compatible. 
\item %If $\Cmc$ is revision-compatible 
If  $\mCon$ is a eviction operator on $\Cmc^-$, then there is a revision operator $\rev$ on $\Cmc$ satisfying all rationality postulates such that 
$ \mCon(\baseb, \Mminus) =  \rev(\baseb, \emptyset, \Mminus) $.
\item If  $\mExp$ is a reception operator on $\Cmc^+$, then there is a revision operator $\rev$ on $\Cmc$ satisfying all rationality postulates such that 
$ \mExp(\baseb, \Mplus) =  \rev(\baseb, \Mplus, \emptyset) $. 
\end{enumerate}

\end{theoremrep}
\begin{proof}
Follows from \cref{thm:posrev_posrev_two} and \cref{th:revToEvcExp}.

\end{proof}

From \cref{thm:posrev}, eviction and reception are special cases of revision, whereas revision can only be performed in classes of models compatible with both reception and eviction. 
This connection between revision with eviction and reception allows to translate (in)compatibility results from eviction and reception to revision, as we see 
in \cref{sec:revdl}.
\section{Revision: DL Concepts}\label{sec:revdl}
Here, we briefly consider the revision of DL concepts. It follows from \cref{thm:posrev} and the results in \cref{results} that neither $\logsys(\EL_{\bot \text{concepts}})$ nor $\logsys(\ALC_{\text{concepts}})$ are, in general, revision-compatible. We establish that these satisfaction systems are also not revision-compatible when we restrict to 
%the binary class of 
finite tree-shaped pointed interpretations.
%, with or without further restricting the sets of models or the signature to be finite. % focusing in \ALC.
\begin{theoremrep}\label{thm:alcposrevision}
	$\logsys(\EL_{\bot \text{concepts}})$ and $\logsys(\ALC_{\text{concepts}})$ are not revision-compatible in the binary class of finite tree-shaped pointed interpretations. This also holds if we restrict to finite sets of models and if we restrict to  a finite signature.  
\end{theoremrep}
\begin{proof}
	Consider the case in which $(\Imc,d
	)$ and $(\Jmc,e
	)$ are two different but bisimilar finite tree-shaped pointed interpretations over a finite signature.  
	Since both \EL and \ALC are invariant under bisimilar pointed interpretations, there is no 
	$Y \in \FRsets(\logsys(\EL_{\bot \text{concepts}}))$ and no $Y \in \FRsets(\logsys(\ALC_{ \text{concepts}}))$ such that
	$\{(\Imc,d
	) \}\subseteq Y \mbox{ and } \{(\Jmc,e
	) \} \cap Y = \emptyset$.
	%    $\chi(\logsys,\M^+, \M^-) = \{ Y \in \FRsets(\logsys) \mid  \M^+ \subseteq Y \mbox{ and } \M^- \cap Y = \emptyset \} $
	%   So,
	%    $\min_{\preceq_{\modelsof{\baseb}}}(\chi(\logsys(\ALC_{\text{concepts}}), \{(\Imc,d
	%) \}, \{(\Jmc,e
	%) \})) = \emptyset$. 
	Since we cannot satisfy \textbf{success}, this means that $\logsys(\EL_{\text{concepts}})$ and $\logsys(\ALC_{\text{concepts}})$ are not  revision-compatible in the class of finite tree-shaped pointed interpretations for finite sets of models and finite signature (\cref{def:revisionmod}).  
	%\todo{add ref to def of revision, relevant proposition }    
\end{proof}
So we restrict the binary class of models that we consider even further.
We consider the binary class of finite tree-shaped pointed interpretations for finite sets of models \emph{union their closure under bisimulation} over a finite signature. We argue that this class is revision-compatible.
We say that two sets of pointed interpretations are \emph{bisimulation disjoint} iff
there is no pointed interpretation in one of the sets that is bisimilar to a pointed interpretation in the other set. 

\begin{theoremrep}\label{thm:alcposrevision}
	$\logsys(\ALC_{\text{concepts}})$ is    revision-compatible in the binary class  of sets of pointed interpretations which are the closure under bisimulation of finite sets of finite tree-shaped pointed interpretations  over a (unique) finite signature.  
\end{theoremrep}
\begin{proof}
	Let $\{(\Imc_1,d_1),\ldots,$ $(\Imc_n,d_n)\}$ 
	and $\{(\Jmc_1,e_1),\ldots,$ $(\Jmc_m,e_m)\}$ be bisimulation-disjoint finite sets of finite tree-shaped pointed interpretations over a finite signature $\Sigma$ and let
	$\mSet^+$ and $\mSet^-$ be their closure under bisimulation. Let $D_i$ be the $\EL_\bot$ concept with $\Imc_{D_i}$ being isomorphic to $\Imc_i$ and let $E_j$ be the $\EL_\bot$ concept with $\Jmc_{E_j}$ being isomorphic to $\Jmc_j$.
	Let $C$ be an \ALC concept over $\Sigma$.
	%    We need to show that
	%\begin{enumerate}
	%       \item $\modelsof{ (\bigsqcup^n_{i=1}  (D_i)^\dagger \sqcup C)\sqcap \bigsqcap^m_{j=1} \neg (E_j)^\dagger}\in \mathrm{min}_{\preceq_{\modelsof{C}}}(\chi(\logsys(\ALC_{\text{concepts}}),\M^+ , $ $\M^- $ $\cup (\mUni\setminus \modelsof{C})) ) $, if $\Mplus \subseteq \modelsof{C}$;
	%
	%   \item  $\modelsof{ (\bigsqcup^n_{i=1}   (D_i)^\dagger \sqcup C)\sqcap \bigsqcap^m_{j=1} \neg (E_j)^\dagger}\in     \mathrm{min}_{\preceq_{\modelsof{C}}}(\chi(\logsys(\ALC_{\text{concepts}}),\M^+ $ $\cup \ $ \\  $ $ $\modelsof{C}, \M^-))$, if  $\Mplus \not \subseteq \modelsof{C}$, but  $\Mminus \cap \modelsof{C} = \emptyset$;
	%
	%  \item 
	%   $\modelsof{ (\bigsqcup^n_{i=1}   (D_i)^\dagger \sqcup C)\sqcap \bigsqcap^m_{j=1} \neg (E_j)^\dagger}\in   \mathrm{min}_{\preceq_{\modelsof{C}}}(\chi(\logsys(\ALC_{\text{concepts}}),\M^+,  $ \\$ \M^-) )  $, if $\Mplus\not \subseteq \modelsof{C}$ and $\Mminus \cap \modelsof{C} \neq \emptyset$. 
	%\end{enumerate}
	%\[\modelsof{C\sqcap \bigsqcap^n_{i=1} \neg (D_i)^\dagger}\in \FRsubs(\modelsof{C}\setminus \mSet, \logsys),\] 
	%where $(D_i)^\dagger$ and $(E_j)^\dagger$ are as in \cref{def:translation}, for all $1\leq i\leq n$ and all $1\leq j\leq m$. 
	Since $(\Imc_1,d_1),\ldots,$ $(\Imc_n,d_n)$ 
	and $(\Jmc_1,e_1),\ldots,$ $(\Jmc_m,e_m)$ are bisimulation-disjoint finite sets of finite tree-shaped pointed interpretations over   $\Sigma$ and $\mSet^+$ and $\mSet^-$ are their closure under bisimulation,  
	$\mSet^+$ and $\mSet^-$ are disjoint.
	We   argue that
	%\[\modelsof{ (\bigsqcup^n_{i=1}   (D_i)^\dagger \sqcup C)\sqcap \bigsqcap^m_{j=1} \neg (E_j)^\dagger}= (\mSet^+\cup\modelsof{C})\setminus \Mminus=\mSet^+\cup(\modelsof{C}\setminus \Mminus).\] 
	\[\modelsof{ (\bigsqcup^n_{i=1}   (D_i)^\dagger \sqcup C)\sqcap \bigsqcap^m_{j=1} \neg (E_j)^\dagger}= (\mSet^+\cup\modelsof{C})\setminus \Mminus.\] 
	% (since $\mSet^+$ and $\mSet^-$ are disjoint).
	Indeed, by the semantics of \ALC,
	\[\modelsof{(\bigsqcup^n_{i=1}   (D_i)^\dagger \sqcup C)}=\modelsof{\bigsqcup^n_{i=1}   (D_i)^\dagger}\cup \modelsof{C}.\]
	By Lemma~\ref{lem:translationconcept} and the definition of $D_i$, we have that
	\[\modelsof{\bigsqcup^n_{i=1}   (D_i)^\dagger}=\bigcup^n_{i=1}\modelsof{   (D_i)^\dagger}=\mSet^+.\]
	Also, by Lemma~\ref{lem:translationconcept} and the definition of $E_j$, we have that
	\[\modelsof{\bigsqcap^m_{j=1}   \neg(E_j)^\dagger}=\bigcap^m_{j=1}(\mUni\setminus\modelsof{   (E_j)^\dagger})=\mUni\setminus(\bigcup^m_{j=1}\modelsof{   (E_j)^\dagger})=\mUni\setminus\mSet^-.\]
	Then, \[\modelsof{ (\bigsqcup^n_{i=1}   (D_i)^\dagger \sqcup C)\sqcap \bigsqcap^m_{j=1} \neg (E_j)^\dagger}= (\mSet^+\cup\modelsof{C})\cap(\mUni\setminus \Mminus)=(\mSet^+\cup\modelsof{C})\setminus \Mminus.\]
	%\todo{...}
	%This means that (1) if $\Mplus \subseteq \modelsof{C}$ then
	%\[\modelsof{ (\bigsqcup^n_{i=1}   (D_i)^\dagger \sqcup C)\sqcap \bigsqcap^m_{j=1} \neg (E_j)^\dagger}=(\mSet^+\cup\modelsof{C})\setminus \Mminus=\modelsof{C}\setminus \M^- \]\[=\modelsof{C}\setminus (\M^- \cup (\mUni\setminus \modelsof{C})).\]
	%So
	%\[\modelsof{ (\bigsqcup^n_{i=1}  (D_i)^\dagger \sqcup C)\sqcap \bigsqcap^m_{j=1} \neg (E_j)^\dagger}\in \]\[\mathrm{min}_{\preceq_{\modelsof{C}}}(\chi(\logsys(\ALC_{\text{concepts}}),\M^+ , \M^- \cup (\mUni\setminus \modelsof{C})) ) .\]
	%Also (2) if $\Mminus \cap \modelsof{C} = \emptyset$ then
	%\[\modelsof{ (\bigsqcup^n_{i=1}   (D_i)^\dagger \sqcup C)\sqcap \bigsqcap^m_{j=1} \neg (E_j)^\dagger}=(\mSet^+\cup\modelsof{C})\setminus \Mminus=\mSet^+\cup\modelsof{C}\] (since $\mSet^+$ and $\mSet^-$ are disjoint). 
	%\[\modelsof{ (\bigsqcup^n_{i=1}   (D_i)^\dagger \sqcup C)\sqcap \bigsqcap^m_{j=1} \neg (E_j)^\dagger}=(\mSet^+\cup\modelsof{C})\setminus \Mminus=\modelsof{C}\setminus \M^- \]\[=\modelsof{C}\setminus (\M^- \cup (\mUni\setminus \modelsof{C})).\]
	%So
	%\[\modelsof{ (\bigsqcup^n_{i=1}   (D_i)^\dagger \sqcup C)\sqcap \bigsqcap^m_{j=1} \neg (E_j)^\dagger}\in   \]\[   \mathrm{min}_{\preceq_{\modelsof{C}}}(\chi(\logsys(\ALC_{\text{concepts}}),\M^+ \cup \modelsof{C}, \M^-)).\]
	% Finally, we argue that Case (3) holds.
	%
	%\todo{...}
	%So \[\modelsof{ (\bigsqcup^n_{i=1}   (D_i)^\dagger \sqcup C)\sqcap \bigsqcap^m_{j=1} \neg (E_j)^\dagger}\in   \mathrm{min}_{\preceq_{\modelsof{C}}}(\chi(\logsys(\ALC_{\text{concepts}}),\M^+,   \M^-) ) .\] 
	The conditions of~\cref{cor:alcrev} are satisfied for this binary class, $\logsys(\ALC_{\text{concepts}})$ is revision-compatible.
\end{proof}

Regarding the case of $\EL_\bot$,  we do not have the same expressivity we have in \ALC for restricting the models that are satisfied by a concept. We leave it as an open question.
From the positive results for eviction and reception for certain classes of models for $\EL_\bot$ and \ALC in \cref{results}, 
%and \cref{thm:posrev}, 
we obtain the existence of the revision operators  in \cref{thm:posrev}.

\section{Conclusion}
We investigated eviction and reception for DL concepts, establishing various results considering different classes of models. We find classes of models
where we obtain eviction and reception compatibility for both \ALC and $\EL_\bot$ concepts. 
It turns out that the class of models where we obtain positive results for \ALC is much more restricted than the class for $\EL_\bot$. We also introduce the notion of model revision
%, which considers binary classes of models. We 
%
and relate various postulates with the revision operation. Revision cannot be seen as a mere combination of eviction and reception, which is evidenced by negative results for DL concepts.
As future work, it would be interesting to investigate model eviction, reception, and revision of DL concepts in a more practical setting, expand our work to more expressive DLs.
%, and study model revision for ontologies.

\section*{Acknowledgments}
Ozaki is supported by the Research Council of Norway, project (316022, 322480). This work was also supported by the Research Council of Norway, Integreat - Norwegian Centre for knowledge-driven machine learning (332645).
\bibliographystyle{apalike}
\bibliography{references}

\end{document}